\documentclass[letterpaper,11pt]{article}
\usepackage[letterpaper,margin=1in]{geometry}

\usepackage{booktabs}

\usepackage{microtype}
\usepackage{xcolor}
\usepackage{xspace}
\usepackage{amsmath}
\usepackage{amssymb}
\usepackage{amsthm}
\usepackage{enumitem}
\usepackage{textgreek}
\usepackage{graphicx}
\usepackage{framed}
\usepackage[small]{caption}
\usepackage{tikz-cd}
\usepackage[unicode]{hyperref}
\usepackage{doi}
\usepackage[capitalize, nameinlink]{cleveref}
\usepackage[linesnumbered,ruled,vlined,algo2e]{algorithm2e}
\usepackage{algorithm}
\usepackage{algpseudocode}
\Crefname{remark}{Remark}{Remarks}
\Crefname{observation}{Observation}{Observations}

\theoremstyle{plain}
\newtheorem{theorem}{Theorem}
\newtheorem{lemma}[theorem]{Lemma}

\newtheorem{corollary}[theorem]{Corollary}

\theoremstyle{definition}
\newtheorem{definition}[theorem]{Definition}

\newtheorem{observation}[theorem]{Observation}

\theoremstyle{remark}

\newcommand{\LOCAL}{\ensuremath{\mathsf{LOCAL}}\xspace}
\newcommand{\CONGEST}{\ensuremath{\mathsf{CONGEST}}\xspace}
\DeclareMathOperator*{\argmax}{arg\,max}

\newcommand{\myceil}[1]{\left \lceil #1 \right \rceil }
\newcommand{\myfloor}[1]{\left \lfloor #1 \right \rfloor }

\definecolor{darkgreen}{rgb}{0,0.5,0}
\definecolor{darkred}{rgb}{0.4,0,0}
\hypersetup{
	colorlinks=true,
	linkcolor=darkred,
	citecolor=darkgreen,
	filecolor=black,
	urlcolor=[rgb]{0,0.1,0.5},
	pdftitle={On the Average Complexity of Locally Checkable Problems on Trees},
	pdfauthor={TBD}
}

\newenvironment{myabstract}
{\list{}{\listparindent 1.5em%
		\itemindent    \listparindent
		\leftmargin    1cm
		\rightmargin   1cm
		\parsep        0pt}%
	\item\relax}
{\endlist}

\newenvironment{mycover}
{\list{}{\listparindent 0pt
		\itemindent    \listparindent
		\leftmargin    1cm
		\rightmargin   1cm
		\parsep        0pt}%
	\raggedright
	\item\relax}
{\endlist}

\newcommand{\myemail}[1]{\,$\cdot$\, {\small #1}}
\newcommand{\myaff}[1]{\,$\cdot$\, {\small #1}\par\smallskip}

\newcommand{\calA}{\mathcal{A}}

\newcommand{\calG}{\mathcal{G}}

\newcommand{\fA}{\mathcal{A}}

\DeclareMathOperator{\E}{\mathbb{E}}

\newcommand{\nodeavg}{\mathsf{AVG}_V}

\newcommand{\dist}{\operatorname{dist}}

\newcommand{\gs}[1]{{\color{purple} \textbf{Gustav}: #1}}

\newcounter{open}
\newtheorem{oq}[open]{Open Problem}

\graphicspath{{figs/}}

\renewcommand{\subparagraph}[1]{\medskip\noindent\textbf{#1}}

\crefname{algocf}{Alg.}{Algs.}
\Crefname{algocf}{Algorithm}{Algorithms}

\begin{document}

\setcounter{page}{0}

\renewcommand*{\thefootnote}{\fnsymbol{footnote}}

\begin{mycover}
    {\huge\bfseries On the Node-Averaged Complexity of Locally Checkable Problems on Trees\footnote{This work has been partially funded by the European Union - NextGenerationEU under the Italian Ministry of University and Research (MUR) National Innovation Ecosystem grant ECS00000041 - VITALITY – CUP: D13C21000430001, and by the German Research Foundation (DFG), Grant 491819048.}\par}
  \bigskip
  \bigskip
  \bigskip

  \textbf{Alkida Balliu}
  \myemail{alkida.balliu@gssi.it}
  \myaff{GSSI L'Aquila, Italy}
  \textbf{Sebastian Brandt}
  \myemail{brandt@cispa.de}
  \myaff{CISPA, Saarbr\"ucken, Germany}
  \textbf{Fabian Kuhn}
  \myemail{kuhn@cs.uni-freiburg.de}
  \myaff{University of Freiburg, Germany}
  \textbf{Dennis Olivetti}
  \myemail{dennis.olivetti@gssi.it}
  \myaff{GSSI L'Aquila, Italy}
  \textbf{Gustav Schmid}
  \myemail{schmidg@informatik.uni-freiburg.de}
  \myaff{University of Freiburg, Germany}

\end{mycover}
\bigskip

\renewcommand*{\thefootnote}{\arabic{footnote}}
\addtocounter{footnote}{-1}

\begin{myabstract}
Over the past decade, a long line of research has investigated the distributed complexity landscape of locally checkable labeling (LCL) problems on bounded-degree graphs, culminating in an almost-complete classification on general graphs and a complete classification on trees.
The latter states that, on bounded-degree trees, any LCL problem has deterministic \emph{worst-case} time complexity $O(1)$, $\Theta(\log^* n)$, $\Theta(\log n)$, or $\Theta(n^{1/k})$ for some positive integer $k$, and all of those complexity classes are nonempty. Moreover, randomness helps only for (some) problems with deterministic worst-case complexity $\Theta(\log n)$, and if randomness helps (asymptotically), then it helps exponentially.

In this work, we study how many distributed rounds are needed \emph{on average per node} in order to solve an LCL problem on trees. We obtain a partial classification of the deterministic \emph{node-averaged} complexity landscape for LCL problems. As our main result, we show that every problem with worst-case round complexity $O(\log n)$ has deterministic node-averaged complexity $O(\log^* n)$. \\
Then we show how using randomization we can speed this up and show that every problem with worst case round complexity $O(\log n)$ has randomized node-averaged complexity $O(1)$.\\
We further establish bounds on the node-averaged complexity of problems with worst-case complexity $\Theta(n^{1/k})$: we show that all these problems have node-averaged complexity $\widetilde{\Omega}(n^{1 / (2^k - 1)})$, and that this lower bound is tight for some problems. The lower bound holds even for the randomized case and the upper bound is a deterministic algorithm.\\
\end{myabstract}

\section{Introduction}
\label{sec:intro}

The family of locally checkable labeling (LCL) problems was introduced in the seminal work of Naor and Stockmeyer~\cite{NaorStockmeyer95} and since then, understanding the distributed complexity of computing LCLs has been at the core of the research on distributed graph algorithms. Roughly speaking, LCLs are labelings of the nodes or edges of a graph $G=(V,E)$ with labels from a finite alphabet such that some local, constant-radius condition holds at all the nodes. In the distributed context, $G$ represents a network and one typically assumes that the nodes of $G$ can communicate over the edges of $G$ in synchronous rounds. If this communication is unrestricted, this is known as the \LOCAL model of computation and if messages must consist of $O(\log n)$ bits (where $n$ is the number of nodes), it is known as the \CONGEST model. In our paper, we focus on the \LOCAL model and we therefore do not explicitly analyze the required message sizes of our algorithms. We however believe that all our algorithms can be made to work in the \CONGEST model with minor modifications.

Often LCL problems are studied in the context of bounded-degree graphs. In this case, LCLs include problems such as properly coloring the nodes of $G$ with $\Delta+1$ colors, where $\Delta$ is the maximum degree of $G$. Especially over the last decade, researchers have obtained a thorough understanding of the complexity landscape of distributed LCL problems in general bounded-degree graphs \cite{CP19timeHierarchy,CKP19exponential,FischerGhaffari17LLL,RG20NetDecomposition,BHKLOS18lclComplexity,BBOS20paddedLCL} and also in more special graph families such as in particular in bounded-degree trees~\cite{brandt21trees,BHOS19HomogeneousLCL,CKP19exponential,CP19timeHierarchy,BBOS18almostGlobal,chang20}. Most of this work focuses on the classic notion of worst-case complexity: If all nodes start a computation at time $0$ and communicate in synchronous rounds, how many rounds are needed until \emph{all nodes} have decided about their outputs. In some case however, the worst-case round complexity might be determined by a small number of nodes that require a lot of time to compute their outputs, while most of the nodes find their outputs much faster. Consider for example the simple randomized $(\Delta+1)$-coloring algorithm where in every step, every node picks a random available color and permanently keeps this color if there is no conflict. It is not hard to show that in every step, every uncolored node becomes colored with constant probability~\cite{Johansson99}. Hence, while we need $\Omega(\log n)$ steps (and thus also $\Omega(\log n)$ rounds) until all nodes are colored, for each individual node, the expected time to become colored is constant and consequently the time that nodes need on average to become colored is also constant w.h.p. In some contexts (e.g., when considering the energy cost of a distributed algorithm), this average completion time per node is more meaningful than the worst-case completion time and consequently, researchers have recently showed interest in determining the \emph{node-averaged} time complexity of distributed graph algorithms~\cite{Feuilloley17,BarenboimT19,ChatterjeeGP20,BalliuGKO22_average}. In the present paper, we continue this work and we study the \emph{node-averaged complexity of LCL problems in bounded-degree trees}. Before describing our contributions, we first briefly summarize some of the relevant previous work.

\subparagraph{Previous results on node-averaged complexity.} The first paper that explicitly considered the node-averaged complexity of distributed graph algorithms is by Feuilloley~\cite{Feuilloley17}. The paper mainly considers LCL problems on paths and cycles (i.e., on graphs of maximum degree $2$). It is known that on paths and cycles, when considering the worst-case complexity of LCL problems, randomization does not help and the only complexities that exist are $O(1)$, $\Theta(\log^* n)$, and $\Theta(n)$~\cite{NaorStockmeyer95,CKP19exponential,CP19timeHierarchy}. In \cite{Feuilloley17}, it is shown that for deterministic algorithms, the worst-case complexity and the node-averaged complexity of LCL problems on paths and cycles is the same. This for example implies that the classic $\Omega(\log^* n)$ lower bound of \cite{Linial92} for coloring cycles with a constant number of colors also applies to node-averaged complexity. While this is true for deterministic algorithms, it is also shown in \cite{Feuilloley17} that the randomized node-averaged complexity of $3$-coloring paths and cycles is constant. As sketched above and also explicitly proven in \cite{BarenboimT19}, the same is true for the more general problem of computing a $(\Delta+1)$-coloring in arbitrary graphs. While the results of \cite{Feuilloley17} imply results for general LCLs on paths and cycles, the additional work on node-averaged complexity focused on the complexity of specific graph problems, in particular on the complexity of well-studied classic problems such as computing a maximal independent set (MIS) or a vertex coloring of the given graph. Barenboim and Tzur~\cite{BarenboimT19} show that in graphs of small arboricity, some coloring problems have a deterministic node-averaged complexity that is significantly smaller than the corresponding worst-case complexity. For example, it is shown that if the arboricity is constant, an $O(k)$-vertex coloring can be computed in node-averaged complexity $O(\log^{(k)} n)$ for any fixed integer $k\geq 1$, where $O(\log^{(k)} n)$ denotes the $k$ times iterated logarithm of $n$. As one of the main results of \cite{BalliuGKO22_average}, it was shown that the MIS lower bound of \cite{KMW16} can be generalized to show that even with randomization, computing an MIS on general (unbounded degree) graphs requires node-averaged complexity $\Omega(\sqrt{\log n/\log\log n})$. Hence, while the problem of coloring with $(\Delta+1)$ colors and, as also shown in \cite{BalliuGKO22_average}, the problem of computing a $2$-ruling set have randomized algorithms with constant node-averaged complexity, the same is not true for the problem of computing an MIS.

\subparagraph{LCL complexity in bounded-degree trees.} One of the goals of this paper is to make a step beyond understanding individual problems and to start studying the landscape of possible node-averaged complexities of general LCL problems. We do this by studying LCL problems on bounded-degree trees, a graph family that we believe is relevant and that has recently been studied intensively from a worst-case complexity point of view (e.g., \cite{CKP19exponential,CP19timeHierarchy,ChangHLPU20_treeLLL,BHOS19HomogeneousLCL,FischerGhaffari17LLL,RG20NetDecomposition,brandt21trees}). In bounded-degree trees, we know that for deterministic algorithms, exactly the following worst-case complexities are possible: $O(1)$, $\Theta(\log^* n)$, $\Theta(\log n)$, and $\Theta(n^{1/k})$ for some fixed integer $k\geq 1$. It was shown in \cite{brandt21trees} (and earlier for a special subclass of LCLs in \cite{BHOS19HomogeneousLCL} and for paths in \cite{NaorStockmeyer95,CP19timeHierarchy}) that on bounded-degree trees, there are no deterministic or randomized asymptotically optimal algorithms with a time complexity in the range $\omega(1)$ to $o(\log^* n)$. Further, in \cite{CKP19exponential}, it was shown that even for general bounded-degree graphs, there are no deterministic LCL complexities in the range $\omega(\log^*n)$ to $o(\log n)$. Finally, it was shown in \cite{CP19timeHierarchy} that every LCL problem that requires $\omega(\log n)$ rounds on bounded-degree trees has a worst-case deterministic and randomized complexity of the form $\Theta(n^{1/k})$ for some fixed integer $k\geq 1$ (and all those complexities also exist). It is further known that randomization can only help for LCL problems with a deterministic complexity of $\Theta(\log n)$. Those problems have a randomized complexity of either $\Theta(\log n)$ or $\Theta(\log\log n)$ (and both cases exist)~\cite{CP19timeHierarchy,ChangHLPU20_treeLLL}. We discuss additional related work on the complexity landscape of LCLs in \Cref{apx:related-lcls}.

\subsection{Our Contributions}
\label{sec:contributions}

As our main result, we show that the $\Theta(\log n)$ complexity class vanishes when considering the node-averaged complexity of LCLs on bounded-degree trees. More formally, we prove the following theorem.

\begin{theorem}\label{thm:logstaraverage}
    Let $\Pi$ be an LCL problem for which there is an $O(\log n)$-round deterministic algorithm on bounded-degree trees. Then, $\Pi$ can be solved deterministically with node-averaged complexity $O(\log^* n)$ on bounded-degree trees. 
\end{theorem}

A standard example for an LCL problem that requires $\Theta(\log n)$ rounds deterministically is the problem of $3$-coloring a tree. So for 3-coloring \Cref{thm:logstaraverage} states that there is a deterministic distributed $3$-coloring algorithm, for bounded degree trees, with node-averaged complexity $O(\log^* n)$ rounds. Meaning that the average node terminates after $O(\log^*n)$ rounds. Note that for $3$-coloring trees deterministically, this is tight. As shown in \cite{Feuilloley17}, $3$-coloring has deterministic node-averaged complexity $\Omega(\log^* n)$ even on paths. Below, we will use the $3$-coloring problem as a simple example to illustrate some of the challenges in obtaining the above theorem, but first we state the rest of our results.

We show that using randomization gives us a significant advantage over the deterministic case and results in a speedup to constant node-averaged complexity.
\begin{theorem}\label{thm:RandomConst}
    Let $\Pi$ be an LCL problem for which there is an $O(\log n)$-round deterministic algorithm on bounded-degree trees. Then using randomization, $\Pi$ can be solved with node-averaged complexity $O(1)$ on bounded-degree trees. 
\end{theorem}

The main bottleneck of the algorithm from \cref{thm:logstaraverage} is to separate long paths into constant length subpaths. Using randomization we can achieve this in a way that the nodes only run for a constant number of rounds \emph{in expectation}. We show that if we are careful enough we can alter the algorithm to achieve a speedup to $O(1)$ node-averaged complexity. 

In addition to our results about the $O(\log n)$ worst case regime, we also investigate the node-averaged complexity of LCL problems that require polynomial time in the worst case (i.e., time $\Theta(n^{1/k})$ for some integer $k\geq 1$). We show that for such problems, also the node-averaged complexity is polynomial. However at least in some cases, it is possible to obtain a node-averaged complexity that is significantly below the worst-case complexity. In \cite{CP19timeHierarchy}, the hierarchical $2\frac{1}{2}$-coloring problem with parameter $k$ is defined as an example problem with worst-case complexity $\Theta(n^{1/k})$. We show that the node-averaged complexity of this LCL problem is significantly smaller.

\begin{theorem}\label{thm:twohalfcoloring}
    The deterministic node-averaged complexity of the hierarchical $2 \frac{1}{2}$-coloring problem with parameter $k$ is $O(n^{1/(2^k - 1)})$.
\end{theorem}

Finally, we show that for a problem with worst-case complexity $\Theta(n^{1/k})$, this is essentially the best possible node-averaged complexity. Meaning that we also prove that our algorithm for hierarchical $2 \frac{1}{2}$-coloring problems is optimal up to one $\log n$ factor.

\begin{theorem}\label{thm:polylowerbound}
    Let $\Pi$ be an LCL problem with (deterministic or randomized) worst-case complexity $\Omega(n^{1/k})$. Then, the randomized node-averaged complexity of $\Pi$ is $\Omega(n^{1/(2^k - 1)} / \log n)$.
\end{theorem}

Note that the algorithm of \Cref{thm:twohalfcoloring} is deterministic, and that the lower bound of \Cref{thm:polylowerbound} holds for randomized  algorithms as well.

\subsection{High-level Ideas and Challenges}

We next discuss some of the ideas that lead to the known results about solving LCL problems on bounded-degree trees and we highlight some of the challenges that one has to overcome and some of the ideas we use to prove \Cref{thm:logstaraverage,thm:twohalfcoloring,thm:polylowerbound}. 

\subparagraph{Rake-and-compress decomposition.} We start by sketching a generic algorithm that can be used to solve all LCL problems in bounded-degree trees. The generic algorithm can be used to obtain algorithms with an asymptotically optimal worst-case complexity for all problems with worst-case complexity $\Omega(\log n)$. As a first step, the algorithm uses a technique that is known as \emph{rake-and-compress}~\cite{MillerR85} to partition the nodes of a given tree $T=(V,E)$ into $O(\log n)$ layers such that each layer is either a rake layer that consists of a set of independent nodes or it is a compress layer that consists of a sufficiently separated set of paths. Every node in a rake layer has at most one neighbor in a higher layer, and in each path of a compress layer, the two nodes at the end have exactly one neighbor in a higher layer and the other nodes on the path have no neighbors in a higher layer.\footnote{The actual decomposition that we use is a bit more complicated and the formal definition (see \Cref{def:modi}) requires some additional details.} Such a decomposition can be computed in an iterative process that produces the layers in increasing order. Given some tree (or forest), a rake layer can be obtained by taking the set of all leaf nodes\footnote{When two degree-$1$ nodes are neigbors, one just takes one of the two nodes.} and a compress layer can be created by the paths (or more precicely by the inner part of the paths) induced by degree-$2$ nodes. It is not hard to show that when alternating rake and compress layers, this process completes after creating $O(\log n)$ layers~\cite{MillerR85}. 

\subparagraph{Applying the decomposition.} As an example of how to use rake-and-compress to solve an LCL problem, we look at the case of $3$-coloring the nodes of a tree $T$. Given a decomposition into rake and compress layers, this can be done in $O(\log n)$ rounds as follows. First, color each of the paths of the compress layers with $O(1)$ colors. This can be done in $O(\log^* n)$ rounds. Then, the $3$-coloring of $T$ is computed by starting at the highest layer of the decomposition. When processing a rake layer, each node can just be colored with a color different from its (at most one) neighbor in a higher layer. When processing a compress layer, we just have to $3$-color the paths of the layer such that each node at the end of a path picks a color that differs from the color of its neighbor in a higher layer. Given the initial $O(1)$-coloring of the path, this can be done in constant time for each path. The time to compute the coloring is therefore proportional to the number of layers and thus $O(\log n)$. The generic algorithm for solving more general LCL problems is more involved, but still similar at a high level. While creating the decomposition, for each node $v$, one can create a list of labels that can be assigned to $v$ such that the labeling of lower layer nodes that depend on $v$ can still be completed. The LCL problem needs to allow labelings that are flexible enough such that when having long paths of nodes that each can be the root of an arbitrary subtree, the nodes of the path can still be labeled efficiently (in constant time given an appropriate initial coloring of the path).

\subparagraph{Implementation with low node-averaged complexity.} The main challenge when trying to obtain an $o(\log n)$ node-averaged complexity is the following. The generic algorithm first computes the decomposition and it then computes the labeling by starting with the nodes in the highest layers. In the worst case, we therefore need $\Theta(\log n)$ rounds before even the label of a single node is determined. And moreover, most of the nodes are in the first few layers, which are labeled at the very end of the algorithm. In order to obtain a low node-averaged complexity, we therefore need to label most of the nodes already in the ``bottom-up'' phase when creating the rake and compress layers. For some problems, this is challenging: for example, in the $3$-coloring problem, if we ever obtain a node with $3$ neighbors of lower layers that have $3$ different colors, then we cannot complete the solution in any valid way. Hence, we have to label the nodes in such a way that the ``top-down'' phase is still able to extend the partial labeling to a valid labeling of all the nodes. To keep things simple, we here assume that the tree has diameter $O(\log n)$. In this case, it suffices to create rake layers and we do not need compress layers. We further only look at the problem of $3$-coloring the nodes of $T$. This problem is significantly easier to handle than general $O(\log n)$-worst case complexity LCL problems, the solution for $3$-coloring however already requires some of the ideas that we use in the general case.

Let us therefore assume that we have an $O(\log n)$-diameter tree $T$ with maximum degree $\Delta=O(1)$. If we decompose by using only rake layers, we obtain $O(\log n)$ layers, where each layer is an independent set of $T$ and every node has exactly one neighbor in a higher layer, except for the single node $u$ that is in the top layer. We refer to $u$ as the root node and for each node $v$, we refer to the single neighbor $w$ of $v$ in a higher layer as the parent of $v$. Note that when assigning a color to a node $v$ in the top-down phase, only $v$'s parent has already been assigned a color. To complete the top-down phase, it therefore suffices if every node $v$ can choose its color from an arbitrary subset $S_v$ of size $2$ of the colors. So if no nodes below $v$ have already decided on a color, this will always be possible. Hence, if we try to color some nodes already in the bottom-up phase, we have to make sure that all the uncolored nodes still have at least two available colors. This is guaranteed as long as every uncolored node has at most one colored neighbor.

When constructing the layering we therefore proceed as follows. We only color nodes that have already been assigned to some rake layer. Whenever we decide to color a node $v$ in the bottom-up phase, we also directly color the whole subtree of $v$.\footnote{After coloring the root of a subtree, the coloring of the subtree can be done in parallel while proceedings with the rest of the algorithm.} The high-level idea of the algorithm to achieve this is as follows. After each rake step, i.e., after each creation of a new layer, we check whether or not there are some nodes that can be colored. Consider the situation after the $t^{\mathit{th}}$ rake step, let $G^{(t)}$ be the set of nodes that have not been raked at that time (i.e., that have not been assigned to some layer), and let $R^{(t)}$ be the set of nodes that have already been assigned to some layer. Note that if a node $u\in G^{(t)}$ has some neighbor $v\in R^{(t)}$, then $u$ will in the end be the unique neighbor of $v$ in a higher layer. We can therefore think of the nodes in $G^{(t)}$ as the roots of the already raked subtrees. This is illustrated in \Cref{fig:3ColExample}.
\begin{figure}[h!]
    \centering
    \includegraphics[width=.6\textwidth]{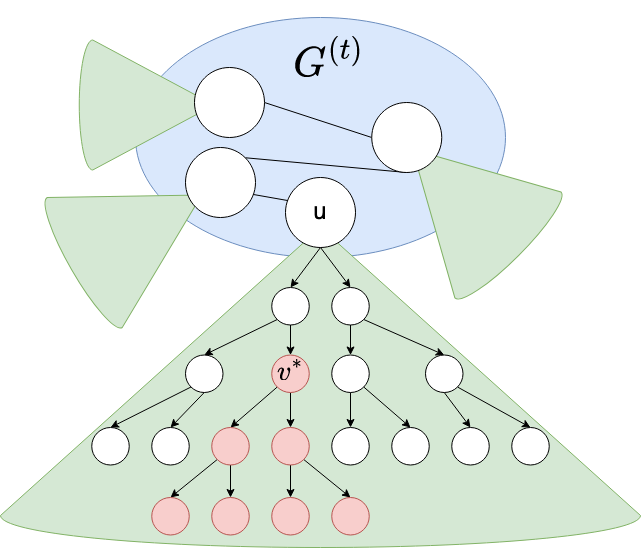}
    \caption{The graph $G^{(t)}$ of nodes that are not yet raked away is colored blue. The already raked away nodes $R^{(t)}$ are colored green. The node $u$ chooses $v^*$ since it has the largest subtree, colored in red, attached and both $v^*$ as well as its entire subtree become colored.}
    \label{fig:3ColExample}
\end{figure}
After each rake step $t$, each node $u\in G^{(t)}$ tries to color some node at distance $2$ in its subtree.\footnote{By only coloring nodes at distance at least $2$ from $u$, we make sure that neighbors of nodes that are not yet layered remain uncolored.}
Node $u$ chooses $v^*$ to be a node at distance $2$ in its subtree such that the subtree rooted at $v^*$ has the largest number of uncolored nodes among all nodes at distance $2$ of $u$ in the subtree of $u$ (observe that nodes can keep track of such numbers). If there are no colored $2$-hop neigbors of $v^*$ outside the subtree of $v^*$ (i.e., no colored siblings of $v^*$), then $u$ decides to color $v^*$ and its complete subtree. Otherwise, no new nodes in $u$'s subtree are getting colored. If $v^*$ and its subtree get colored, then a constant fraction of the uncolored nodes in $u$'s subtree get colored. Otherwise, a sibling $v'$ of $v^*$ with a larger subtree has already been colored while $u$ was the root of the tree. Note that at this time, the subtree of $v^*$ was already in the same state and therefore $v'$ colored more nodes than $v^*$ does. One can use this to show that whenever the height of a raked subtree increases, a constant fraction of the uncolored nodes gets colored. One can further show that this suffices to show that over the whole tree, a constant fraction of the remaining nodes gets colored every constant number of rounds and thus the node-averaged complexity is constant. The algorithm and the analysis for the general family of LCLs for which \Cref{thm:logstaraverage} holds uses similar basic ideas, dealing with the general case is however significantly more involved. 

\subparagraph{Constant Average using randomization}
The $\log^*n$ term in the deterministic case stems from the fact, that we need to break symmetry in the compress paths. We achieve this by precomputing a distance $s$ coloring, for a constant $s$ and with a constant number of colors. As a result we can in every execution of compress compute a constant distance ruling set by iterating through the colors. Using randomization we can omit the initial coloring and instead compute such a ruling set with constant node-average. We then show that replacing the original compress procedure of the deterministic algorithm with the new randomized compress procedure still gives us the same guarantees, but omits the initial coloring. As a result we get a randomized algorithm with node-averaged complexity $O(1)$.

\subparagraph{Improved upper bounds in the polynomial regime.}
We prove that the node-averaged complexity of the hierarchical $2 \frac{1}{2}$-coloring problem with parameter $k$ is $O(n^{1/(2^k - 1)})$. In order to give some intuition for this, we focus on the case $k=2$ where the worst-case complexity is $\Theta(\sqrt{n})$. Instead of providing a formal definition of the problem, it is helpful to present the problem by describing how a worst-case instance for the problem looks like, and how a solution in such an instance looks like. A worst-case instance for this problem consists of a path $P$ of length $\Theta(\sqrt{n})$, where to each node $v_j$ of $P$ is attached a path $Q_j$ of length $\Theta(\sqrt{n})$. We call the nodes of the path $P$ \emph{$p$-nodes} and we call the nodes of a path $Q_j$ \emph{$q$-nodes}. For each path $Q_j$, the algorithm has to decide to either $2$-color it or to mark the whole path as \emph{decline}. Then, the subpaths of $P$ induced by nodes that are neighbors of $q$-nodes that output decline need to be labeled with a proper $2$-coloring. In particular, \emph{decline} is not allowed on $p$-nodes.
Let us now describe an algorithm with optimal worst-case complexity for instances with a similar structure, but where the paths may have different lengths.
For $q$-nodes, the algorithm first checks if the length of the path containing those nodes is $O(\sqrt{n})$ (note that, in order to perform this operation, the algorithm needs to know $n$, and it is actually unknown whether an LCL problem can have $\Theta(\sqrt{n})$ worst-case complexity when $n$ is unknown). In such a case, the algorithm is able to produce a proper $2$-coloring of the path. Otherwise, the path is marked as decline. Then, it is possible to prove that the subpaths of $P$ induced by nodes having  $q$-node neighbors that output decline must be of length $O(\sqrt{n})$, and hence they can be properly $2$-colored in $O(\sqrt{n})$ rounds. We observe that in the worst-case instance described above, the majority of the nodes of the graph are $q$-nodes, and hence, from an average point of view, it would be fine if $p$-nodes spend more time. In fact, it is possible to improve the node-averaged complexity of the described algorithm by letting $q$-nodes run for at most $O(n^{1/3})$ rounds and $p$-nodes for at most $O(n^{2/3})$ rounds. In this case, a worst-case instance contains a path $P$ of length $O(n^{2/3})$ and all paths $Q_j$ are of length $O(n^{1/3})$. We obtain that both the $p$-nodes and the $q$-nodes contribute $O(n^{4/3})$ to the sum of the running times, obtaining a node-averaged complexity of $O(n^{1/3})$.

\subparagraph{Lower bounds in the polynomial regime.}
It is known by \cite{chang20} that if an LCL problem $\Pi$ has worst-case complexity $o(n^{1/k})$, then it can actually be solved in $O(n^{1/(k+1)})$ rounds. The intuition about what determines the exact value of $k$ in the complexity of a problem is related to how many compress layers of a rake-and-compress decomposition one can handle. In the example presented above, namely $3$-coloring, one can handle an arbitrary number of compress paths and that is the reason why the problem can be solved in $O(\log n)$ rounds. In particular, no matter how many rake or compress operations have been applied, we can handle any compress path by producing a $3$-coloring on it and leaving the endpoints uncolored (such nodes can decide their color after their higher layer neighbors picked a color), and this can be done fast. Not all problems are of this form, that is, for some problems we cannot handle an arbitrary amount of compress paths: it is possible to define problems in which different labels need to be used in compress paths of different layers (hierarchical $2\frac{1}{2}$ coloring is indeed such a problem where in fact $p$-nodes are not allowed to output \emph{decline}).
For such problems, it may not be possible at all to efficiently produce a valid labeling for long compress paths of layers that are too high, say of layers strictly more than $k$. In order to solve this issue, we can modify the generic algorithm sketched above by increasing the number of rake operations that are performed between each pair of compress operations. When using $\Omega(n^{1/(k+1)})$ rake operations at the beginning and between any two compress operations, the total number of compress layers is at most $k$. This however makes the algorithm slower, resulting in a complexity of $\Theta(n^{1/(k+1)})$ (while $3$-coloring has worst-case complexity $\Theta(\log n)$).

In other words, for some LCL problems, compress paths are something that is difficult to handle, and the number of compress layers that we can recursively handle is what determines the complexity of a problem. If we can handle an arbitrary amount of compress layers, then the problem can be solved in $O(\log n)$ rounds, but if we can handle only a constant amount of compress layers, say $k$, then the complexity of the problem is $\Theta(n^{1/(k+1)})$.
In \cite{chang20} it is proved that, if a problem has complexity $o(n^{1/k})$, then it is possible to handle $k$ compress layers, implying a complexity of $O(n^{1/(k+1)})$. We show that the same can be obtained by starting from an algorithm $\mathcal{A}$ with node-averaged complexity $o(n^{1/(2^k - 1)} / \log n)$, implying that if a problem has complexity $\Omega(n^{1/k})$, then it cannot have node-averaged complexity $o(n^{1/(2^k - 1)} / \log n)$, since otherwise it would imply that the problem can actually be solved in $O(n^{1/(k+1)})$ rounds in the worst case, which then leads to a contradiction. Starting from an algorithm that only has guarantees on its node-averaged complexity instead of on its worst-case complexity introduces many additional challenges that we need to tackle. For example, in \cite{chang20} it is argued that an $o(n^{1/k})$-rounds algorithm can never see both the endpoints of a carefully crafted path that is too long. This kind of reasoning, that is very common when we deal with worst-case complexity, does not work for node-averaged complexity.

\section{Road Map}

The remainder of the paper is organized as follows.

\subparagraph{Preliminaries.} We start, in \Cref{sec:definitions}, by providing some definitions. In particular, we define the class of problems that we consider, and the notion of node-averaged complexity.

\subparagraph{Locally checkable labelings.}
We continue, in \Cref{sec:lcls}, by providing an overview of what is known about solving LCL problems on trees of bounded degree. 
We present a generic algorithm that is known to be able to solve all LCLs, that has optimal worst-case complexity whenever the considered problem has worst-case complexity $\Omega(\log n)$.  
The content of this section is heavily based on existing results.

\subparagraph{A fast algorithm for problems with intermediate complexity.}
In \Cref{sec:lowreg}, we present an algorithm with node-averaged complexity $O(\log^* n)$, that is able to solve all problems that have $O(\log n)$ worst-case complexity. This algorithm is based on the one presented in \Cref{sec:lcls}, but we need to tackle many challenges in order to improve its node-averaged complexity. Then using randomization, we improve on this algorithm and obtain randomized $O(1)$ node-averaged complexity.

\subparagraph{A fast algorithm for some problems with polynomial complexity.}
In \Cref{sec:polyupper} we consider a class of problems, called hierarchical $2 \frac{1}{2}$ coloring, that are parametrized by an integer $k$. For these problems it is known that their worst-case complexity is $\Theta(n^{1/k})$. We show that these problems can be solved with node-averaged complexity $O(n^{1/(2^k - 1)})$. 

\subparagraph{A lower bounds for problems with polynomial complexity.}
It is known that all LCL problems on trees either have worst-case complexity $O(\log n)$, or $\Theta(n^{1/k})$ for some integer $k \ge 1$. While we show in  \Cref{sec:lowreg} that the worst-case complexity class $O(\log n)$ becomes $O(\log^* n)$ for node-averaged complexity, in \Cref{sec:lower} we show that all problems that have polynomial worst-case complexity also have polynomial node-averaged complexity. In particular, we show that if a problem has worst-case complexity  $\Theta(n^{1/k})$, then it has node-averaged complexity $\Omega(n^{1/(2^k -1)} / \log n)$. 

\subparagraph{Open questions}
In \Cref{sec:open}, we provide some open questions.

\subparagraph{More on LCLs}
In \Cref{apx:related-lcls}, we provide additional related work about LCLs.

\subparagraph{\boldmath An algorithm for solving all LCLs in $O(D)$ rounds.}
In \Cref{sec:diamsolver}, we show a simple bandwidth-efficient algorithm that solves all LCL problems in $O(D)$ rounds, where $D$ is the diameter of the graph, and then we give some intuition on the challenges that one needs to tackle in order to improve its complexity. The content of this section can be useful to better understand \Cref{sec:lcls}.

\subparagraph{Different ways to define LCLs.}
In \Cref{sec:different-lcls}, we provide different definitions of LCLs, and we prove that the notion that we study (called black-white formalism) is equivalent, for node-averaged complexity, to the standard one studied in the literature (this was previously known only for the case of worst-case complexity).

\section{Preliminaries}\label{sec:definitions}

\subparagraph{LCLs in the black-white formalism.}
We start by defining the class of problems that we consider, called LCLs in the black-white formalism. We show in \Cref{lem:node-edge-enough} that on trees, studying this class of problems is equivalent to studying LCLs as they are usually defined in the literature.
A problem $\Pi$ described in the black-white formalism is a tuple $(\Sigma_{\mathrm{in}},\Sigma_{\mathrm{out}},C_W,C_B)$, where:
\begin{itemize}
    \item $\Sigma_{\mathrm{in}}$ and $\Sigma_{\mathrm{out}}$ are finite sets of labels.
    \item $C_W$ and $C_B$ are both multisets of pairs, where each pair $(\ell_{\mathrm{in}},\ell_{\mathrm{out}})$ is in $\Sigma_{\mathrm{in}} \times \Sigma_{\mathrm{out}}$. 
\end{itemize}
Solving a problem $\Pi$ on a graph $G$ means that:
\begin{itemize}
    \item $G = (W \cup B,E)$ is a graph that is properly $2$-colored, and in particular each node $v \in W$ is labeled $c(v) = W$, and each node $v \in B$ is labeled $c(v) = B$.
    \item To each edge $e \in E$ is assigned a label $i(e) \in \Sigma_{\mathrm{in}}$.
    \item The task is to assign a label $o(e) \in \Sigma_{\mathrm{out}}$ to each edge $e \in E$ such that, for each node $v \in W$ (resp.\ $v \in B$) it holds that the multiset of incident input-output pairs is in $C_W$ (resp.\ in $C_B$).
\end{itemize}

Note that when expressing a given LCL problem on a tree $T$ in the black-white formalism, we often have to modify the tree $T$ as follows. We subdivide every edge $e$ of $T$ by inserting one node in the middle of the edge. Each edge is then split into two ``half-edges'' and the new tree is trivially properly $2$-colored (say the original nodes of $T$ are the black nodes and the newly inserted nodes for each edge of $T$ are the white nodes).

\subparagraph{Node-averaged complexity.}
We define the notion of node-averaged complexity as in \cite{BalliuGKO22_average}.
Let $\calA$ be an algorithm that solves a problem $\Pi$. Assume $\calA$ is run on a given graph $G=(V,E)$. Let $v\in V$. We define $T_v^G(\calA)$ to be the number of rounds after which $v$ terminates when running $\calA$. 
The node-averaged complexity of an algorithm $\calA$ on a family of graphs $\calG$ is defined as follows.
\[
    \nodeavg(\calA)  :=  \max_{G\in \calG} \,\frac{1}{|V|}\cdot\E\left[\sum_{v\in V(G)}T_v^G(\calA)\right]\ =\
    \max_{G\in \calG} \frac{1}{|V|}\cdot\sum_{v\in V(G)}\E\big[T_v^G(\calA)\big]\\
\]
The complexity of $\Pi$ is defined as the lowest complexity of all the algorithms that solve $\Pi$.

\section{Locally Checkable Labelings}\label{sec:lcls}
In this section we introduce a class of problems called Locally Checkable Labelings (LCLs) and we summarize known results about the worst-case complexity of LCLs in the case in which we restrict the family of considered graphs to be trees of bounded degree. We first provide an overview of what are the possible complexities of LCLs on trees, and then we describe a generic method that can be used to solve some of these problems optimally.

\subsection{Introduction}
LCLs have been extensively studied on trees of bounded degree and we know that in this family of graphs, these problems can only have the deterministic complexities $O(1)$, $\Theta(\log^* n)$, $\Theta(\log n)$, and $\Theta(n^{1/k})$ for any $k \in \mathbb{N}$ \cite{CP19timeHierarchy,chang20,BBOS18almostGlobal,LCLs_in_rooted_trees,B0COSS22_LCLregularTrees}. Moreover, we know that randomness may only help for problems with deterministic complexity $\Theta(\log n)$. Finally, we know that for problems with deterministic complexity $\Omega(\log n)$, given the description of an LCL problem, we can automatically determine its time complexity. These decidability results have first been shown in two important papers \cite{CP19timeHierarchy, chang20}, which we summarize in the rest of the section.

On a high level, these decidability results have been proven as follows: first, there is a generic method to solve all problems, based on a procedure called \emph{rake-and-compress}; then, it is shown that this method has optimal time complexity, meaning that if this method is not able to provide a fast algorithm, then the problem cannot be solved fast with any other algorithm.
In \Cref{sec:diamsolver}, we present the procedure for a simplified setting in which we aim at solving a restricted set of problems in $O(D)$ rounds.
There, we also explain which challenges one needs to tackle in order to improve the running time.

\subsection{A Generic Way to Solve All LCLs}\label{ssec:lcls}
In this section we present known results about solvability of LCLs. The content of this section is heavily based on results presented in \cite{CP19timeHierarchy,chang20,bcmos21}, that provide a generic method that is able to solve any (solvable) LCL. This method has an optimal worst-case complexity for all problems that require $\Omega(\log n)$ worst-case rounds in the deterministic setting. In later sections, we will use some ingredients that we present in this section, in order to provide algorithms that have faster node-averaged complexity. We follow a similar route of \cite{bcmos21}: in order to keep our proofs more accessible we prove our statements for LCLs expressed in the black-white formalism (which is simpler to deal with than standard LCLs as they are usually defined in the literature). In \Cref{lem:node-edge-enough} we prove that on trees, for any standard LCL, we can define an LCL in the black-white formalism that has the same asymptotic node-averaged complexity as the original one, implying that our results hold for all standard LCLs as well.

\subparagraph{Classes.}
A summary of the generic algorithm for solving any LCL in $O(D)$ rounds presented in \Cref{sec:diamsolver} is the following: nodes of degree $1$ are recursively removed from the tree; at each step, nodes that become leaves compute the set of labels that can be put on the edge connecting them with the rest of the tree, in a way that the labeling in their removed subtree can be completed in a valid manner; once all edges get a set assigned, it is possible to pick a valid labeling for all the edges by processing nodes in reverse order.

The set computed by a node essentially behaves as an interface between the remaining tree and the part of the tree that got removed, in the sense that it is not important what is the actual subtree, and the only thing that matters is the content of the set. Informally, we call this set the \emph{class} of the subtree.
In order to obtain algorithms that are faster than $O(D)$ rounds, in the removal process it is required to handle also nodes of degree $2$, and this is what complicates the formal definition of class. 

While the following definition is generic (in order to minimize repetition), it may help the reader to check \Cref{fig:label-set}, which shows the two possible cases in which the definition will be applied.

\begin{definition}[\cite{bcmos21}]\label{def:classes}
    Assume we are given an LCL $\Pi = (\Sigma_{\mathrm{in}},\Sigma_{\mathrm{out}},C_W,C_B)$ in the black-white formalism.
    Consider a tree $G = (V,E)$, and a connected subtree $H = (V_H, E_H)$ of $G$. Assume that the edges connecting nodes in $V_H$ to nodes in $V \setminus V_H$ are split into two parts, $F_{\mathrm{incoming}}$ and $F_{\mathrm{outgoing}}$, that are called, respectively, the set of incoming and outgoing edges. Assume also that for each edge $e \in F_{\mathrm{incoming}}$ is assigned a set $L_e \subseteq \Sigma_{\mathrm{out}}$. This set is called \emph{the label-set of $e$}. Let $\mathcal{L}_{\mathrm{incoming}} = (L_e)_{e \in F_{\mathrm{incoming}}}$. 
    A \emph{feasible labeling} of $H$ w.r.t.\ $F_{\mathrm{incoming}}$, $F_{\mathrm{outgoing}}$, and $\mathcal{L}_{\mathrm{incoming}}$ is a tuple $(L_{\mathrm{outgoing}},L_{\mathrm{incoming}},L_{\mathrm{H}})$ where:
    	\begin{itemize}
		\item $L_{\mathrm{incoming}}$ is a labeling $(l_e)_{e\in F_{\mathrm{incoming}}}$ of $F_{\mathrm{incoming}}$ satisfying $l_e\in (\mathcal{L}_{\mathrm{incoming}})_e$ for all $e \in F_{\mathrm{incoming}}$,
		\item $L_{\mathrm{outgoing}}$ is a labeling $(l_e)_{e\in F_{\mathrm{outgoing}}}$ of ${F_{\mathrm{outgoing}}}$ satisfying $l_e \in \Sigma_{\mathrm{out}}$ for all $e \in F_{\mathrm{outgoing}}$,
		\item $L_H$ is a labeling $(l_e)_{e\in E_H}$ of $E_H$ satisfying $l_e \in \Sigma_{\mathrm{out}}$ for all $e \in E_H$,
		\item the output labeling of the edges incident to nodes of $H$ given by $L_{\mathrm{outgoing}}, L_{\mathrm{incoming}},$ and $L_H$
        is such that all node constraints of each node $v\in V_H$ are satisfied.
	\end{itemize}
 	Also, we define the following:
	\begin{itemize}
		\item a \emph{class}  is a set of feasible labelings,
		\item a \emph{maximal class}  is the unique inclusion maximal class, that is, it is the set of all feasible labelings,
		\item an \emph{independent class} is a class $A$ such that
		for any $\big(L_{\mathrm{outgoing}}, L_{\mathrm{incoming}}, L_H\big)\in A$ and $\big(L'_{\mathrm{outgoing}}, L'_{\mathrm{incoming}}, L'_H\big)\in A$ the following holds. Let $L''_{\mathrm{outgoing}}$ be an arbitrary combination of $L_{\mathrm{outgoing}}$ and $L'_{\mathrm{outgoing}}$, that is, $L''_{\mathrm{outgoing}} = (l_e)_{e\in F_{\mathrm{outgoing}}}$ where $l_e \in \{ (L_{\mathrm{outgoing}})_e, (L'_{\mathrm{outgoing}})_e \}$. There must exist some $L''_{\mathrm{incoming}}$ and $L''_H$ satisfying $\big(L''_{\mathrm{outgoing}}, L''_{\mathrm{incoming}}, L''_H\big)\in A$.
	\end{itemize}
\end{definition}
Note that the maximal class with regard to some given $\Pi, H, F_{\mathrm{incoming}}, F_{\mathrm{outgoing}},$ and $\mathcal{L}_{\mathrm{incoming}}$, is unique. In contrast, there may be different ways (or none) to restrict a maximal class to a (nonempty) independent class. 

\subparagraph{Computing label-sets.}
We will use \Cref{def:classes} for two specific types of graphs $H$: either single nodes, or short paths. In each of these cases we will need to compute a label-set for each outgoing edge. We now describe the two cases in detail, and we provide a way to compute the label-sets. The cases are shown in \Cref{fig:label-set}.

\begin{figure}
	\centering
	\includegraphics[width=0.6\textwidth]{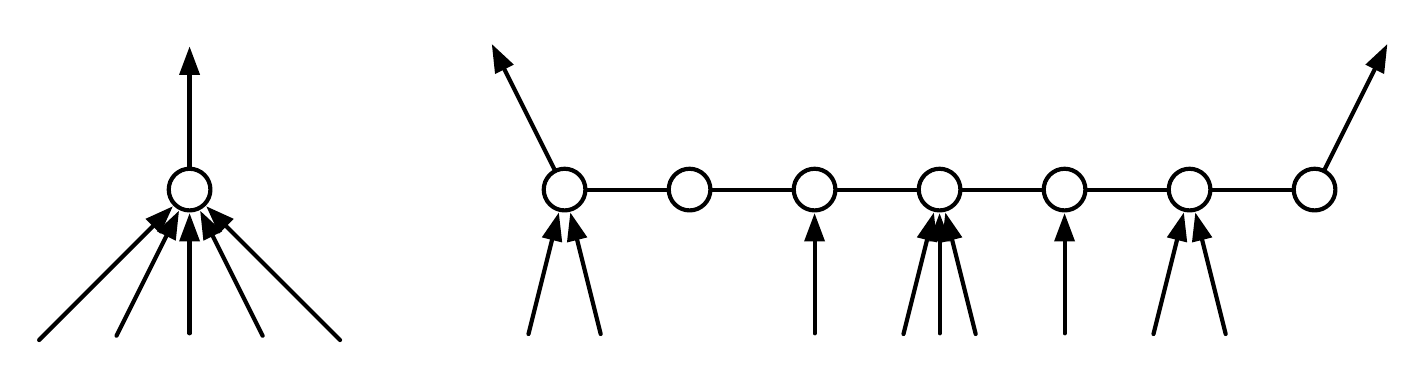}
	\caption{The figure illustrates the two cases of the label-set computation, where it is assumed that the incoming edges have already a label-set assigned and the goal is to assign a label-set to the outgoing edges; the left side depicts the case of a single node, the right side shows the case of a short path.}
	\label{fig:label-set}
\end{figure}

\begin{definition}[label-set computation]\label{def:computing-label-set}
Assume we are given some function $f_{\Pi, k}$ (to be specified later). We define a function $g(v)$ that can be used to compute label-sets for the outgoing edges as a function of $H$, $\Pi$, $F_{\mathrm{incoming}}$, $F_{\mathrm{outgoing}}$, $\mathcal{L}_{\mathrm{incoming}}$, and $f_{\Pi,k}$, for two specific types of graphs $H$.
\begin{itemize}
	\item \textbf{Single nodes:} the graph $H$ consists of a single node $v$ that has a single outgoing edge $e$, and hence $F_{\mathrm{outgoing}} = \{e\}$. All the other edges (which might be $0$) are incoming, and for each of them we are given a label-set (where $\mathcal{L}_{\mathrm{incoming}}$ represents this assignment). We assign, to the outgoing edge, the label-set $g(v)$, that consists of the set of labels that we can assign to the outgoing edge, such that we can pick a label for each incoming edge in a valid manner. More in detail, let $B$ be the maximal class of $H$ w.r.t.\ $\Pi, F_{\mathrm{incoming}}, F_{\mathrm{outgoing}},$ and $\mathcal{L}_{\mathrm{incoming}}$. Then, we denote $g(v)=\bigcup_{(L_{\mathrm{outgoing}}, L_{\mathrm{incoming}}, L_H)\in B}\{(L_{\mathrm{outgoing}})_e\}$. We have $g(v)\subseteq \Sigma_{\mathrm{out}}$. Observe that each node $v$ can compute $g(v)$ if it is given the value of $g(u)$ (that is, the label-set of the edge $\{u,v\}$) for each incoming edge $\{v,u\}$. 
	\item \textbf{Short paths:} the graph $H$ is a path of length between $\ell$ and $2 \ell$, for some $\ell = O(1)$ that depends solely on $\Pi$ and the target running time. The endpoints of the path are $v_1$ and $v_2$. The outgoing edges are $F_{\mathrm{outgoing}} = \{e_1,e_2\}$, where $e_1$ (resp.\ $e_2$) is the outgoing edge incident to $v_1$ (resp.\ $v_2$). Let $B$ be the maximal class of $H$.
	We assume to be given a function $f_{\Pi,k}$, that depends solely on $\Pi$ and some parameter $k$ (that, in turn, depends on the target running time), that maps a class $B$ into an independent class $B' = f_{\Pi,k}(B)$. For $i \in \{1,2\}$, let $g(v_i) = \bigcup_{(L_{\mathrm{outgoing}}, L_{\mathrm{incoming}}, L_H)\in B'}\{(L_{\mathrm{outgoing}})_{e_i}\}$. We have $g(v_i)\subseteq \Sigma_{\mathrm{out}}$.  The label-set of $e_1$ (resp.\ $e_2$) is $g(v_1)$ (resp.\ $g(v_2)$). Observe that the values of $g(v_i)$, for $i\in\{1, 2\}$, can be computed given $H$ and $\mathcal{L}_{\mathrm{incoming}}$.
\end{itemize}
\end{definition}

\subparagraph{Tree decompositions.}
All problems with worst-case complexity $O(\log n)$ or $O(n^{1/k})$ for any $k \in \mathbb{N}$ can be solved by following a generic algorithm \cite{CP19timeHierarchy,chang20,bcmos21}.
This algorithm decomposes the tree into layers by iteratively removing nodes in a rake-and-compress manner \cite{MillerR85} (and then uses the computed decomposition to solve the given problem).

We first define the decomposition that the algorithm uses and then elaborate on how fast (and how) it can be computed.

\begin{definition}[$(\gamma,\ell,L)$-decomposition]\label{def:modi}
    Given three integers $\gamma, \ell, L$, a \emph{$(\gamma,\ell,L)$-decomposition} is a partition of $V(G)$ into $2L-1$ layers $V_1^R = (V^R_{1,1},\ldots,V^R_{1,\gamma}), \ldots, V_{L}^R  = (V^R_{L,1},\ldots,V^R_{L,\gamma})$, $V_1^C, \ldots, V_{L-1}^C$ such that the following hold.
	\begin{enumerate}
		\item \label{prop:compress} Compress layers: The connected components of each $G[V_i^C]$ are paths of length in $[\ell,2\ell]$, the endpoints have exactly one neighbor in a higher layer, and all other nodes do not have any neighbor in a higher layer.
		\item \label{prop:rake} Rake layers: The diameter of the connected components in $G[V_i^R]$ is $O(\gamma)$, and for each connected component at most one node has a neighbor in a higher layer.
		\item \label{prop:isolated}The connected components of each sublayer $G[V^R_{i,j}]$ consist of isolated nodes. Each node in a sublayer $V^R_{i,j}$ has at most one neighbor in a higher layer or sublayer. 
	\end{enumerate}
\end{definition}

In the following, for completeness, we provide an algorithm of \cite{CP19timeHierarchy,chang20} that shows how to compute a $(\gamma,\ell,L)$-decomposition where, at the end, each node knows the layer that it belongs to.

\begin{enumerate}
    \item Iteratively do the following, until the obtained graph is empty, starting with $i=1$:
    \begin{enumerate}
        \item Iteratively, perform $\gamma$ \emph{rake operations}. For each $1 \leq j \leq \gamma$, the $j$th rake operation consists of removing all nodes that have degree $1$ in the graph induced by the remaining nodes and putting them into the preliminary set $W^R_{i,j}$.  For technical reasons, if two adjacent nodes have degree $1$, the rake operation only removes one of them, chosen arbitrarily.
        \item Perform one \emph{compress operation}, that is, remove all nodes of degree $2$ that, in the graph induced by the remaining nodes, are contained in paths of length at least $\ell$ where each node of the path has degree exactly $2$. Put the removed nodes into the preliminary set $W_i^C$.
        \item $i \leftarrow i+1$.
    \end{enumerate}
    \item Observe that, for all $i$, $W_i^C$ is a collection of paths of length at least $\ell$. Split long paths into shorter paths as follows. Compute, for each $W_i^C$, an independent set $I_i \subset W_i^C$ of nodes that satisfies two properties: no endpoint of any path in $W_i^C$ is contained in $I_i$, and the maximal connected components of the graph obtained by removing the nodes in $I_i$ from $W_i^C$ are paths of length between $\ell$ and $2 \ell$.
    For each $i$, we promote the nodes in $I_i$ to the next rake layer, that is, we define $V_i^C$ as $W_i^C \setminus I_i$, $V_{i+1}^R$ as $W_{i+1}^R \cup I_i$, and $V_{i+1,1}^R$ as $W_{i+1,1}^R \cup I_i$.
    We denote with $L$ the largest index $i$ such that $V_{i}^R \cup V_{i-1}^C \neq \emptyset$.
\end{enumerate}

The following lemma provides upper bounds for the deterministic worst-case complexity of computing a $(\gamma,\ell,L)$-decomposition using the algorithm of \cite{CP19timeHierarchy,chang20}.

\begin{lemma}[\cite{CP19timeHierarchy, chang20}]\label{lem:decomposition}
	Assume $\ell = O(1)$. Then the following hold.
    \begin{itemize}
        \item For any positive integer $k$ and $\gamma = n^{1/k}(\ell / 2)^{1 - 1/k}$, a $(\gamma,\ell,k)$-decomposition can be computed in $O(k\cdot n^{1/k})$ rounds.
        \item For $\gamma = 1$ and $L = O(\log n)$, a $(\gamma,\ell,L)$-decomposition can be computed in $O(\log n)$ rounds.
    \end{itemize}
\end{lemma}

\subparagraph{The generic algorithm with optimal worst-case complexity.}
We now explain the algorithm due to \cite{CP19timeHierarchy, chang20, bcmos21} that is able to solve any LCL $\Pi$ with worst-case complexity $\Theta(\log n)$ or $\Theta(n^{1/k})$ asymptotically optimally. 
The initial objective is to compute a $(\gamma,\ell,L)$-decomposition; however, first we need to determine suitable parameters $\gamma$, $\ell$, and $L$.
To this end, we begin by determining the worst-case time complexity of the given problem $\Pi$, which, by \cite{CP19timeHierarchy, chang20}, can be computed in finite time solely as a function of $\Pi$. Then, as a function of the target time complexity, we can determine $\gamma$ and $\ell$: if the target time complexity is $O(\log n)$, then $\gamma = 1$, and then $\ell$ can be computed as a function of $\Pi$; otherwise, if the target time complexity is $O(n^{1/k})$, then we can first compute $\ell$ as a function of $\Pi$ and $k$, and then set $\gamma = n^{1/k}(\ell / 2)^{1 - 1/k}$. By \Cref{lem:decomposition}, we get that if in the former case we set $L = k$, and in the latter case we set $L = O(\log n)$, then a $(\gamma,\ell,L)$-decomposition can be computed within a running time that matches the target complexity. In \cite{CP19timeHierarchy, chang20} it is shown how to determine the value of $\ell$ in each case.

After computing a $(\gamma,\ell,L)$-decomposition with the determined parameters, the decomposition is used to propagate label-sets up through the layers.
The goal is to iteratively assign \emph{label-sets} to edges, in a way that, when we handle nodes in layer $i$, all their edges connected to nodes of layers $< i$ have already a label-set assigned.
As a last step of the generic algorithm, we will use the label-sets to pick a valid solution, by propagating labels through the layers in reverse order.

We first explain the procedure that computes a label-set for all edges of the graph by going up through the layers of the computed $(\gamma,\ell,L)$-decomposition. This procedure uses \Cref{def:computing-label-set}, and hence it requires to be given a function $f_{\Pi,k}$ (if the  target runtime is $O(\log n)$, then let $k = \infty$). This function, as shown in \cite{CP19timeHierarchy,chang20,bcmos21}, can be computed solely as a function of $\Pi$ and the target time complexity (also, similarly as for $\ell$, this function does not depend on $n$). 
\begin{itemize}
    \item For $j = 1, \ldots, \gamma$, by \Cref{def:modi}, we have that the graph induced by $V_{i,j}^R$ is composed of isolated nodes. Each node $v \in V_{i,j}^R$ waits until all its neighbors $z$ in lower layers have assigned a label-set to the edge $\{v,z\}$. We are now in the first case of \Cref{def:computing-label-set}. Hence, $v$ computes $g(v)$, the label-set of the edge connecting $v$ to its only neighbor $w$ in upper layers, if it exists, and sends this label-set to $w$.
    \item Nodes in $V_i^C$, by \Cref{def:modi}, form paths of length between $\ell$ and $2 \ell$. For each of these paths, nodes wait until a label-set has been computed for each edge connecting them to lower layers. We are now in the second case of \Cref{def:computing-label-set}. Hence, the endpoints of the path can compute the label-sets for the edges connecting them to their neighbor in upper layers. Each endpoint sends its computed label-set to its higher-layer neighbor.
\end{itemize}
Now we explain how, in the last step of the generic algorithm, we use the computed label-sets to determine the final output labels by going through the layers of the computed decomposition in reverse order (i.e., from larger to smaller index). Observe that it may happen that rake nodes do not have any outgoing edge. In this case, a node considers its maximal class $B$ (that is guaranteed to be non-empty), takes an arbitrary element $((),L_{\mathrm{incoming}},())$ from $B$, and uses $L_{\mathrm{incoming}}$ to assign a label to each incoming edge. Then, we process the rest of the graph as follows.
\begin{itemize}
    \item Each rake node, once a label $l$ for its outgoing edge is assigned, considers its maximal class $B$, picks an arbitrary element $(L_{\mathrm{outgoing}},L_{\mathrm{incoming}},())$ compatible with $l$ from $B$, and uses $L_{\mathrm{incoming}}$ to assign a label to each incoming edge, where compatible with $l$ means that $L_{\mathrm{outgoing}}$ assigns $l$ to the outgoing edge.
    \item Similarly, each short path, once a label is assigned to the edges outgoing from the endpoints, can pick a valid assignment for each internal and incoming edge.
\end{itemize}
This concludes the description of the generic algorithm.
We note that the generic algorithm does not require a specific $(\gamma,\ell,L)$-decomposition---any $(\gamma,\ell,L)$-decomposition (for the parameters $\gamma,\ell,L$ determined in the beginning of the generic algorithm) works.
We will make use of this fact when designing algorithms with a good node-averaged complexity in \Cref{sec:lowreg}.

We now define a total order on the layers of a $(\gamma, \ell, L)$-decomposition in the natural way.
This will be useful in the design of our algorithm in \Cref{sec:lowreg}.
\begin{definition}[layer ordering]\label{def:ordering}
    We define the following total order on the (sub)layers of a $(\gamma, \ell, L)$-decomposition.
    \begin{itemize}
        \item  $V_{i,j}^R < V_{i',j'}^R$ iff $i < i' \lor (i = i' \land j < j')$
        \item $V_{i,j}^R < V_{i}^C$
        \item $V_{i}^C < V_{i+1,j}^R$
    \end{itemize}
    Accordingly, we will use terms such as ``lower layer'' to refer to a layer that appears earlier in the total order than some considered other layer.
\end{definition}
For the interested reader, in \Cref{ssec:good-functions}, we provide some more intuition about the function $f_{\Pi,k}$, and about how the existence of this function is related with the complexity of a problem.

\section{Algorithm for Intermediate Worst-Case Complexity Problems}\label{sec:lowreg}
In this section, we provide an algorithm with node-averaged complexity $O(\log^* n)$ on bounded-degree trees for all LCLs that can be solved in worst-case complexity $O(\log n)$ on bounded-degree trees. The high-level idea of our algorithm is similar to the strategy in the generic algorithm of \Cref{ssec:lcls}: compute a $(\gamma,\ell, L)$-decomposition, propagate types up through the layers, and then propagate a choice of labels through the layers in reverse order.
Importantly, in order to obtain a good node-averaged complexity, we will not wait for the decomposition to be completely computed, but instead we show that we can allow a sufficient amount of nodes to terminate early. A crucial ingredient for achieving this is identifying some nodes with nice properties which we promote to be a \emph{local maximum} (that is, nodes that have a layer number that is strictly higher than all their neighbors), allowing it and many other nodes to terminate.

However this approach requires us to find a suitable $(\gamma,\ell, L)$-decomposition that in particular guarantees that sufficiently many nodes become local maxima. Observe that the standard decomposition algorithm may not guarantee this property: in a balanced binary tree for example, only the root would become a local maximum. To achieve the required property we will compute an altered decomposition that in some sense leaves a bit of extra space between two compress layers. We will use this extra space to insert new compress paths creating local maxima. Here we need to take special care that such a promotion is possible while guaranteeing that at the end we still have a valid $(\gamma,\ell, L)$-decomposition.

As this part is quite long and it is possible to get lost in the details, we next provide an overview of the structure of this section.

\begin{itemize}
    \item \textbf{\Cref{sec:newdec}:} The Decomposition Algorithm
        \begin{itemize}
            \item Stating a centralized version of the algorithm. \Cref{alg:Rake,alg:Compress,alg:Promote,alg:Decomposition}
            \item Proving that it computes a valid $(\gamma, \ell, L)$-decomposition and thereby proving correctness. \Cref{def:partialdec,lem:rakeCorrect,lem:compressCorrect,lem:promoteCorrect}
            \item Proving that we need $O(\log n)$ iterations. \Cref{lem:rcremaining,lem:epsilon,cor:inlog}
        \end{itemize}
    \item \textbf{\Cref{sec:insights}:} Local Maxima and Bounding Quality
        \begin{itemize}
            \item Decomposing the tree into subtrees. \Cref{def:OrientationNames,def:subtreeAssigned,lem:decomposeV}
            \item The promotion creates local maxima and they account for a large amount of quality. \Cref{lem:fixV,lem:fixX}
            \item Upper bounding the quality of a free node. \Cref{lem:layersPartiallyFixed,lem:distanceOfJustRemoved,lem:boundAssignedNode}.
        \end{itemize}
    \item \textbf{\Cref{sec:dectoave}:} Distributed Algorithm and Node Averaged Complexity
        \begin{itemize}
            \item The distributed implementation and the round complexity of an iteration. \Cref{lem:TimeOfIteration}
            \item Proving $O(\log^*n)$ node-averaged complexity. \Cref{lem:TimeToFixSubtree,lem:FixedInIteration,lem:AvgIteration,cor:markedNodesTerminate}
        \end{itemize}
    \item \textbf{\Cref{sec:randomConstant}:} Improved Node Averaged Complexity using randomization
        \begin{itemize}
            \item Compress as the key bottleneck. \cref{lem:FastCompressGivesConstant}
            \item Fast randomized algorithm for Compress \cref{alg:RandCompress,alg:Elect} and \Cref{lem:ProbJoinZ,lem:ElectRuntime,lem:RulingSet,lem:RandCompressCorrect,lem:ProbZDone,lem:constantExpectation}
        \end{itemize}
\end{itemize}

\subsection{The Decomposition Algorithm}\label{sec:newdec}
During the execution of the algorithm we will not yet have a complete $(\gamma, \ell, L)$-decomposition. Say we only have layers up until the $i$-th layer, $V_{1}^R, \dots, V_{i}^R$ and $V_{1}^C, \dots, V_{i}^C$, that already satisfy the definition of a $(\gamma, \ell, L)$-decomposition. We distinguish between nodes that have already been assigned a layer and those that have not.
\begin{definition}[free and assigned nodes] \label{def:freeOrAssigned}
    We call a node that has not been assigned to a layer a \emph{free} node.
    A node that has been assigned to a layer is called \emph{assigned}.
\end{definition}

\begin{figure}[h!]
    \centering
    \includegraphics[width=.8\textwidth]{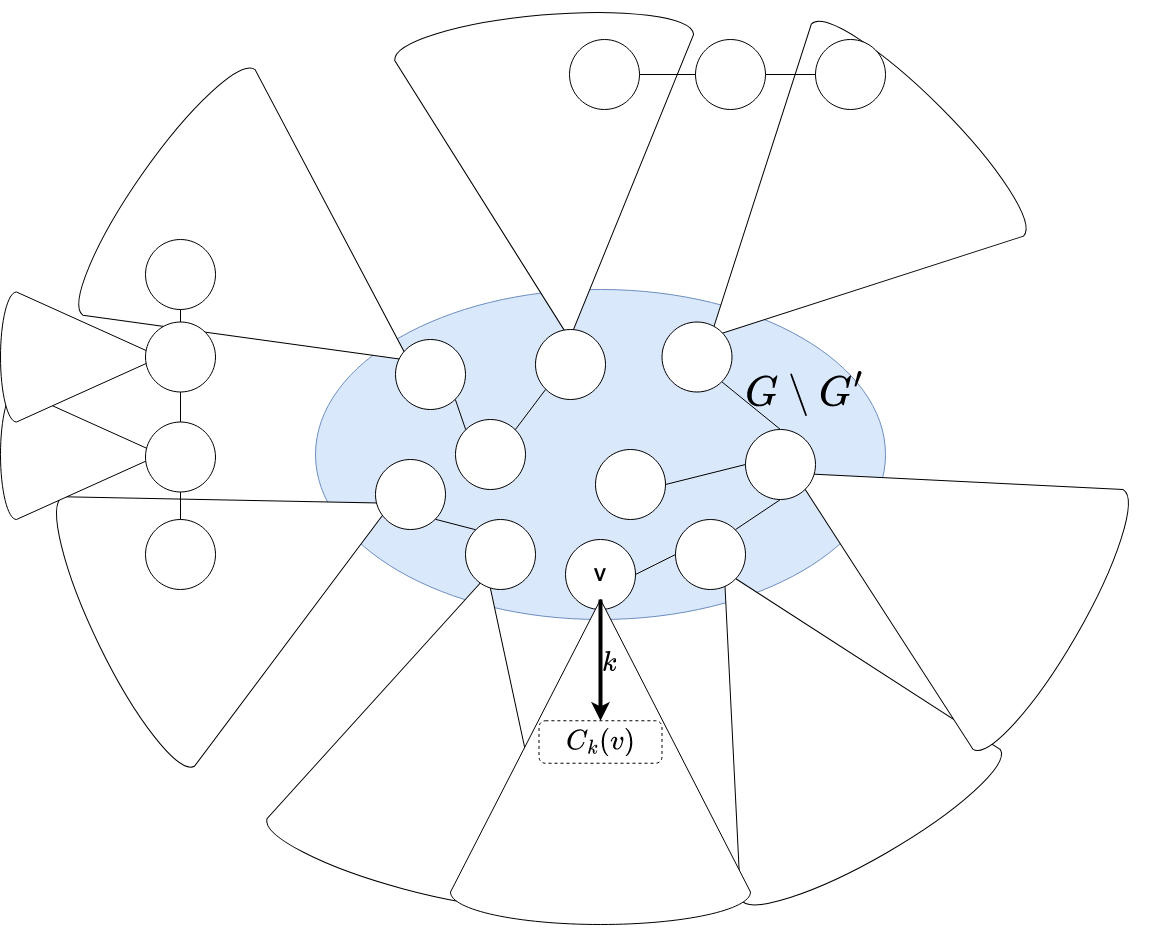}
    \caption{A picture during the execution of a rake and compress like algorithm. We have in blue the graph of remaining free nodes. All of the already assigned nodes are hanging in subtrees from the nodes of $G\setminus G'$. The compress paths are not part of any subtrees and connect the components of free nodes. Also the set $C_k(v)$ is exactly the nodes at distance $k$ from $v$ that are in the tree hanging from $v$.}
    \label{fig:structure}
\end{figure}
Let $G$ denote our input tree and assume that a subset of nodes has already been assigned to some layers.
Let $G'$ denote the subgraph of $G$ induced by all assigned nodes. View \Cref{fig:structure} for the general picture that should be kept in mind for the rest of this section.\\
We now define two notions that will be crucial for our approach.
Recall that we have a total ordering on the layers of a decomposition due to \Cref{def:ordering} (that naturally extends to partial decompositions).
\begin{definition}[local maximum]\label{def:locMax}
    A \emph{local maximum} is an assigned node $v \in V(G')$ with the following two properties:
    \begin{enumerate}
        \item Node $v$ and all of its neighbors are assigned, i.e., they are all contained in $V(G')$.
        \item For each neighbor $w$ of $v$, the layer of $w$ is strictly smaller than the layer of $v$.
    \end{enumerate}
\end{definition}
\Cref{alg:Promote} is a subroutine that changes an existing partial decomposition, by \emph{promoting} a node to a higher layer in order to create a new local maximum. We will later see, that local maxima and all nodes \emph{below} them can terminate early, so this is a powerful tool. The second notion allows us to measure how efficient a node would be as a local maximum and we will use it to decide which nodes to promote. The quality of a node $v$ is the number of nodes that are depending only on $v$ to pick its label so they can pick their labels. So if $v$ terminates and chooses an output all of these nodes can also terminate. 

To keep track of which nodes are waiting for other nodes in order to terminate, we will orient edges. Our algorithm will orient the edges in such a manner that any node $u$ will have a consistently oriented path from $v$ to $u$ if and only if $u$ will be able to terminate if $v$ does. So the orientation keeps track of how labels would propagate through the graph see \Cref{fig:quality}. 
In practice this means if $u$ is raked away we orient the single edge from $u$'s parent towards itself and orienting the ends of compress paths inwards.
\begin{figure}[h!]
    \centering
    \includegraphics[width=.8\textwidth]{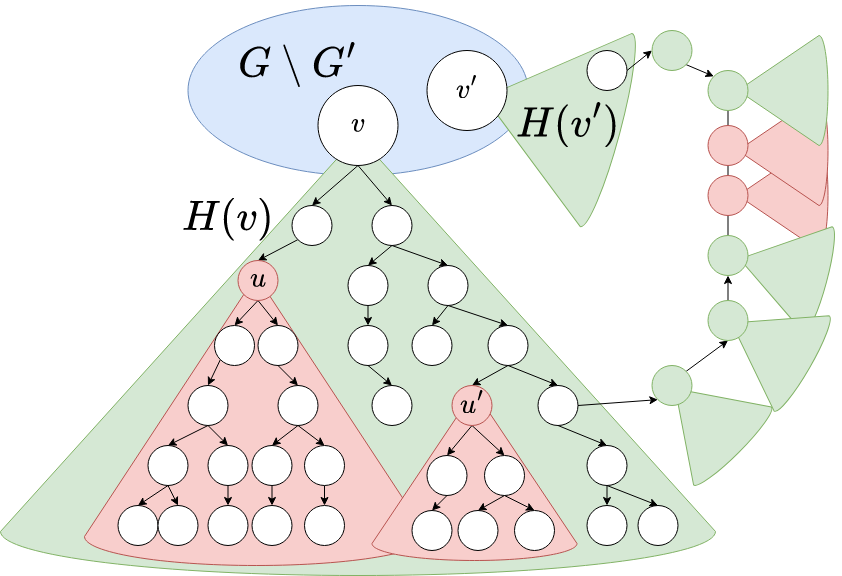}
    \caption{Illustrating the quality of a node. $v$ and $v'$ are free nodes inside of $G\setminus G'$. $H(v)$ (respectively $H(v')$) are all nodes inside the green cone attached to $v$ ($v'$) and the green nodes in the compress path. $u$ and $u'$ are local maxima, so because of point 3 in \Cref{def:quality} they and the red trees hanging from them do not contribute to $H(v)$ (respectively $H(v')$.}
    \label{fig:quality}
\end{figure}
If an edge is not explicitly oriented it is not considered oriented at all. For some given $v$ this orientation now defines $H(v)$ a subgraph of nodes that can be reached over oriented paths. This graph is then exactly all of the nodes that are only waiting for $v$ to choose an output and hence these could all terminate if $v$ became a local maximum.

\begin{definition}[quality] \label{def:quality}
    For any node $v \in V(G)$, let $H(v)$ denote the set of all nodes $w$ that can be reached from $v$ via a path $(v=v_0, v_1, \dots, v_j = w)$ such that the following hold:
    \begin{enumerate}
        \item The edge $\{ v_{i-1}, v_i \}$ is oriented from $v_{i-1}$ to $v_i$, for each $1 \leq i \leq j$.
        \item All nodes on the subpath from $v_1$ to $w$ are assigned, i.e., they are all contained in $V(G')$.
        \item The layer of $v_i$ is smaller than or equal to the layer of $v_{i-1}$, for each $2 \leq i \leq j$, and if $v_0$ is assigned, then the layer of $v_1$ is smaller than or equal to the layer of $v_0$.
    \end{enumerate}
    If $v$ is a local maximum, or a descendant of a local maximum, then the quality $q(v):= 0$.
    Otherwise the \emph{quality} $q(v)$ of a node $v$ is the number of nodes in $H(v)$, i.e., $q(v) := |H(v)|$.
\end{definition}

We chose the name quality, since we will later decide on which nodes to fix by trying to maximize this quantity. We will later also see that the sum over the quality of all free nodes are exactly all of the nodes that are yet to terminate. So by giving an upper bound on $\sum_{v\in G\setminus G'}q(v)$ we will be able to argue about the progress of our algorithm.

\subparagraph{The algorithm.}
For the correctness of our algorithm we maintain the invariant of a valid partial $(\gamma, \ell, L)$-decomposition. It is almost an exact restatement of \Cref{def:modi}, with two key differences. First we only care about a subset $V' \subseteq V(G)$ of nodes. Second, we add additional rake and compress layers for promoted nodes.

\begin{definition}[partial $(\gamma, \ell, L)$-decomposition] \label{def:partialdec}
    A \emph{partial $(\gamma, \ell, L)$-decomposition} of a graph $G$ is a partition of some subset $V' \subseteq V(G)$ into $2L-1$ layers $V_1^R = (V^R_{1,1},\ldots,V^R_{1,\gamma}), \ldots, V_{L}^R  = (V^R_{L,1},\ldots,V^R_{L,\gamma})$, $V_1^C, \ldots, V_{L-1}^C$ plus additional layers $V^C_{i,p}$ for all appropriate $i$.
    The layering satisfies the following.
	\begin{enumerate}
		\item \label{prop:compressprime} Compress layers: The connected components of each $G[V_i^C]$ and $G[V_{i,p}^C]$ are paths of length in $[\ell,2\ell]$, the endpoints have exactly one neighbor that is in a higher layer or free, and all other nodes do not have any neighbor that is in a higher layer or free.
		\item \label{prop:rakeprime} Rake layers: The diameter of the connected components in $G[V_i^R]$ is $O(\gamma)$, and for each connected component at most one node has a neighbor that is in a higher layer or free.
		\item \label{prop:isolatedprime} The connected components of each sublayer $G[V^R_{i,j}]$ consist of isolated nodes. Each node in a sublayer $V^R_{i,j}$ has at most one neighbor in a higher layer or sublayer.
	\end{enumerate}
    The ordering is now given by
    \begin{itemize}
        \item  $V_{i,j}^R < V_{i',j'}^R$ iff $i < i' \lor (i = i' \land j < j')$
        \item $V_{i,j}^R < V^C_{i,p} < V_{i}^C < V^R_{i+1,1}$
    \end{itemize}
\end{definition}

We note, that the above definition is equivalent to the old definition in the sense that any decomposition as described above can  be turned into a valid $(\gamma, \ell, 2L)$-composition, by simply inserting $\gamma$ empty rake layers before each $V^C_{i,p}$.

\begin{corollary}\label{cor:PartialIsProperDecomposition}
A partial $(\gamma, \ell, O(\log n))$-decomposition with additional promotion compress layers $V^C_{i,p}$ is equivalent to a partial $(\gamma, \ell, O(\log n)))$-decomposition without them, where just the constant in the big $O$ changes by a factor of 2.
\end{corollary}

We first describe the two altered versions of rake and compress that we use in \cref{alg:Rake} and \cref{alg:Compress} respectively. Then we describe the procedure that promotes nodes in order to create local maxima in \cref{alg:Promote}. Finally the entire algorithm is given in \cref{alg:Decomposition} which then only consists of iterating the other procedures.\\
We start with the rake procedure, where the only difference is that we also create the edge orientation necessary to keep track of the quality of nodes.

\begin{algorithm2e}
\caption{Orienting Rake}\label{alg:Rake}
\KwIn{$G=(V,E), i, \gamma$} 
$G^{(i)} \gets G \setminus (\bigcup_{j<i}(V_i^R\cup V_{j}^C))$ \Comment{Graph induced by free nodes.}\\
\For{$j = 1$ \KwTo $\gamma$} 
{
    \For{every node $v \in G^{(i)}$ of degree 0 or 1 node}{
        \If{$v$ has a neighbor $u$} {
		      Orient $(u,v)$ from $u$ to $v$. \label{alg:line:orient3}
	    }
        $V_{i,j}^R \gets V_{i,j}^R \cup v$
    }    
}
\end{algorithm2e}
For the compress procedure given in \cref{alg:Compress} we change a bit more. We will later need the guarantee that compress paths are far enough away from nodes we want to promote. We achieve this by leaving a bit of slack at the ends of the compress path. Furthermore we also keep track of dependencies here by orienting the edges accordingly. The rest is just the classic compress procedure.
Such a compress procedure will always be followed by a set of orienting rakes, as a result all of the slack nodes are raked away nicely.

\begin{algorithm2e}
\caption{Compress with Slack}\label{alg:Compress}
\KwIn{$G=(V,E), i, \ell$} 
$G^{(i)} \gets G \setminus (\bigcup_{j<i}(V_i^R\cup V_{j}^C))$ \Comment{Graph induced by free nodes.}\\
\For{each maximal path $P$ consisting of nodes of degree exactly $2$ in $G^{(i)}$}  
{
    \If{$|V(P)|\geq 4\ell + 9$} {
        $P' \gets$ the subpath of $P$ consisting of all nodes that have distance at least $\ell + 3$ from both endpoints of $P$ \\
        compute a subset $Z \subseteq V(P')$ such that no two nodes in $Z$ are adjacent, no endpoint of $P'$ is in $Z$, and every maximal subpath of $P$ consisting only of nodes in $V(P') \setminus Z$ has length in $[\ell, 2\ell]$ \label{computez} \\
        add every node in $V(P') \setminus Z$ to $V_{i}^C$ \label{alg:line:firstCompress}\\
        add every node in $Z$ to $V_{i+1,1}^R$ \label{alg:line:maxone} \\
        let $P' = (v_0, \dots, v_w, \dots, v_y, \dots v_j)$, and let $v_w \in Z$ be the first node on $P'$ that is in $Z$ and $v_y \in Z$ the last. \\
        orient the edges in the path $(v_0,\dots, v_{w-1})$ from $v_0$ towards $v_{w-1}$ \label{alg:line:orient2.1}\\
        orient the edges in the path $(v_{y+1},\dots, v_{s})$ from $v_s$ towards $v_{y+1}$ \label{alg:line:orient2.2}\\ 
        $N^{(i)} \gets N^{(i-1)} \cup \{v_w, v_{w+1}, \dots , v_{y-1}, v_y\}$ \label{alg:line:assignN} \Comment{Only for the analysis}\\
	}
}
\end{algorithm2e}

In \cref{fig:compressOperation} we illustrate exactly whats happening in \cref{alg:Compress}. It is essentially just doing the normal compress procedure on the inner path $P'$ and orienting the edges inwards. Furthermore we note that the 
set $N^{(i)}$ consists exactly of those nodes that are either local maxima themselves or are between two local maxima. We will later see that all of the nodes in $N^{(i)}$ can already terminate after the compress procedure is done.

\begin{figure}[h!]
    \centering
    \includegraphics[width=\textwidth]{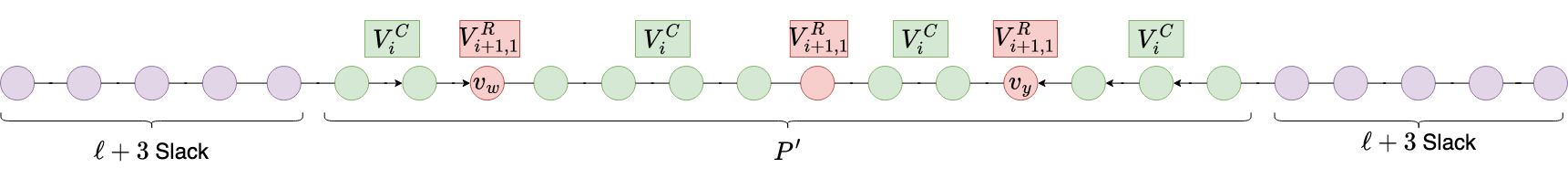}
    \caption{Illustrating the Compress with Slack procedure for $\ell = 2$, the purple nodes are slack at both ends of the path, the red nodes are in the set $Z$ and the green nodes make up the actual compress layer.}
    \label{fig:compressOperation}
\end{figure}

Now that we have the two subroutines that will create our decomposition we will take care of actually promoting some nodes to local maxima. For this we use the following definition
For each node $v \in V(G)$, let $C_k(v)$ denote the set of assigned nodes $w$ that have distance exactly $k$ from $v$ and for which all internal nodes of the unique path from $v$ to $w$ are assigned as well, i.e.,
\[
    C_k(v) = \{ w \in V(G') \mid \dist(v,w) = k, \text{ $x \in V(G')$ for all $x \neq v$ on the path from $v$ to $w$} \}.
\]
Refer to \Cref{fig:structure} for an example.
We determine the node that we want to promote by choosing the node $v^*$ that has the highest quality among nodes in $C_b(r)$. This  ensures that a lot of nodes can terminate early.\\
We promote a node $v^*$ by inserting an extra compress path and having $v^*$ be one of the endpoints of the new compress path. This only works if we do not intersect with any other compress path that already exists, as a result we are not always able to promote $v^*$. We will later see that when we are not able to, then this is because we have very recently been successful in promoting another close-by node. 

\begin{algorithm2e}
\caption{Promote if possible}\label{alg:Promote}
\KwIn{$G=(V,E), \Pi, i, b$} 
$G^{(i)} \gets G \setminus \left( \bigcup_{1 \leq a \leq i} V_a^R \cup \bigcup_{1 \leq a \leq i - 1} V_a^C \right)$ \label{alg:line:definitionG}\Comment{all free nodes} \\
\For{each node $r \in G^{(i)}$ \label{alg:line:eachr}}{
    $v^* \gets \argmax_{v \in C_b(r)} q(v)$ \Comment{breaking ties arbitrarily} \label{alg:line:argmax} \\
    $P \gets$ the path from $v^*$ to $r$ \label{alg:line:promotePath}\\
    \If{no node in $V(P)$ is currently assigned to a layer of the form $V_a^C$, or $V_{a,p}^C$ for some $1 \leq a \leq i-1$ \label{alg:line:conditionFixV*}}
    {
        remove every node in $V(P) \setminus \{ v^*, r \}$ from the layer it is currently assigned to and add it to $V_{i,p}^C$ \label{alg:line:reassigncee} \\
        remove $v^*$ from the layer it is currently assigned to and add it to $V_{i+1,1}^R$ \label{alg:line:maxtwo} \\
        mark $v^*$ as \emph{promoted} \label{alg:line:markPromoted}
    }
}
\end{algorithm2e}

\cref{fig:afterPromotion} shows how things will look after having successfully promoted $v^*$. 

\begin{figure}[h!]
    \centering
    \includegraphics[width=.8\textwidth]{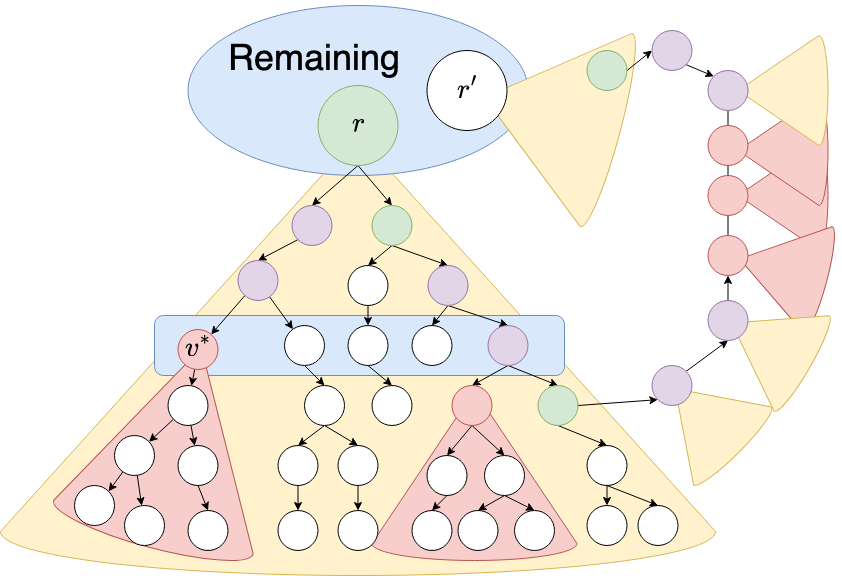}
    \caption{Illustrating the Promote procedure for $\ell = 2$, the purple nodes in the path from $r$ to $v^*$ are the compress layer $V^C_{i,p}$. Red nodes are nodes that can already terminate, and since $v^*$ is now a local maximum, it can choose a label and propagate its choice down.}
    \label{fig:afterPromotion}
\end{figure}

Now the algorithm that performs our modified $(\gamma, \ell, L)$-decomposition, just consists of calling these 3 subroutines in every iteration.
Note that \Cref{alg:Decomposition} provides a description of the steps of the algorithm without specifying how the algorithm is implemented in a distributed manner. We will give the distributed implementation, and the details of when nodes terminate in \Cref{sec:dectoave}.
\begin{algorithm2e}
\caption{Compute Decomposition}\label{alg:Decomposition}
\KwIn{$G=(V,E), \Pi$} 
\SetKwFunction{FOrientingRakes}{orientingRake}
\SetKwFunction{FCompressWithSlack}{compressWithSlack}
\SetKwFunction{FPromoteIfPossible}{promoteIfPossible}
compute $\ell$ from $\Pi$ \label{alg:line:one}\\
$b \gets \ell + 2$ \\ 
$\gamma \gets \ell + 3$ \\
\FOrientingRakes{$G$, $1$, $\gamma$}\\
$i \gets 2$ \\
\For{$O(\log n)$ times (until every node is assigned to a layer) \label{alg:line:startout}}{
    \FCompressWithSlack{$G$, $i-1$, $\ell$}\\
    \FOrientingRakes{$G$, $i$, $\gamma$} \Comment{Removing the Slack from Compress}\\
    \FPromoteIfPossible{$G$, $i$, b}\\
    $i \gets i + 1$
}
\end{algorithm2e}

We now want to show that \cref{alg:Decomposition} actually computes a valid $(\gamma, \ell, L)$-decomposition.

We note that at all times $V_i^R$ is to be understood as the union of all nodes in the layers $V_{i,j}^R$, i.e.,
\[
V_i^R := \bigcup_{1 \leq j \leq \gamma} V_{i,j}^R.
\]

We now prove for every subroutine that they produce valid partial $(\gamma, \ell, L)$-decompositions. The only tricky part, is that inside of the promotion subroutine we put the promoted node in layer $V^R_{i+1,1}$, while the rest of the rake nodes so far are only in layer $V^R_{i,j}$. We essentially put $v^*$ in the layer that will be computed in the next iteration. As a result we need something slightly stronger than just a pre existing partial decomposition.

\begin{lemma}\label{lem:compressCorrect}
Assume we are given a graph $G$ and a valid partial $(\gamma, \ell, i-1)$-decomposition. Then after running \cref{alg:Compress} \emph{Compress with Slack} with parameter $(G,i-1, \ell)$ we have a valid partial $(\gamma, \ell, i)$-decomposition. All nodes in $N^{(i)}$ are either a local maximum themselves, or are between two local maximums that are at distance at most $2\ell$. Furthermore it holds that for all free nodes $v$ all neighbors of $v$ are in a layer strictly smaller than $V^R_{i,1}$.
\end{lemma}
\begin{proof}
The nodes that are put into $V^C_{i-1}$ satisfy the properties of a $(\gamma, \ell, i)$-decomposition by construction.
All of the nodes in $Z$ are local maxima and all other nodes in $\{v_w, v_{w+1}, \dots , v_{y-1}, v_y\} \setminus Z$ are in paths of length $[\ell,2\ell]$ that are enclosed by nodes in $Z$. The same holds for the set $N^{(i-1)}$ by induction.
Lastly the only nodes that are put into $V^R_{i,1}$ are the nodes in $Z$ which are enclosed by lower layer nodes, so the furthermore part holds. 
\end{proof}

\begin{lemma}\label{lem:rakeCorrect}
Assume we are given a graph $G$ and a valid partial $(\gamma, \ell, i)$-decomposition, where free nodes are only adjacent to nodes in layers that are strictly smaller than $V^R_{i,1}$. Then after running \cref{alg:Rake} \emph{Orienting Rake} with parameter $(G,i, \gamma)$ we have a valid partial $(\gamma, \ell, i)$-decomposition.
\end{lemma}
\begin{proof}
We only add nodes to rake layers $V^R_i = V^R_{i,1}\cup \dots \cup V^R_{i,\gamma}$, so we only have to ensure that each such sublayer consists only of isolated nodes and that each such node has at most one neighbor in a higher layer, or free.\\
Now initially all active nodes are adjacent to layers that are strictly smaller than $V^R_{i,1}$ so no active node that we put in $V^R_{i,j}$ for any $1\leq j\leq \gamma$ can be adjacent to a node of the same sublayer. As a result of this and because we only add degree 0 or 1 nodes, all of these layers $V^R_{i,1},\dots,V^R_{i,\gamma}$ consist of only isolated nodes. Also they can only have one neighbor that is free or in a higher layer, because they had degree 0 or 1 when they were added.
\end{proof}

We will later need that regular compress paths are far away from active nodes. This will help us make a distinction between paths that are created as a result of the normal compress procedure and paths that are created in the promotion procedure.
\begin{lemma}\label{lem:compressNodesAreDeep}
Assume we are given a graph $G$ and a valid partial $(\gamma, \ell, i-1)$-decomposition. Then after running \cref{alg:Compress} \emph{Compress with Slack} with parameter $(G,i-1, \ell)$ and afterwards running \cref{alg:Rake} \emph{Orienting Rake} with parameter $(G,i, \gamma)$, every free node $v$ it has distance at least $\gamma$ from any node in layer $V^C_{i-1}$.
\end{lemma}
\begin{proof}
This follows directly from the slack that was left in the compress subroutine. Consider $G'$ to be the graph induced by active nodes after the compress operation. Then because of the way the compress subroutine works any active node $u \in G'$ that is adjacent to a node in $V^C_{i-1}$, is contained in a path $P_u$ of length exactly $\gamma$ nodes. Furthermore $u$ has degree exactly 1 in $G'$ and all other nodes in $P_u$ have degree exactly 2. As a result after running the rake operation where nodes of degree 1 are removed $\gamma$ times, the entire path $P_u$ gets assigned to layer $V^R_i$. Clearly $u$ now has distance $\gamma -1$ from any free node $v$ and by extension the neighbor of $u$ that is in layer $V^C_{i-1}$ now has distance $\gamma$ to any free node.
\end{proof}

Before we prove that also the promotion procedure produces a valid partial decomposition, we make an observation about orientations.
\begin{lemma}\label{lem:CloseNodesAreDescendants}
For all iterations $i$, consider a free node $r$ and any node already assigned node $u$ at distance at most $b$. If the path $P$ from $r$ to $u$ consists only of assigned nodes, then it is consistently oriented from $r$ to $u$.
\end{lemma}
\begin{proof}
    Let $v$ be a node that was already assigned a layer and let $P$ be a path from $r$ to $v$ of length at most $b$. We need to show, that all edges are oriented from $r$ to $v$.
    By \Cref{lem:compressNodesAreDeep} we get that none of these nodes were assigned a layer inside the compress operation. As a result all of them must have been assigned a layer in a rake operation. An edge oriented in \cref{alg:Rake} is always oriented towards a free node. As a result all of the edges in $P$ must have some orientation.\\
    Lets call nodes without incoming edges orphans.
    Clearly every node that is not an orphan has a path that is oriented from an orphan to itself, since orientations are only ever made towards orphans. Then since $r$ is an orphan, $\{r, v_0\}$ must therefore be oriented towards $v_0$. \\
    Now assume for contradiction that the edges inside $P$ are not consistently oriented from $r$ to $v$, then there must exist a first edge $\{v_j,v_{j+1}\}$ that is oriented in the other direction, so from $v_{j+1}$ to $v_j$. But now $v_{j+1}$ would have two edges oriented away from it, which cannot happen.
\end{proof}

This now allows us to prove the correctness of the promotion procedure.
\begin{lemma}\label{lem:promoteCorrect}
Assume we are given a graph $G$ and a valid partial $(\gamma, \ell, i)$-decomposition. Assume that there are no $V^C_{i-1}$ nodes at distance $\gamma$ from any free node. Then after running \cref{alg:Promote} with parameter $(G,i, b = \gamma -1)$ we have a valid partial $(\gamma, \ell, i+1)$-decomposition. Additionally $v^*$ is a local maximum and free nodes are only adjacent to nodes in layers that are strictly smaller than $V^R_{i,1}$.
\end{lemma}
\begin{proof}
Because we have a valid partial $(\gamma, \ell, i)$-decomposition, no node is in layer $V^C_{i,p}$ and the highest layer so far is $V^R_{i,\gamma}$.
Consider the path $P$, first because of the if condition, all nodes in $P$ are in layers of the form $V^R_a$ for $1 \leq a \leq i$. Analogous to the argumentation in \cref{lem:CloseNodesAreDescendants} and since $r$ is a free node, the nodes in $P$ have strictly increasing layers from $v^*$ to $r$.
The nodes in $P' := P\setminus \{v^*,r\}$ are added to $V^C_{i,p}$ and $v^*$ is added to $V^R_{i+1,1}$. \\
We first argue, that the reassigned nodes satisfy the conditions of a $(\gamma, \ell, i+1)$-decomposition and then we argue that we did not disrupt any of the other layers. The component of layer $V^R_{i+1,1}$ created here, consists only of $v^*$ and all of its neighbors are in smaller layers, since no node had a layer bigger or equal then $V^R_{i+1,p}$ before. So $V^R_{i+1,1}$ satisfies the conditions for a $(\gamma, \ell, i)$-decomposition and $v^*$ is now a local maximum.\\
Furthermore consider the nodes in $P$, they form a path of length exactly $b = \ell +2$. Now as a result the path $P'$ without $r$ and $v^*$ is a path of length exactly $\ell$. At one of its endpoints is $v^*$, which is in a strictly higher layer and at the other endpoint is $r$ which is still an active node, but will be assigned a strictly higher layer in the future. Again all other nodes adjacent to $P'$ have strictly smaller layers, since $V^R_{i, \gamma} < V^C_{i,p}$. So the nodes in $P'$ for am valid compress layer.\\
Now consider any node $u$ adjacent to a node in $P\setminus \{r\}$, we will argue that the layer assignment for $u$ still satisfies all of the necessary conditions for a $(\gamma, \ell, i)$-decomposition. Let $w \in P\setminus \{r\}$ be a neighbor of $u$ in the path, since $w$ has a rake label, it cannot be in a compress path together with $u$. As a result $u$ is either a rake node itself, or is an endpoint of a compress path. In either case $u$ must have at most one neighbor of a higher layer. Since the path $P$ had strictly increasing layers from $v^*$ to $r$, $w$ already has its unique higher layer neighbor in $P$. Since $w$ cannot have two neighbors in a higher layer, $w$ is the unique higher layer neighbor of $u$. Since $w$ was assigned a layer that is higher than its previous layer, $w$ is still a neighbor with a higher layer. So for all nodes that were not reassigned a new layer, they still fulfill all of the necessary conditions of a $(\gamma, \ell, i)$-decomposition.
\end{proof}

We get as a corollary of \cref{lem:compressCorrect,lem:rakeCorrect,lem:promoteCorrect} that our algorithm always maintains a valid partial decomposition.

\begin{corollary}\label{cor:itispartdec}
In \cref{alg:Decomposition} for every iteration $i$, after every execution of a subroutine we have a valid partial $(\gamma, \ell, i)$-decomposition, or a valid $(\gamma, \ell, i+1)$-decomposition.
\end{corollary}

Next we argue that our algorithm behaves in a similar way to a normal rake and compress algorithm. We use the following lemma.
\begin{lemma}[\cite{chang20}]\label{lem:rcremaining}
Given a tree with $n$ nodes, by performing $\alpha$ rakes\footnote{A rake operation is the removal of all nodes of degree $0$ or $1$.} and $1$ compress\footnote{A compress operation consists of the removal of all paths of nodes of degree exactly 2, that are of length at least $\beta$.} with minimum path length $\beta$, the number of remaining nodes is at most $\frac{\beta}{2 \alpha} n$.
\end{lemma}
Now to show that our algorithm fulfills the conditions for this lemma every few iterations.
\begin{lemma}\label{lem:epsilon}
There exists some constant $0< \varepsilon < 1$ such that, for any positive integer $i \ge 5$, the number of free nodes after iteration $i$ of \Cref{alg:Decomposition} is at most  $n\varepsilon^{i}$.
\end{lemma}
\begin{proof}
    We argue that after 5 iterations of our algorithm at least one ''normal'' compress with $\beta = 4\ell + 9$ and $\alpha = 4\ell + 12$ rakes are performed.
    In our algorithm, the minimum path length is $\beta = 4\ell + 9$, but not all of these nodes are immediately put into a compress layer. Instead we leave $\gamma$ nodes at both ends. However all of these nodes that were left are immediately raked away with the next rake procedure. So in one iteration all nodes in paths of length at least $\beta$ are assigned a layer and therefore a compress is performed. To keep the argumentation clean, we will consider this iteration to be used exclusively for the compress.
    In the following 4 iterations, completely ignore the compress procedure and simply focus on the fact that we perform $\gamma$ rakes per iteration for a total of $\alpha = 4\gamma = 4\ell +12$.\\
    As a result, by \cref{lem:rcremaining} after 5 iterations the number of remaining nodes decreases by a factor of 
    \[
    \frac{\beta}{2\alpha} = \frac{4\ell +9}{2(4\ell +12)} < \frac{1}{2}
    \]
    Hence, at iteration $i$, for $i$ integer multiple of $5$, we get that the number of remaining nodes is at most $n / 2^{(i/5)}$. In order to obtain a smooth progression, we pick $\varepsilon = \frac{1}{\sqrt[10]{2}}$. Observe that, for all $i \ge 5$, $n \epsilon^i \ge n / 2^{(i/5)}$, and hence the claim follows.
\end{proof}

As a result of this \Cref{lem:epsilon} we immediately get that we will be done in at most $O(\log n)$ rounds and because of \Cref{cor:itispartdec} we also get that we will have a complete $(\gamma, \ell, O(\log n))$-decomposition at the end.

\begin{corollary}\label{cor:inlog}
    The decomposition (i.e., the assignment of nodes to layers) produced by \Cref{alg:Decomposition} is a $(\gamma, \ell, L)$-decomposition, where $\gamma$ and $\ell$ are the parameters from \Cref{alg:Decomposition} and $L$ is a suitably chosen value from $O(\log n)$.
\end{corollary}

In \Cref{sec:dectoave} we will then show that local maxima allow us to have nodes terminate early, such that we will obtain a good node-averaged complexity at the end. In the next section our main goal is to prove that enough of these local maxima actually exist and are nicely distributed in the graph.

\subsection{Local Maxima and Bounding Quality}\label{sec:insights}
We will first see that during the execution of our algorithm we can actually decompose the graph into two parts. First the free nodes and all already raked nodes hanging from them and second the nodes that are in compress paths who also have raked nodes hanging from them see \Cref{fig:structure} from before. In compress paths there already are \emph{natural} local maxima just from the way the original decomposition algorithm works. These local maxima are sufficient to have most of the nodes in compress paths terminate. The only exception are the ends of the compress paths. These ends are oriented such that they are charged to the quality of some remaining free node. We will see that by summing over the quality of all free nodes we get exactly all of these nodes that are not yet descendants of a local maximum.

\smallskip

We start with the following definitions with regards to the orientation of the nodes. Intuitively nodes are oriented, such that if a node is raked away, the edge is oriented from a parent towards the just removed node and the ends of compress paths are oriented inwards. This is illustrated in \Cref{fig:orientation}. Let $i$ be some arbitrary iteration and consider the input tree $G$ together with the partial assignment of nodes to layers at the end of the iteration. Let $G^{(i)}$ denote the set of nodes that have not yet been assigned any layer.
\begin{figure}[h!]
    \centering
    \includegraphics[width=.6\textwidth]{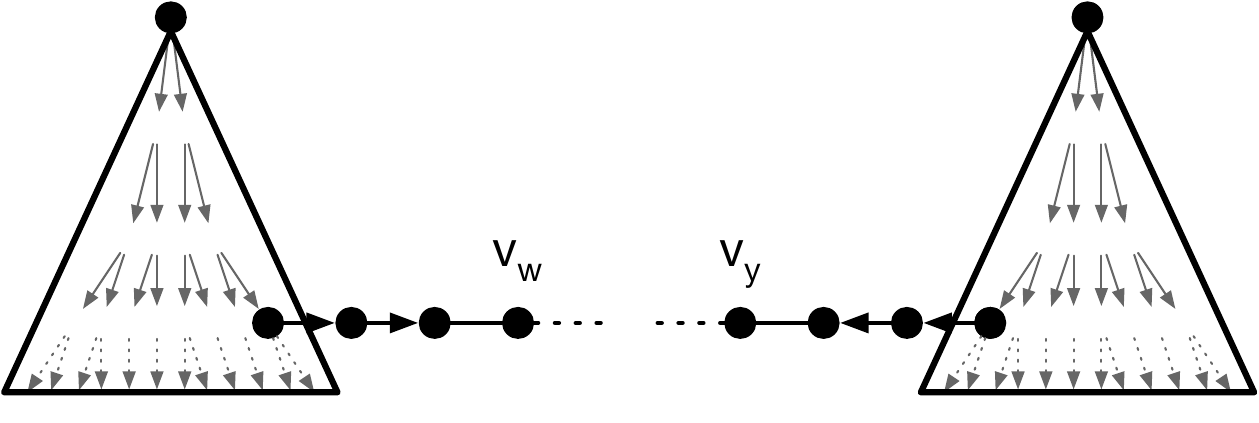}
    \caption{An example of how edges are oriented. The two nodes at the top are free nodes and they are connected by a compress path. The nodes inside of the respective trees are oriented from the root towards the leaves. The ends of the compress path are oriented inwards, however the middle part of the compress path is not oriented in any way, as illustrated at the innermost nodes, the incident edges of which are not oriented.}
    \label{fig:orientation}
\end{figure}
Again, let $G'$ denote the subgraph of $G$ consisting of all nodes that have already been assigned a layer. We can express $G'= G\backslash G^{(i)}$.

\begin{definition}[child, parent, descendant, ancestor, orphan] \label{def:OrientationNames}
    For any edge $\{ w, w'\}$ oriented from $w$ to $w'$, we call $w'$ a \emph{child} of $w$ and $w$ the {parent} of $w'$.
    For any oriented path $(w, \dots, w')$ that is consistently oriented from $w$ to $w'$, we call $w'$ a \emph{descendant} of $w$ and $w$ an \emph{ancestor} of $w'$. We call a node with no edges oriented towards itself an \emph{orphan}.
\end{definition}

We give one of the most important definitions, that of a subtree of assigned nodes. We illustrate how they exist in relation to $G^{(i)}$ in \Cref{fig:Subtrees}.

\begin{figure}[h!]
    \centering
    \includegraphics[width=.5\textwidth]{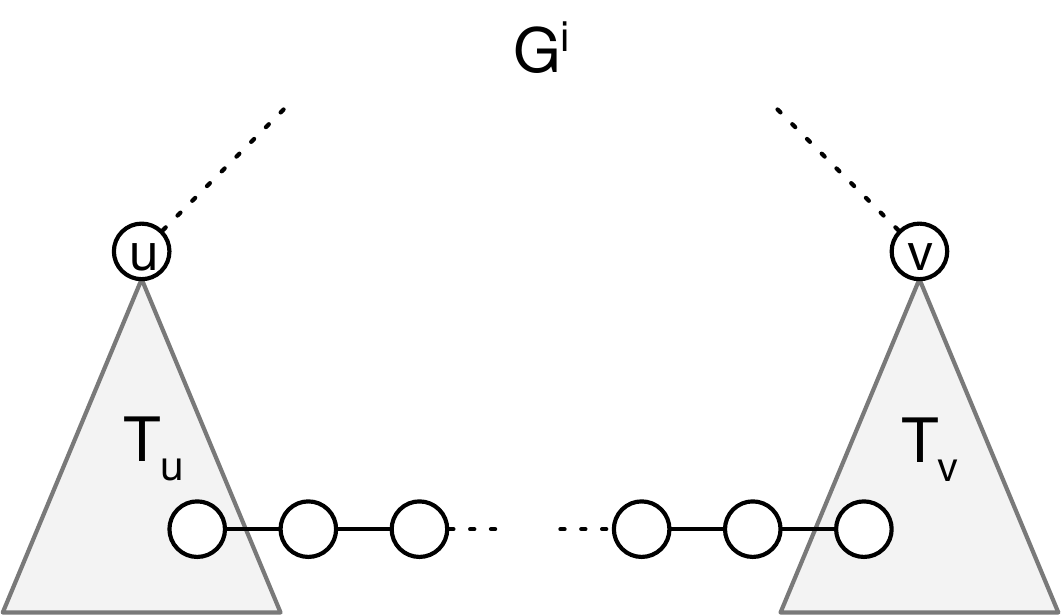}
    \caption{The graph induced by $G^{(i)}$ is not necessarily connected, it might happen that some two nodes $u \in G^{(i)}$ and $v \in G^{(i)}$ are both having a subtree of assigned nodes hanging from them. In this illustration we name them $T_u$ and $T_v$ respectively. These trees are then necessarily connected by a compress path.}
    \label{fig:Subtrees}
\end{figure}

\begin{definition}[subtree of assigned nodes]\label{def:subtreeAssigned}
     For any node $v \in V(G)$ and any positive integer $i$ let , the subtree of assigned nodes $T^{(i)}(v)$ denote the set that contains $v$ and all nodes $w$ that can be reached from $v$ via a path $(v=v_0, v_1, \dots, v_j = w)$ with the following property: the edge $\{ v_{a-1}, v_a \}$ is oriented from $v_{a-1}$ to $v_a$, for each $1 \leq a \leq j$.
\end{definition} 
We denote by $h^{(i)}(v)$ the height of $T^{(i)}(v)$, formally
\[
    h^{(i)}(v) = \max \{\dist(v,u)\}_{u \in T^{(i)}(v)}.
\]

To make the definition of $T^{(i)}(v)$ a bit clearer, we are going to make some small observations about these trees. 
\begin{observation} \label{obs:Orientations}
The following hold:
    \begin{enumerate}
        \item \textbf{Orientations} of the edges are only set once and then will not be changed anymore. This holds true simply because edges are only oriented when a node is first assigned a layer and not when layers are changed later. Additionally each node only ever orients one edge toward itself. For an illustration of the orientation refer to \Cref{fig:orientation}
        \item \textbf{Compress Paths}: For a path $P'$ that is handled in \cref{alg:Compress}, only the first and last few nodes of the path are oriented. As a result the nodes between the first and last nodes  $v_w$ and $v_y$ only have incident edges that are oriented away from them (and unoriented incident edges). Therefore all of the nodes on the unique path $v_w, \dots, v_y$ are orphans.
        \item \textbf{Promoted Compress Paths}: For a path that has its layers changed inside of \cref{alg:Promote}, the orientation stays unchanged.
         \item \textbf{The Set} $N^{(i)}$ defined in Line \ref{alg:line:assignN} of the compress procedure only contains nodes that have no edge oriented towards them.
    \end{enumerate}
\end{observation}

We now show how we can express the entire input graph $G$ in terms of these trees. For this we will make use of the sets $N^{(i)}$, defined in the compress procedure in \cref{alg:line:assignN}. This set contains exactly those nodes which were assigned a layer, but which can not be reached through oriented edges.  Intuitively any node that is already assigned a layer has to hang in some subtree of assigned nodes. By taking a union over all of the subtrees hanging from $G^{(i)}$ and $N^{(i)}$ we will cover the entire tree.

\begin{lemma} \label{lem:decomposeV}
    For each positive integer $i$, the following holds: 
    \[
        V(G) = \left( \bigcup_{v \in G^{(i)}}T^{(i)}(v)\right) \cup \left(\bigcup_{v \in N^{(i)}}T^{(i)}(v)\right)
    \]
    Furthermore the two big unions are disjoint.
\end{lemma}
\begin{proof}
    Now any node that has an edge oriented towards it, is clearly hanging in a tree rooted at an orphan node. So as long as we have all trees of nodes that are orphans, we cover the entirety of $V(G)$. The only nodes that are orphans are the nodes of $G^{(i)}$ and the nodes in $N^{(i)}$, since all other nodes had an edge oriented towards them, when they were assigned to a layer.\\
    Since $v$ is in its own subtree of assigned nodes we trivially get that all nodes of $G^{(i)}$ are in $\bigcup_{v \in G^{(i)}}T^{(i)}(v)$. And the same way we get that also all nodes of $N$ are in $\bigcup_{v \in N^{(i)}}T^{(i)}(v)$.\\
    We see that the two are disjoint, because $G^{(i)}$ and $N^{(i)}$ are disjoint. Furthermore any node has at most one edge oriented towards itself and it therefore also only belongs to one subtree of assigned nodes, rooted at an orphan.
\end{proof}

Notice that for each of these subtrees of assigned nodes, the quality of the root counts exactly how many nodes in this tree still need to terminate. So by giving a good upper bound on the quality of nodes in $G^{(i)}$ we will be able to show that in each iteration a constant fraction of the remaining nodes terminate. \\

We first prove that if the statement in Line~\ref{alg:line:conditionFixV*} is true and we promote $v^*$ to layer $V^R_{i+1,1}$ we get that $q(v^*)$  will be a large part of $q(r)$. Thereby showing that with each promotion we reduce the quality by a constant fraction. Recall the definitions of the quality $q(v)$ of a node $v$ and $H(v)$, provided in \Cref{def:quality}. Also we will omit some of the $^{(i)}$, where the iteration is clear, to keep things cleaner.

\begin{lemma}\label{lem:fixV}
    If the condition in Line \ref{alg:line:conditionFixV*} holds and $v^*$ is promoted, then $q(v^*) \geq \frac{q(r)}{2\Delta^b}$ when Line~\ref{alg:line:conditionFixV*} is evaluated.
\end{lemma}
\begin{proof}
    The statement follows from the fact that
    \[
    q(r) \leq \Delta^b + \sum_{v \in C_b(r)} q(v)
    \]
    This is true, because we can separate $H(r)$ into the nodes that are close and those that are far. Concretely a node $u$ is close, if the path to $r$ is strictly less than $b$. In this case we get by \Cref{lem:CloseNodesAreDescendants} that all of these close nodes are actually descendants of $r$. But only $\Delta^b$ such nodes can exist. A node $u$ is far, if the path to $r$ is at least $b$ nodes long, at which point it has to pass through one of the nodes $v \in C_b(r)$. Now since $u \in H(r)$ the unique path from $r$ to $u$ satisfies the criteria for $u$ to be in $H(r)$, then the subpath from $w$ to $u$ must also satisify these criteria and $u$ is therefore included in $q(w)$.
    We then get
    \[
    q(r) \leq \Delta^b + \sum_{v \in C_b(r)} q(v) \leq \Delta^b + |C_b(r)| \cdot q(v^*) \leq \Delta^b + \Delta^b q(v^*) \leq 2\Delta^b q(v^*).
    \]
    The second inequality comes from the fact $v^*$ has the highest quality among nodes in $C_b(r)$.
\end{proof}

\begin{figure}[h!]
    \centering
    \includegraphics[width=.5\textwidth]{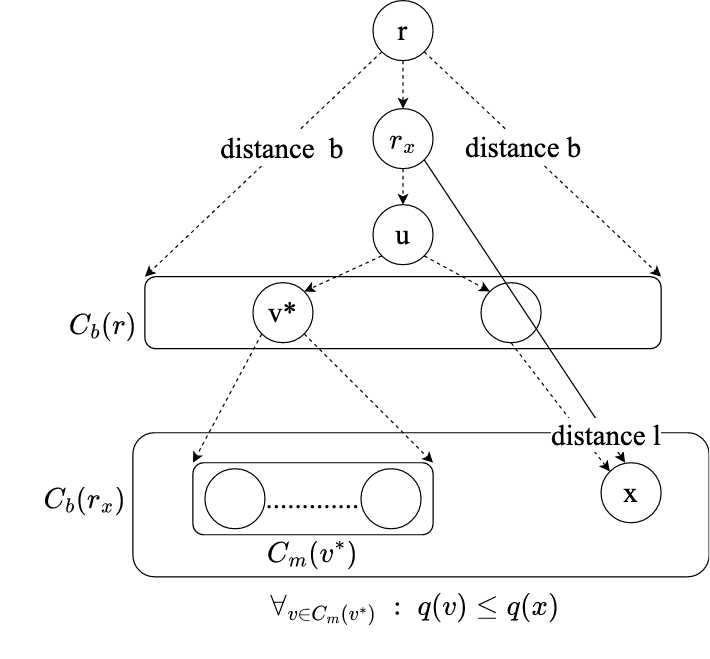}
    \caption{$v^*$ cannot be fixed, because $x$ is already fixed. $u$ is the common ancestor and $r_x$ was root, when $x$ was fixed. Clearly the way $m$ is chosen, $C_m(v^*)$ is a subset of $C_b(u)$ \gs{fix distance arrow towards x to say $b$}}
    \label{fig:obstructed}
\end{figure}
However the condition in Line \ref{alg:line:conditionFixV*} may not hold in every iteration. So if it does not hold, we show that there must exist some node $x$ promoted (earlier) in Line~\ref{alg:line:markPromoted} that is close to $v^*$. 
\begin{lemma}\label{lem:fixX}
    If the condition in Line \ref{alg:line:conditionFixV*} of the promotion Procedure \cref{alg:Promote} does not hold, then there exists a promoted node $x$ at distance at most $2b-1$ from $r$, such that $q(x) \geq \frac{q(v^*)}{2\Delta^b}$ immediately before $x$ was promoted.
\end{lemma}
\begin{proof}
    As the intuition of this proof is a lot easier to grasp visually we have provided a sketch of this scenario in \Cref{fig:obstructed}.
    Consider the path $P$ from $r$ to $v^*$.
    Since the condition in Line \ref{alg:line:conditionFixV*} is not true, some node $u \in V(P) \setminus \{ r \}$ must already be in a compress layer (i.e., a layer of the form $V_a^C$), with $\dist(u,v^*) < b$. By \cref{lem:compressNodesAreDeep} and the fact that $\gamma = b + 1$, we have that this compress node $u$ must have been added as the result of a previous promotion procedure, hence there must also be a node $x$ that was actually promoted. Now because of \Cref{lem:CloseNodesAreDescendants} $u$ is an ancestor of both $v^*$ and $x$. Then we have that $\text{dist}(u,x) \leq b$ and $\text{dist}(u,v^*) < b$. Now let $r_x$ be the node that was the free node which decided to promote $x$, so $\text{dist}(r_x, x) = b$, then $\text{dist}(r_x,v^*) < b$.
    Now define $m = d(r_x,x) - d(r_x,v^*) < b$, then we get that $C_m(v^*) \subset C_b(r_x)$ and since $x$ was the choice of $r_x$, we get that 
    \[
    \forall_{v \in C_m(v^*)} \; q(v) \leq q(x)
    \]
    Now to bound $q(v^*)$
    \[
    q(v^*) \leq \Delta^m + \sum_{v \in C_m(v^*)} q(v) \leq \Delta^b + \Delta^b q(x) \leq 2\Delta^b q(x)
    \]
\end{proof}

So the idea is, that during iteration $i$ even though we might fail to promote during iteration $i$, this must mean that we actually were successful recently.

\subparagraph{Upperbounding the quality of a free node}\\
The next lemma will be the main technical result that shows that enough nodes are in subtrees of local maxima. We will show this implicitly by upperbounding the quality of the remaining free nodes. However we will have to introduce some more notation.
We are partitioning each assigned subtree $T^{(i)}(v)$ into $\alpha = \myceil{\frac{h^{(i)}(v) + 1}{b}}$ subsets $S^{(i)}_0(v) , \dots , S^{(i)}_{\alpha - 1}(v)$, where, for each $0 \leq j \leq \alpha - 1$,
\[
S^{(i)}_j(v) := \{u \in T^{(i)}(v) \; | \; b \cdot j\leq \text{dist}(u,v) < b \cdot (j + 1)\}.
\]
So every $b$ layers of the tree are grouped into one such set. This is illustrated in \Cref{fig:layerSets}.

\begin{figure}[h!]
    \centering
    \includegraphics[width=.4\textwidth]{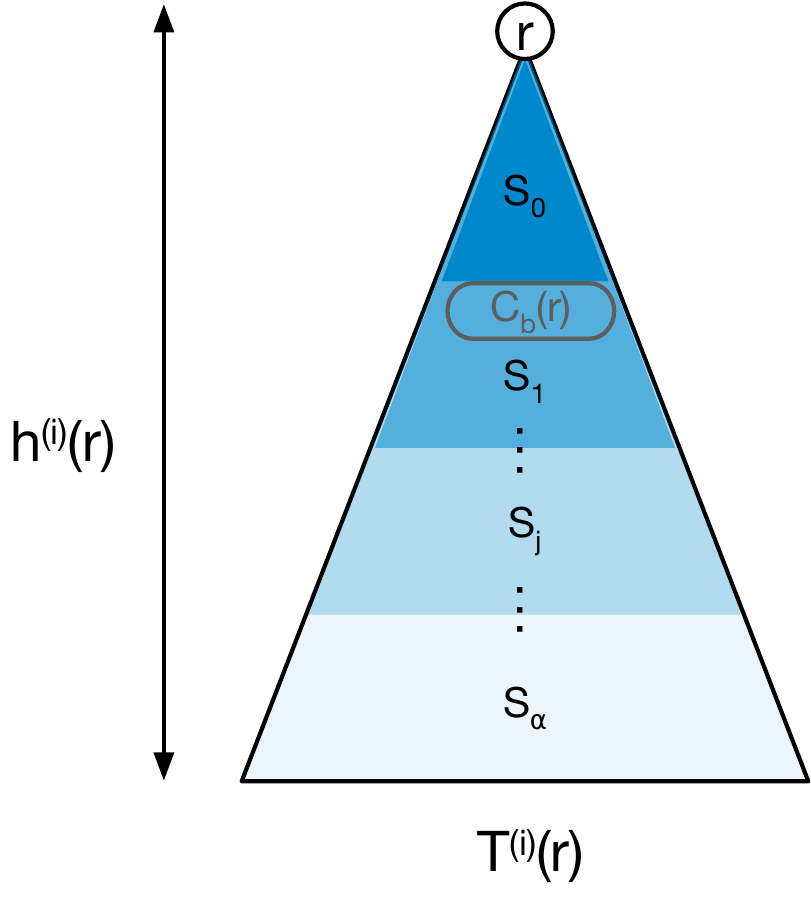}
    \caption{A tree rooted at $r$ is split into the different layer sets $S_0, \dots, S_\alpha$. Since $r$ in included in $S_0$ the union of those layer sets will form the entire tree.}
    \label{fig:layerSets}
    \end{figure}

Clearly,
\[
\bigcup_{0 \leq j \leq \alpha - 1} S^{(i)}_j(v) = T^{(i)}(v).
\]
One useful equality following from these definitions is
\[
\bigcup_{v \in C_b(r)} S^{(i)}_j(v) = S^{(i)}_{j+1}(r),
\]
which holds for any assigned node $r$ and any $1 \leq j \leq \alpha - 1 = \myceil{\frac{h^{(i)}(r) + 1}{b}}$.

For convenience, we also formally define $S^{(i)}_j(v)$ for indices that are larger than $\alpha - 1$: we set $S^{(i)}_j(v) := \emptyset$ for all $j \geq \alpha$.

Now we are ready to give the main result about the quality of free nodes. Essentially we want to look at assigned subtrees that are hanging from free nodes. The main intuition is that the deeper layers contain more nodes and also contain more local maximums. This is simply as a result of them being in the removed part for more iterations.\\

If we were able to proof that for each layer $j$ it was true, that only $\lambda^j |S^{(i)}_j|$ nodes remain in that layer for some constant fraction $\lambda$, then this would be enough for constant node averaged complexity. However such a statement is simply not true. This is because nodes are chosen based on their quality and this might result in some layers that have almost no nodes fixed. But then this implies that these layers are very sparse. What we will instead be able to show is that if we sum over all of the layers of the tree, then a similar statement is actually true.

\begin{lemma} \label{lem:layersPartiallyFixed}
    There exists a constant $0 < \lambda < 1$ (that only depends on $\Pi$ and $\Delta$) such that for all even positive integers $i$, the following inequality holds at the end of iteration $i$, for all nodes $r \in G^{(i)}$:
    \[
    q^{(i)}(r) \leq \sum_{j = 0}^{\myceil{(h^{(i)}(r) + 1) / b} - 1} \lambda^{j}|S^{(i)}_j(r)|
    \]
\end{lemma}
We are going to prove \Cref{lem:layersPartiallyFixed} by induction on the height $h^{(i)}(r)$ of the subtree of assigned nodes. The idea is to use the induction hypothesis on the nodes in $C_b(r)$. Then by using the equality 
\[
\bigcup_{v \in C_b(r)} S^{(i)}_j(v) = S^{(i)}_{j+1}(r),
\]
we get a first good bound on $q^{(i)}(r)$. However our induction hypothesis holds only for free nodes and the nodes we want to use our induction hypothesis on are already assigned. As we will see in the following lemma it takes some work to get a statement about already assigned nodes from it. First we need to show that in each iteration trees don't grow too much. Notice the $(G^{(i)} \setminus G^{(i+1)})$ is the set of nodes that gets removed in the next iteration. Intuitively the next lemma says that the nodes that were added during the last iteration are not deep in the tree.

\begin{lemma} \label{lem:distanceOfJustRemoved}
For any free node $w$ for all iterations $i>1$. For all nodes $v \in (G^{(i)} \setminus G^{(i+1)}) \cap T^{(i+1)}(w)$, $v$ has distance less than $3b$ from $w$.
\end{lemma}
\begin{proof}
Let $R^{(i+1)} := G^{(i)} \setminus G^{(i+1)}$ be the set of all nodes that were first assigned a layer in iteration $i+1$. 
Let $v \in R^{(i+1)} \cap T^{(i+1)}(w)$. As a result there must be an oriented path from $w$ to $v$. Since $v$ was free before iteration $i+2$ and therefore was also an orphan at that point in time, such a path must consist completely of nodes in $R^{(i+1)}$. 
Hence, by giving an upper bound on the longest such oriented path in $R^{(i+1)}$, we can prove the lemma.
        
At the beginning of iteration $i+1$ all nodes in $R^{(i+1)}$ are free. Edges are only oriented in the rake procedure and in the compress procedure. After the orientations in the compress procedure in Lines \ref{alg:line:orient2.1} and \ref{alg:line:orient2.2} the longest oriented path in $R^{(i+1)}$ has length at most the length of one of those paths. Since components in $V(P') \setminus Z$ have length in $[\ell, 2\ell]$ these paths have at most length $2\ell$. So after these lines any oriented path consisting of nodes from $R^{(i+1)}$ has length at most $2\ell$.
What remains is the rake procedure that is called afterwards. Consider some oriented path of nodes in $R^{(i+1)}$ of length at most $2\ell$.
In each execution of Line \ref{alg:line:orient3} we extend this path by at most 1. So after the $\gamma$ iterations of this line, the longest oriented path consisting of nodes from $R^{(i+1)}$ has length at most $2\ell + \gamma = 2\ell + \ell +3 = 3\ell +3 < 3\ell +6 = 3b$.
\end{proof}

Now we are ready to bound the quality of an already assigned node. In the proof \Cref{lem:layersPartiallyFixed}  we will do an inductive argument so for now lets assume it already is true. Intuitively what we will do is ignore all of the nodes that were added in the last iteration (corresponding to the first three layer sets) and then get a bound on everything underneath by applying  \Cref{lem:layersPartiallyFixed}.

\begin{lemma}\label{lem:boundAssignedNode}
Suppose \Cref{lem:layersPartiallyFixed} holds. Suppose $w$ is some already assigned node during iteration $i$. Let $k<i$ be the last iteration in which $w$ was still a free node. Then 
\[
    q^{(i)}(w) \leq q^{(k+1)}(w) 
\]
\[
  \leq |S_0^{(k+1)}(w)| + |S_1^{(k+1)}(w)| + |S_2^{(k+1)}(w)| + \sum_{j = 3}^{\myceil{(h^{(k+1)}(w) + 1) / b} - 1} \lambda^{j-3}|S_j^{(k+1)}(w)|.
\]
\end{lemma}
\begin{proof}
Critically we have to specify in which iteration $k$ we apply \Cref{lem:layersPartiallyFixed}, in which $w$ was still a free node. The problem is that if we derive some bound on $q^{(k)}(w)$ for some iteration $k$ in which $w$ was still free, then in iteration $k+1$ some neighbor $v$ of $w$ might be raked away and $T^{(k)}(v)$ might become part of $T^{(k+1)}(w)$.
As a result $T^{(k+1)}(w)$ might be larger than $T^{(k)}(w)$ and therefore invalidate our bound on $q^{(k)}(w)$, since $q^{(k+1)}(w) > q^{(k)}(w)$ might be true. The solution to this is the fact that once $w$ is assigned a layer, the subtree of assigned nodes hanging from $w$ can no longer grow\footnote{This is true, because $w$ will no longer be a free node and edges are only ever oriented from free nodes towards nodes that were just assigned a layer, so no new paths from $w$ to other nodes can be created, after $w$ was assigned a layer.} and therefore the quality of $w$ can no longer increase. As a result a bound on the quality that holds when $w$ is assigned a layer, will always stay true afterwards. So we choose $k$ to be last iteration in which $w$ was still a free node and show that in that last iteration the quality didn't grow too much.

If $w$ is assigned in iteration $k = 1$ or before, then by \Cref{lem:distanceOfJustRemoved} the subtree hanging from $w$ must be of height strictly less than $3b$. Therefore the quality is bounded by
\[
q^{(i)}(w) = |S^{(i)}_0(w)| + |S^{(i)}_1(w)| + |S^{(i)}_2(w)|.
\]
As this is a stronger statement than what we will derive later, we can freely assume that $w$ was not assigned in the first iteration.
Let iteration $k$ be the last iteration such that $w$ was a free node at the end of iteration $k$.
This in particular implies that $k \leq i-1$ and that $w$ is assigned a layer in iteration $k+1 \leq i$. So if we give a bound on $q^{(k+1)}(w)$ then this bound still holds in iteration $i$, i.e., for $q^{(i)}(w)$.
(Note that, once a node $w$ is assigned, the definition of $q(w)$ ensures that the quality of $w$ cannot increase since the reassignments of assigned nodes to new layers in \Cref{alg:Decomposition} always assign the nodes to higher layers than used before.)
    
To give a bound on $q^{(k+1)}(w)$ we first need to understand how $T^{(k+1)}(w)$ looks like. Clearly $T^{(k)}(w) \subseteq T^{(k+1)}(w)$, but there might also be some additional nodes that were free nodes after iteration $k$ that now have joined $w$'s subtree of assigned nodes. These nodes then also come with their own subtrees. As a result these nodes would now also be in $T^{(k+1)}(w)$. Let us call the set of these nodes 
\[
V^{(k)}_w := \{v \in G^{(k)}\; | \;v \in T^{(k+1)}(w) \}.
\]
We note that $v \in T^{(k+1)}(w)$ implies $ v \notin G^{(k+1)}$, just from the definition of subtrees of assigned nodes.
Now we can express $T^{(k+1)}(w)$ via
\[
T^{(k+1)}(w) = T^{(k)}(w) \cup \bigcup_{v \in V^{(k)}_w} T^{(k)}(v).
\]
Since all $v \in V^{(k)}_w$ are also contained in $G^{(k)}$ we can apply \Cref{lem:layersPartiallyFixed} and we have 
\[
q^{(k)}(v) \leq \sum_{j = 0}^{\myceil{(h^{(k)}(v) + 1) / b} - 1} \lambda^{j}|S_j(v)|.
\]
To make things precise we need to relate the height of these trees rooted at such a $v$ to the height of the tree of $w$.
Notice that for any node $u \in T^{(k+1)}(v)$ its distance to $w$ is exactly its distance to $v$ plus the distance from $v$ to $w$, i.e.,
\[
\dist(u,w) = \dist(u,v) + \dist(v,w).
\]
By using \Cref{lem:distanceOfJustRemoved} we get that for any node $v \in V^{(k)}_w$ the distance to $w$ is less than $3b$.
    As a result, we get that for any node $u \in T^{(k)}(v)$ for any $v \in V^{(k)}_w$
    \[
    \text{dist}(u,w) = \text{dist}(u,v) + \text{dist}(v,w) \leq \text{dist}(u,v) + 3b.
    \]
    If $u$ was in $S^{(k)}_j(v)$ and the distance from $v$ to $w$ were actually exactly $3b$, then $u$ would be in $S^{(k+1)}_{j-3}(w)$.
    Now since 
    \[
        q^{(k)}(v) \leq \sum_{j = 0}^{\myceil{(h^{(k)}(v) + 1) / b} - 1} \lambda^{j}|S_j^{(k)}(v)|
    \]
    and 
    \[
        T^{(k+1)}(w) = T^{(k)}(w) \cup \bigcup_{v \in V^{(k)}_w} T^{(k)}(v),
    \]
    we get that 
    \[
        q^{(k+1)}(w) \leq q^{(k)}(w) + \sum_{v \in V^{(k)}_w} q^{(k)}(v) 
    \]
    \[
        \leq \sum_{j = 0}^{\myceil{(h^{(k)}(w) + 1) / b} - 1} \lambda^{j}|S_j^{(k)}(w)| + \sum_{v \in V^{(k)}_w} \sum_{j = 0}^{\myceil{(h^{(k)}(v) + 1) / b} - 1} \lambda^{j}|S_j^{(k)}(v)|
    \]
    Now again we look at $S^{(k)}_j(v)$ for some arbitrary $j$. We already know that only $\lambda^j$ of all of these nodes contribute to $q^{(k+1)}(w)$. We now want to express them with respect to the layers $S^{(k+1)}_0(w), \dots$\\
    We want to relate everything to the layer sets of $w$ 
    \[
        q^{(k+1)}(w) \leq \sum_{j = 0}^{\myceil{(h^{(k)}(w) + 1) / b} - 1} \lambda^{j}|S_j^{(k)}(w)| + \sum_{v \in V^{(k)}_w} \sum_{j = 0}^{\myceil{(h^{(k)}(v) + 1) / b} - 1} \lambda^{j}|S_j^{(k)}(v)| 
    \]
    \[
        \leq |S_0^{(k+1)}(w)| + |S_1^{(k+1)}(w)| + |S_2^{(k+1)}(w)| + \sum_{j = 3}^{\myceil{(h^{(k+1)}(w) + 1) / b} - 1} \lambda^{j-3}|S_j^{(k+1)}(w)|.
    \]
    We argue that the second inequality holds. For any node $u \in S^{(k)}_j(v)$, we get that its distance from $w$ is at most $3b$ larger than its distance to $v$. So the farthest layer $S^{(k+1)}_0(w), \dots$ it could be in would be $S^{(k+1)}_{j+3}(w)$ so if we only want to show that an $\lambda^j$ fraction of nodes in $S^{(k+1)}_{j+3}$\footnote{respectively a $\lambda^{j-3}$ fraction of nodes in $S^{(k+1)}_{j}$} contribute to the quality of $q^{(k+1)}(w)$ we can simply assume the worst case and say that $v$ has this distance $3b$ from $w$. 
    Because if the distance was smaller, it would only mean that some of the nodes from $S^{(k)}_j(v)$ end up in layers before $S^{(k+1)}_{j+3}$ and hence have a bigger coefficient.
    
    So now, as we argued before, since this bound holds in iteration $k+1$ in which $w$ was already assigned a layer, this bound will also hold for later iterations. So we get 
    \[
    q^{(i)}(w) \leq q^{(k+1)}(w) 
    \]
    \[
    \leq |S_0^{(k+1)}(w)| + |S_1^{(k+1)}(w)| + |S_2^{(k+1)}(w)| + \sum_{j = 3}^{\myceil{(h^{(k+1)}(w) + 1) / b} - 1} \lambda^{j-3}|S_j^{(k+1)}(w)|.
    \]
\end{proof}
Now we can prove \Cref{lem:layersPartiallyFixed}.
\begin{proof}[Proof of \Cref{lem:layersPartiallyFixed}]
We will prove the claim by induction on $h^{(i)}(r)$, the height of the assigned subtree $T^{(i)}(r)$. \\
\textbf{Base Case: $0 \leq h(r) < b$.}
For $h(r)<b$, the set $S_0(r)$ is just the entire tree $T(r)$, so the statement is trivially true for any $0 < \lambda < 1$.\\
\noindent\textbf{Induction step:} Now we want to prove the statement for some tree with height $h^{(i)}(r)$, so let the statement be true for all trees that have height at most $h^{(i)}(r) - b$.

We will use the induction hypothesis on all of the nodes in $C_b(r)$.
Notice that this is valid, since their heights can be at most $h^{(i)}(r) - b$.
Let $w \in C_b(r)$ be one such node. By invoking \Cref{lem:boundAssignedNode} we get
\[
q^{(i)}(w) \leq q^{(k+1)}(w) 
\]
\[
\leq |S_0^{(k+1)}(w)| + |S_1^{(k+1)}(w)| + |S_2^{(k+1)}(w)| + \sum_{j = 3}^{\myceil{(h^{(k+1)}(w) + 1) / b} - 1} \lambda^{j-3}|S_j^{(k+1)}(w)|.
\]
We note that this is not a cyclic argumentation, as by the inductive hypothesis \Cref{lem:layersPartiallyFixed} already holds for $w$ and all nodes in its subtree (Critically all of the nodes in the set $V_w^{(k)})$.
We get such a bound for every node $w \in C_b(r)$. To keep things clearer we will refer to the iteration $k$ which we used for the induction hypothesis as $k_w$ for a specific $w$, as these iterations might be different for different $w$. So for every $w \in C_b(r)$ we get 
\[
q^{(i)}(w) \leq q^{(k_w+1)}(w) 
\]
\[
\leq |S_0^{(k_w+1)}(w)| + |S_1^{(k_w+1)}(w)| + |S_2^{(k_w+1)}(w)| + \sum_{j = 3}^{\myceil{(h^{(k_w+1)}(w) + 1) / b} - 1} \lambda^{j-3}|S_j^{(k_w+1)}(w)|
\]
\[
= \sum^2_{j=0}|S_j^{(k_w+1)}(w)| + \sum_{j = 3}^{\myceil{(h^{(k_w+1)}(w) + 1) / b} - 1} \lambda^{j-3}|S_j^{(k_w+1)}(w)|
\]
Critically, we get bounds on their quality during iteration $k_w+1$ in which they were just assigned a layer, which must have happened before any node was promoted in Line \ref{alg:line:markPromoted}.
As a result, we get that no node that was at distance $2b-1$ from $r$ that was promoted, is accounted for in these bounds on the qualities. This is because the nodes we have used the induction hypothesis on are at distance exactly $b$ from $r$ and at the point at which they were free nodes, they had only promoted nodes at distance exactly $b$ from themselves, so of distance strictly larger than $2b -1$ with respect to $r$. And at the point in time at which our bounds on $q^{(i)}(w)$ were derived, we did not yet consider any node being promoted in iteration $k_w+1$. We will denote by $q'(r)$ the quality of $r$ at this point in time, i.e., at a point in time before considering any node at distance $2b-1$ from $r$ as promoted.
We get
\[
q'(r) \leq |S^{(i)}_0(r)| + \sum_{w \in C_b(r)} q^{(k_w+1)}(w)  
\]
\[
\leq |S^{(i)}_0(r)| + \sum_{w \in C_b(r)} \left(\sum^2_{j=0}|S_j^{(k_w+1)}(w)| + \sum_{j = 3}^{\myceil{(h^{(k_w+1)}(w) + 1) / b} - 1} \lambda^{j-3}|S_j^{(k_w+1)}(w)|\right).
\]
Since $S^{(k_w+1)}_j(w)$ will not change anymore, the set is equal to the set $S^{(i)}_j(w)$.
By then also observing the fact that 
\[
\bigcup_{w \in C_b(r)} S^{(i)}_j(w) = S^{(i)}_{j+1}(r),
\]
we obtain
\[
 q'(r) \leq \sum^3_{j=0}|S_j^{(i)}(r)| + \sum_{j = 4}^{\myceil{(h^{(i)}(r) + 1) / b} - 1}  \lambda^{j-4}|S_j^{(i)}(r)|.
\]
Again this holds before any node at distance $2b-1$ from $r$ was promoted.
Now we have to differentiate between two cases: either the $v^*$ computed by $r$ was promoted or not. If $v^*$ is promoted, we get by Lemma \ref{lem:fixV} that $q'(v^*)$ is some constant fraction of $q'(r)$. Now since node $v^*$ is at distance exactly $b$ to $r$, the nodes that account for $q'(v^*)$ are not yet included in the bounds on $q'(r)$, hence $q(r) = q'(r) - q'(v^*)$. In the other case, where the condition to promote $v^*$ does not hold, we get by Lemma \ref{lem:fixX} that some $x$ exists such that $q'(x)$ is a constant fraction of $q'(v^*)$ and, by extension, also a constant fraction of $q'(r)$. Furthermore Lemma \ref{lem:fixX} gives us that this $x$ must be in the first $2b-1$ layers and therefore also $q(r) \leq q'(r) - q'(x)$ does hold. \\
By Lemmas \ref{lem:fixV} and \ref{lem:fixX} we know that the fraction of $q'(r)$ that will be subtracted is smaller in the second case, where we just have that some $x$ exists, that was promoted. So w.l.o.g. we can just consider the second case.
By Lemmas \ref{lem:fixV} and \ref{lem:fixX} we obtain
\[
q'(x) \geq \frac{q'(v^*)}{2\Delta^b} \geq \frac{q'(r)}{4 \Delta^{2b}},
\]
which implies
\[
q^{(i)}(r) \leq q'(r) - q'(x) \leq q'(r) - \frac{q'(r)}{4 \Delta^{2b}} = (1 - \frac{1}{4 \Delta^{2b}}) q'(r).
\]
By choosing $\lambda^{4} = (1 - \frac{1}{4 \Delta^{2b}})$, we obtain
\[
q^{(i)}(r) \leq (1 - \frac{1}{4 \Delta^{2l}}) q'(r)
\]
\[
= \lambda^4 \left(\sum^3_{j=0}|S_j^{(i)}(r)| + \sum_{j = 4}^{\myceil{(h^{(i)}(r) + 1) / b} - 1}  \lambda^{j-4}|S_j^{(i)}(r)|.\right) 
\]
\[
\leq  \sum_{j = 0}^{\myceil{(h^{(i)}(r) + 1) / b} - 1} \lambda^{j}|S_j^{(i)}(r)|.
\]
\end{proof}

\subsection{Distributed Algorithm and Node Averaged Complexity}\label{sec:dectoave}
In this section, we will describe how we implement the decomposition algorithm, \Cref{alg:Decomposition} in a distributed manner and how we will use it to design an algorithm $\fA$ that solves the given LCL problem $\Pi$, and we will prove an upper bound of $O(\log^* n)$ for the node-averaged complexity of the latter algorithm.
We start by describing our distributed implementation of \Cref{alg:Decomposition}.
For the remainder of the section, set $s:= 10\ell$.

\subparagraph{Distributed implementation.}
The computation of $\ell$ from $\Pi$ in Line~\ref{alg:line:one} of \Cref{alg:Decomposition} can be performed by every node without any communication.
Next, the nodes compute a distance-$s$ coloring with a constant number of colors.
Since $\Delta$ and $\ell$ are constant, this can be done in worst-case complexity $O(\log^* n)$, e.g., by computing a $(\Delta(G^s) + 1)$-coloring of the power graph $G^s$, using the algorithm of Barenboim, Elkin and Kuhn~\cite{bek15} (that computes a $(\Delta + 1)$-coloring in $O(\log^* n + \Delta)$ rounds).
We argue that each subroutine can be implemented in a constant number of rounds:
\begin{itemize}
    \item \textbf{Rake (\cref{alg:Rake}):} Since all nodes know their layer, if they are assigned one, nodes can easily check whether or not they have at most one neighbor that is not assigned a layer. They can then inform all of their neighbors about their new layer and any edges that are newly oriented. This takes just one round, so repeating it $\gamma$ times takes only $\gamma \in O(1)$ rounds.
    \item \textbf{Compress (\cref{alg:Compress}):} The computation of the subset $Z$ in Line~\ref{computez} can be performed in a constant number of rounds by iterating through the color classes of the computed distance coloring and in each iteration adding a node to $Z$ if it does not cut off a subpath of $P'$ of length $< \ell$ (which can be determined in $\ell$ rounds per iteration).
    Note that no path of length $2\ell + 1$ can remain in the graph induced by $V(P') \setminus Z$ since the middle node of this path would have been added to $Z$ during the iteration of its color class. Analogously, no path of length $>2\ell + 1$ can remain after the computation of $Z$.
    \item \textbf{Promotion (\cref{alg:Promote}):} Both $b$, or $|P|$ are constants, so each free node can collect its $b$ hop neighborhood and compute the outcome of the promotion and then in $b+1$ rounds inform all nodes of the changes in their layers, or in the layers of their neighbors.
\end{itemize}
We obtain the following lemma.
\begin{lemma}\label{lem:TimeOfIteration}
    Assume a distance-$s$ coloring with a constant number of colors is given.
    Then for all iterations $i$ of \Cref{alg:Decomposition} the iteration can be executed in a constant number of rounds. Furthermore, the distance coloring is only used for the compress subroutine and nothing else.
\end{lemma}

Note that a node can determine in constant time whether it should take part in a certain operation as this only depends on the node's constant-hop neighborhood.

Next we will describe our algorithm $\fA$ for solving a given LCL problem $\Pi$.

\subparagraph{Algorithm for $\Pi$.}
In order to describe our algorithm $\fA$ for solving $\Pi$ with a small node-average complexity, we first describe an algorithm $\fA'$ that is not optimized for the node-averaged setting and then explain how to tweak $\fA'$ to obtain $\fA$.

Algorithm $\fA'$ proceeds as follows.
Use \Cref{alg:Decomposition} to compute a $(\gamma, \ell, L)$-decomposition, where the values of $\gamma$ and $\ell$ depend on $\Pi$, and $L\in O(\log n)$ (due to \Cref{cor:inlog} and \cref{cor:PartialIsProperDecomposition}).
Then execute the generic algorithm from \Cref{ssec:lcls} using the computed $(\gamma, \ell, L)$-decomposition.
As explained in \Cref{ssec:lcls}, the generic algorithm produces a correct output for $\Pi$ given \emph{any} $(\gamma, \ell, L)$-decomposition. However doing the two algorithms sequentially means, that we first have to wait for $O(\log n)$ rounds until the entire decomposition is done before any node can terminate.

In order to turn $\fA'$ into an algorithm $\fA$ with small node-averaged complexity, we simply let each node start executing the steps in the generic algorithm as soon as the partial decomposition computed so far by \Cref{alg:Decomposition} provides all necessary information. Additionally, when having determined the output labels for all incident edges, each node immediately terminates.

From the description of the generic algorithm provided in \Cref{ssec:lcls}, it follows precisely what information a node needs in order to execute the two steps in the generic algorithm: choosing label-sets for incident edges and choosing labels for incident edges.
In particular, the description of the generic algorithm implies the following:
\begin{enumerate}
    \item\label{prop:theone} Consider a layer computed in the $(\gamma, \ell, L)$-decomposition. At most $2\ell + 1$ rounds after all nodes in all smaller layers (according to the total order given in  \Cref{def:ordering}) have chosen label-sets for all incident edges (or earlier), all nodes in the considered layer know the label-set for each incident edge. If there is no smaller layer, all nodes in the considered layer know their label-sets after the first round.
    \item\label{prop:theother} If a node is a local maximum and knows its incident label-sets, it can immediately output labels for all incident edges.
    If a node is in a rake layer (i.e., in a layer of the form $V_{i,j}^R$) and knows its incident label-sets, then it can output labels for all of its incident (so far unlabeled) edges one round after all incident edges of its parent (if it has one) have an output label.
    If a node is in a compress layer (i.e., in a layer of the form $V_{i}^C$) and each node in its path $P$ in the layer knows its incident label-sets, then each node in the path can output labels for all of their incident (so far unlabeled) edges $2\ell + 1$ rounds after all incident edges of the parents of the two endpoints of the path (if they have one) have an output label.
\end{enumerate}
These two properties enable us to prove the following lemma.

\begin{lemma}\label{lem:TimeToFixSubtree}
    Assume a distance-$s$ coloring with a constant number of colors is given.
    Then there exists an integer constant $t$ such that the following holds:
    if a node $v$ becomes a local maximum in iteration $i$ of \Cref{alg:Decomposition}, then the entire tree $T^{(i)}(v)$ will have terminated after $ti$ rounds in $\fA$. 
    Furthermore for all nodes $u \in N^{(i)}$ both $v$ and all nodes in $T^{(i)}(u)$ will have terminated after $ti$ rounds in $\fA$. 
\end{lemma}
\begin{proof}
    Let $v$ be a node that becomes a local maximum in iteration $i$ of \Cref{alg:Decomposition}.
    The design of \Cref{alg:Decomposition} guarantees that there are at most $4i\gamma$ layers below the layer of $v$, $2i\gamma$ from the rake and compress procedure and another factor of two from \cref{cor:PartialIsProperDecomposition}. Then Property~\ref{prop:theone} in the above discussion implies that $v$ knows the label-sets for its incident edges latest after round $(4i\gamma)(2\ell +1) < 9i\gamma\ell$.
    Then as a direct result of Property~\ref{prop:theother} together with \Cref{lem:promoteCorrect} all of the nodes $u \in N^{(i)}$ can chose output labels for their incident edges and terminate after another at most $9i\gamma\ell$ rounds.
    Now, set $t$ to be the sum of $2\cdot9\gamma\ell = 18\gamma\ell$ and the constant from \Cref{lem:TimeOfIteration}.
\end{proof}

To make the runtime analysis a bit cleaner, we are going to mark all nodes in $T^{(i)}(v)$, once $v$ becomes a local maximum. We emphasize that this is solely for the purpose of the analysis and this does not change the algorithm at all. More specifically, once any node $v$ becomes a local maximum, all of the nodes in $T^{(i)}(v)$ become marked instantly (in $0$ rounds).
We obtain the following corollary from \Cref{lem:TimeToFixSubtree}.

\begin{corollary}\label{cor:markedNodesTerminate}
    Assume a distance-$s$ coloring with a constant number of colors is given.
    If a node $v$ becomes marked in iteration $i$, then $v$ will have terminated in round $ti$, where $t$ is the constant from \Cref{lem:TimeToFixSubtree}.
\end{corollary}
Now we use this result to prove that a large amount of nodes are in subtrees of local maxima after a reasonable amount of time. This will be sufficient to prove that the node-averaged complexity is constant, without the time for the input coloring.
\begin{lemma} \label{lem:FixedInIteration}
    There exists a constant $0 < \sigma < 1$, such that for every iteration $i > 5$ of \Cref{alg:Decomposition} at most $2 \Delta^b n \sigma^{i}$ nodes are not marked.
\end{lemma}
\begin{proof}
Let $F$ be the set of all marked nodes after iteration $i$. We want to upper bound the size of $V \setminus F$. 
By \Cref{lem:decomposeV}, we get a decomposition of $V(G)$ into
\[
    V(G) = \left( \bigcup_{v \in G^{(i)}}T^{(i)}(v)\right) \cup \left(\bigcup_{v \in N^{(i)}}T^{(i)}(v)\right).
\]
For simplicity, we will define the left part as $R_i$ and the right part as $C_i$, i.e.,
\[
R_i := \bigcup_{v \in G^{(i)}}T^{(i)}(v) , \; C_i := \bigcup_{v \in N^{(i)}}T^{(i)}(v).
\]
We obtain that
\[
V\setminus F = (R_i \cup C_i)\setminus F = (C_i \setminus F ) \cup (R_i \setminus F ).
\]
So we now need to bound the size of these two sets.
To bound the size of $C_i\setminus F$, we observe that by \cref{lem:TimeToFixSubtree} all nodes in $N^{(i)}$ are marked in iteration $i$, so all nodes in $C_i$ will be marked. So $C_i\setminus F = \emptyset$ and we obtain
\[
|C_i\setminus F| = 0.
\]
For the nodes in $R_i$ we will have to put a bit more effort into it. If we look at a tree $T^{(i)}(r)$ without any of the already marked nodes, we obtain exactly $H(v)$ from \Cref{def:quality}. This is because any marked node is in a descendant of some local maximum. Therefore such a node cannot be reached by a path that satisfies the requirements for $H(v)$. So we get that $|T^{(i)}(r) \setminus F| = q^{(i)}(r)$, which implies
\[
    |R_i \setminus F| = |\bigcup_{r \in G^{(i)}} \left(T^{(i)}(r) \setminus F\right)| = \sum_{r \in G^{(i)}} q^{(i)}(r).
\] 
Now by applying \Cref{lem:layersPartiallyFixed}, we get that 
\[
 |R_i \setminus F| = \sum_{r \in G^{(i)}} q^{(i)}(r) \leq \sum_{j = 0}^{\myceil{(h^{(i)}(r) + 1) / b} - 1} \lambda^{j}|S^{(i)}_j(r)|
\]
By noticing that
$S_j(r)$ is empty for $jb > h^{(i)}(r) + 1$ and defining $h = \max \{h^{(i)}(r) + 1 \mid r \in G^{(i)}\}$ we can reformulate the expression as
\[
    \sum_{r \in G^{(i)}} \sum_{j = 0}^{\myceil{(h^{(i)}(r) + 1) / b} - 1} \lambda^{j}|S_j(r)|  = \sum_{r \in G^{(i)}} \sum_{j = 0}^{\myceil{h / b} - 1} \lambda^{j}|S_j(r)|.
\]
Observe that by \cref{lem:distanceOfJustRemoved} the maximum height $h$ is upper bounded by $3ib < 3i\gamma$ which is in turn upper bounded by $ti$ where $t$ is the constant from \Cref{lem:TimeToFixSubtree} and \Cref{cor:markedNodesTerminate}.
Hence, we obtain
\[
    \sum_{r \in G^{(i)}} \sum_{j = 0}^{\myceil{h / b} - 1} \lambda^{j}|S_j(r)| \leq \sum_{r \in G^{(i)}} \sum_{j = 0}^{\myceil{ti / b} - 1} \lambda^{j}|S_j(r)|.
\]
By splitting at some $\beta$ fraction of the height, which we will fix later, we obtain
\[
    \sum_{r \in G^{(i)}} \sum_{j = 0}^{\myceil{ti/ b} - 1} \lambda^{j}|S_j(r)| \leq \sum_{r \in G^{(i)}} \left (\sum_{j = 0}^{\myfloor{\beta ti / b}} \lambda^{j}|S_j(r)| + \sum_{j = \myceil{\beta ti / b}}^{\myceil{ti / b} - 1} \lambda^{j}|S_j(r)| \right)
\]
\[
    \leq \sum_{r \in G^{(i)}} \left (\sum_{j = 0}^{\myfloor{\beta ti / b}} 1 \cdot |S_j(r)| + \sum_{j = \myceil{\beta ti/ b}}^{\myceil{ti / b} - 1} \lambda^{\myceil{\beta ti/b}}|S_j(r)| \right).
\]

Now we want to upper bound the number of nodes that are in the first sum.
To this end, we notice that in a tree of height $\beta ti$ there can be at most $\Delta^{\beta ti}$ nodes.
Furthermore, $|G^{(i)}| \leq  n\varepsilon^i$ by \Cref{lem:epsilon}, which implies
\[
    \sum_{r \in G^{(i)}} \sum_{j = 0}^{\myfloor{\beta ti/ b}} |S_j(r)| \leq \sum_{r \in G^{(i)}} \Delta^{\beta ti + b} \leq |G^{(i)}|\Delta^{\beta ti + b} 
\]
\[
    \leq n\varepsilon^i\Delta^{\beta ti + b} = \Delta^b n \cdot \text{exp}\left(i\left(\beta t \ln(\Delta) - \ln(1/\varepsilon)\right)\right)
\]
Now by choosing $\beta = \frac{\ln(1 / \varepsilon)}{2 t \ln(\Delta)}$, we obtain
\[
    \Delta^b n \cdot \text{exp}\left(i\left(\beta t \ln(\Delta) - \ln(1/\varepsilon)\right)\right) = \Delta^b n \cdot \text{exp}\left(-i \cdot \frac{1}{2} \cdot \ln(1/\varepsilon)\right) = \Delta^b n \cdot \left( \sqrt{\varepsilon} \right)^i.
\]
So the bound we get is 
\[
\sum_{r \in G^{(i)}} \sum_{j = 0}^{\myfloor{\beta ti/ b}} |S_j(r)| \leq \Delta^b n \cdot \left( \sqrt{\varepsilon} \right)^i.
\]
Now to bound the size of the second sum, we observe that
\[
\sum_{r \in G^{(i)}} \sum_{j = \myceil{\beta ti/ b}}^{\myceil{ti / b} - 1} \lambda^{\myceil{\beta ti/b}}|S_j(r)| = \lambda^{\myceil{\beta ti/b}} \cdot \sum_{r \in G^{(i)}} \sum_{j = \myceil{\beta ti/ b}}^{\myceil{ti / b} - 1} |S_j(r)|  \leq n \cdot \left( \lambda^{\beta t/b} \right)^i.
\]
Combining the two obtained inequalities yields
\[
|R_i \setminus F| \leq \sum_{r \in G^{(i)}} \left (\sum_{j = 1}^{\myfloor{\beta ti / b}} |S_j(r)| + \sum_{j = \myceil{\beta ti/ b}}^{\myceil{ti / b} - 1} \lambda^{\myceil{\beta ti/b}}|S_j(r)| \right) \leq \Delta^b n \cdot \left( \sqrt{\varepsilon} \right)^i + n \cdot \left( \lambda^{\beta t/b} \right)^i.
\]
It follows that
\[
|V\setminus F| = |(R_i \cup C_i)\setminus F| = |((C_i \setminus F ) \cup (R_i \setminus F )| \leq 0 + \Delta^b n \cdot \left( \sqrt{\varepsilon} \right)^i +  n \cdot \left( \lambda^{\beta \gamma/l} \right)^i.
\]
By choosing $\sigma = \max \{\lambda^{\beta \gamma/l}, \sqrt{\varepsilon} \}$, we obtain
\[
 |V \setminus F| \leq 2\Delta^b n \sigma^i
\]
\end{proof}

We obtain the following lemma.

\begin{lemma} \label{lem:AvgIteration}
    On average, nodes become marked in $O(1)$ iterations.
\end{lemma}
\begin{proof}
     By \Cref{lem:FixedInIteration}, we get that for each iteration $i>5$, only $2\Delta^b n \sigma^i$ nodes are not marked. The total number of iterations is therefore upper bounded by $5n$ for the initial iterations plus the sum over $2\Delta^b n \cdot \sigma^i$. We then get the average number of iterations per node by dividing by $n$. We therefore get an upper bound of
     \[
     \frac{1}{n} \left( 5n + \sum_{i=6}^\infty2\Delta^b n \cdot \sigma^i \right) \leq 5+2\Delta^b \sum_{i=1}^\infty \sigma^i = 5+ \frac{2\Delta^b}{1-\sigma} \in O(1).
     \]
     on the average number of iterations per node.
 \end{proof}

 Then using this lemma together with \Cref{cor:markedNodesTerminate} we get that an average node terminates after a constant number of rounds. However, we still have to pay for the input distance coloring which takes $O(\log^* n)$, as discussed in the beginning of the section. So by first computing this input coloring and then running the algorithm, we obtain a total node-averaged complexity of $O(\log^* n)$, proving \Cref{thm:logstaraverage}.

\section{Improved Node Averaged Complexity using randomization}\label{sec:randomConstant}
We emphasise again, that the initial coloring is only used for the compress procedure and if we are able to speed up the compress operation we will obtain a better node averaged complexity. This is enough to prove \cref{thm:RandomConst}.
We make this precise in the following and show how, using randomization, we can improve the algorithm to obtain constant node averaged complexity.\\

In the following assume $\Pi$ is some fixed LCL with worst case complexity $O(\log n)$, $f_{\Pi,\infty}$ a feasible function in the sense of \cref{def:computing-label-set}.

\begin{definition}[Compress Problem]
Assume we are given a graph $G$ with a valid partial $(\gamma, \ell, i)$-decomposition. Let $G^{(i)}$ bet the subgraph induced by nodes that do not yet have a layer assigned to them. Consider a path of degree 2 nodes $P\subset G^{(i)}$ of length $\beta$ at least $2\ell + 1$. We split this path into three parts $P=(P_l, N, P_r)$, where $P_l=(v_0, \ldots, v_j)$ and $P_r=(v_k, \ldots, v_\beta)$ for integers $j,k$ such that $|P_l|, |P_r| \in [\ell, 2\ell]$ and $N = P \setminus (P_l \cup P_r)$. Then the Compress Problem is defined as follows:
\begin{itemize}
    \item Let $e_l,e_r \in E$ be the two edges that connect $P$ to the rest of $G^{(i)}$. Then $P_l$ and $P_r$ have to output the labelsets, that $f_{\Pi,\infty}$ would have assigned to them and to $e_l,e_r$.
    \item Every node $v \in N$ chooses output labels for their incident edges, that form a locally correct solution for $\Pi$ with respect to the label sets of their lower layer neighbors.
    \item The output labels in $N$ and below are such that no matter what labels are chosen from the labelsets of $e_l,e_r$ the labeling can be extended to a valid labeling.
\end{itemize}
\end{definition}

An illustration of the setting is given in \cref{fig:compressProblem}.

\begin{figure}[h!]
    \centering
    \includegraphics[width=\textwidth]{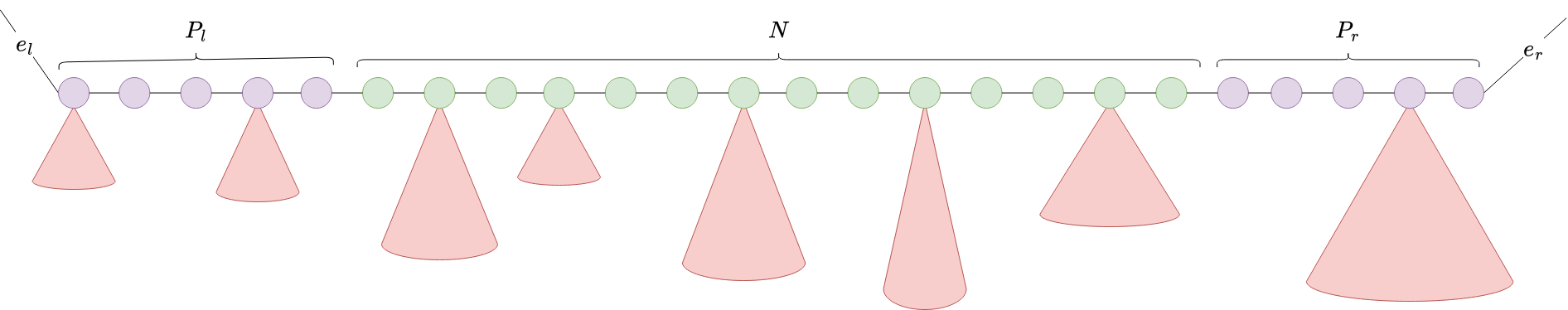}
    \caption{Illustrating the Compress Problem, the purple nodes form the two sets $P_l, P_r$ at the ends of the path. The red cones represent possible lower layer nodes attached to the nodes of $P$.}
    \label{fig:compressProblem}
\end{figure}

This is essentially exactly the problem that the \emph{Compress with Slack} procedure solves every single iteration, by applying the $f_{\Pi, \infty}$ to the compress layers it creates. However doing this requires the small input coloring and therefore $\log^*(n)$ rounds of precomputation.
Assume we are given another algorithm to solve the Compress Problem, then we could change the compress procedure as given in \cref{alg:BetterCompress}. There we manage the extra slack at the ends of the compress path and then create an instance of the Compress Problem which we can then solve using another algorithm.
\begin{algorithm2e}
\caption{Better Compress}\label{alg:BetterCompress}
\KwIn{$G=(V,E), i, \ell, f_{\Pi,\infty}$} 
$G^{(i)} \gets G \setminus (\bigcup_{j<i}(V_i^R\cup V_{j}^C))$ \Comment{Graph induced by free nodes.}\\
\For{each maximal path $P$ consisting of nodes of degree exactly $2$ in $G^{(i)}$}  
{
    \If{$|V(P)|\geq 4\ell + 9$} {
        $P' \gets$ the subpath of $P$ consisting of all nodes that have distance at least $\ell + 3$ from both endpoints of $P$ \\
        compute $P'=(P_l, N, P_r)$\\
        apply $f_{\Pi,\infty}$ to $P_l$ and $P_r$\\
        orient $P_l$ towards $e_l$ and $P_r$ towards $e_r$\\
        assign labelsets to $e_l$ and $e_r$\\
        Solve the compress problem on $P'$\\
        $N^{(i)} \gets N^{(i-1)} \cup N$ \Comment{Only for the analysis}\\
	}
}
\end{algorithm2e}

\begin{corollary}\label{cor:replaceCompress}
When we replace the \emph{Compress with Slack} procedure in the decomposition algorithm \cref{alg:Decomposition} with the \emph{Better Compress} procedure, the algorithm will still output a correct solution.
\end{corollary}
\begin{proof}
We still remove all paths of length $4\ell + 9$, have the exact same amount of slack at the ends and propagate valid labelsets on $e_l,e_r$ and $P_l,P_r$, as in the algorithm from \cref{sec:dectoave}. Note that these labelsets are valid in the sense, that they are produced by $f_{\pi, \infty}$ and we therefore get correctness from \cite{CP19timeHierarchy}.
\end{proof}

We essentially compute a different kind of $(\gamma,\ell,L)$-decomposition, where the length of the compress paths is not required to be in $[\ell,2\ell]$. Instead the paths must be at least $\ell$ long and can instead be arbitrarily long. However for our analysis to work, all nodes in $N^{(i)}$ must be able to terminate in a constant number of rounds in expectation.

So if we find a nice way to solve to compress problem fast, we can apply it to these long paths and improve on the node-averaged complexity.

\begin{lemma}\label{lem:FastCompressGivesConstant}
Consider some LCL $\Pi$ with worst case complexity $O(\log n)$ and feasible function $f_{\Pi,\infty}$. If there exists an algorithm $A$ solving the compress problem such that \emph{in expectation} nodes only use $O(1)$ rounds. Then there exists an algorithm $A'$ solving $\Pi$ with node averaged complexity $O(1)$.
\end{lemma}
\begin{proof}
The new algorithm $A'$ is obtained by replacing the \emph{Compress with Slack} procedure \cref{alg:Compress} with the new \emph{Better Compress} from \cref{alg:BetterCompress}.
Because of \cref{cor:replaceCompress} we get that the algorithm will be correct. Furthermore nodes outside of the compress paths don't have to wait for the compress problem to be solved completely.
We first compute $P_l$ and $P_r$ by differentiating between two cases:
\begin{itemize}
    \item \textbf{Case 1} $(\beta < 3\ell+2)$: Then both endpoints can see the entire path and so we choose $P_l=(v_0, \ldots, v_\ell)$ and $P_r=(v_{\ell+2}, \ldots, v_\beta)$. We then have $N=\{v_{\ell+1}\}$ and put $v_{\ell+1}$ in layer $V^R_{i,1}$ and the nodes in $P_l$ and $P_r$ into $V^C_{i-1}$ as usual.
    \item \textbf{Case 2} $(\beta \geq 3\ell+2)$: Then both endpoints $v_0,v_\beta$ see that the path is long enough and choose $P_l = (v_0, \ldots, v_\ell)$ and $P_r=(v_{\beta-\ell}, \ldots, v_\beta)$ \footnote{We don't actually need to differentiate between the ''left'' and ''right'' path, as both simply apply $f_{\Pi, \infty}$}. We immediately assign layers $V^C_{i-1}$ to $P_l$ and $P_r$ and propagate the labelsets up through $e_l$ and $e_r$. We get that $N = (v_{\ell+1}, \ldots, v_{\beta - \ell-1})$.Since $\beta>3\ell +2$ we get that $|N| > \ell+2$ and so by assigning $v_{\ell+1}$ and $v_{\beta-\ell-1}$ to $V^R_{i,1}$, we get that $|N\setminus\{v_{\ell+1}, v_{\beta-\ell-1}\}| > \ell$. Clearly it would be possible to separate the unassinged nodes in $N$ into paths of length $[\ell,2\ell]$, so there exists a valid solution to the compress problem $P=(P_l,N,P_r)$.
\end{itemize}
In the first case we are clearly already done, in the second case we split the path in a way, such that we obtain a valid Compress Problem. Everything up until this point takes only a constant number of rounds. We then use $A$ to solve this compress problem $P = (P_l,N,P_r)$ and have that in expectation nodes only use $O(1)$ rounds. Crucially since the labelsets on $e_l$ and $e_r$ are already known after a constant number of rounds, the decomposition algorithm does not need to wait for the compress problem to be solved. So the decomposition algorithm still performs one iteration in $O(1)$ rounds.
Now the only key difference is that in the analysis of the original decomposition algorithm we always solved the compress problem in a constant number of rounds and therefore all of the nodes in $N$ did terminate after a constant number of rounds. 
Instead $A$ only gives us the guarantee that the nodes in $N$ terminate after a constant number of rounds \emph{in expectation}.
It is however trivial to follow the previous analysis and see that by linearity of expectation in \cref{lem:TimeToFixSubtree,lem:FixedInIteration,lem:AvgIteration}, we obtain randomized $O(1)$ node averaged complexity.
\end{proof}

We now give an algorithm that solves the compress problem in a randomized way, such that in expectation nodes are done after a constant number of rounds. We don't really do more than computing a $[\ell+1,2\ell+1]$ ruling set, but to fill out the details of how we obtain valid output labels, we present the entire analysis. In \cref{alg:RandCompress} we first take care of the trivial case, where the path is very short. We then proceed to have $O(\log n)$ executions of the subroutine \emph{Elect Maximums} given in \cref{alg:Elect} that computes a $[\ell+1,2\ell+1]$ ruling set $Z$ and has nodes output valid labels for the incident edges as soon as possible.

\begin{algorithm2e}
\caption{randomized Compress}\label{alg:RandCompress}
\KwIn{$P=(P_l,N,P_r), \ell, i-1, f_{\Pi,\infty}$} 
$v_l,v_r \gets$ the nodes adjacent to $P_l,P_r$\\
Apply $f_{\Pi,\infty}$ to $P_l$ and $P_r$.\\
\If{$v_l=v_r$}{
    Then $N=\{v_l\}$\\
    Now $v_l$ has only incoming edges and can compute its maximal class $B$ and terminate.\\
    Terminate\\
}
$Z \gets \{v_l,v_r\}$ \Comment{$N$ must have size at least $\ell+2$}\\
$c \gets O(1)$ \Comment{chose $c$ for desired failure probability}\\
\For{$8\ell \cdot c \cdot \ln n$ times}{ 
     ElectMaximums($\ell, Z, f_{\Pi,\infty}$)\\
}
\end{algorithm2e}

\begin{lemma}\label{lem:ProbJoinZ}
In each execution of \emph{Elect Maximums} \cref{alg:Elect}, each active node joins $Z$ with probability at least $\frac{1}{8\ell}$.
\end{lemma}
\begin{proof}
Consider some active node $u$. The decision whether or not $u$ will join $Z$ depends only on its random choice to become a candidate and its $\ell$ hop neighborhood. Let $C= \{v \in N \mid dist(u,v) \leq \ell \text{ and $v$ is active}\}$ be all of the nodes that are close enough to influence $u$'s decision and that will also potentially become a candidate. Then the probability that $u$ joins $Z$ is:
\begin{align*}
Pr[u \text{ joins } Z] &= Pr[u \text{ is a candidate}] \prod_{v \in C} Pr[v \text{ is no candidate}]= \frac{1}{2\ell} \cdot (1- \frac{1}{2\ell})^{|C|} \\
    &\geq \frac{1}{2\ell} \cdot (1- \frac{1}{2\ell})^{2\ell} \geq \frac{1}{2\ell} \cdot (1-\frac{1}{2})^2 \geq \frac{1}{8\ell}
\end{align*}
using that $|C| \leq 2\ell$ and $(1+\frac{x}{m})^m \geq 1+x$ when $m \geq 1, |x| \leq m$ for $m=\ell, x = -\frac{1}{2}$
\end{proof}

\begin{lemma}\label{lem:ElectRuntime}
One execution of \emph{Elect Maximums} \cref{alg:Elect} needs at most $6\ell$ rounds.
\end{lemma}
\begin{proof}
We will argue that none of the steps requires more information than the $6\ell$ neighborhood.
\begin{itemize}
    \item For loop in Line~\ref{alg:line:CandidateLoop}: Nodes can decide to become active seeing only their $\ell$ hop neighborhood. To determine whether or not $u$ does not have other candidates in its $\ell$ hop neighborhood, it needs to know the random bits and the $\ell$ hop neighborhoods of all of its $\ell$ hop neighbors. So its $2\ell$ neighborhood in total.
    \item For loop in Line~\ref{alg:line:MaximumsTerminateLoop}: For all nodes at distance at most $2\ell+1$, $z$ needs to know whether or not they are in $Z$. Since these nodes might have only joined in this execution of Elect Maximums, we need to additionally see the $2\ell$ neighborhood of each of these nodes, as discussed in the previous point. So the $4\ell+1$ neighborhood of $z$ in total.
    \item For loop in Line~\ref{alg:line:PathsTerminateLoop}: The path $P$ has length at most $2\ell$, since otherwise $z_1,z_2$ would not be done yet. As discussed in the previous point, nodes in $Z$ know whether or not they are done after $4\ell+1$ rounds. As a result nodes in $P$ will know if they are between two nodes that are done after at most $6\ell+1$ rounds.
\end{itemize}
\end{proof}

\begin{algorithm2e}
\caption{Elect Maximums}\label{alg:Elect}
\KwIn{$\ell, Z, f_{\Pi,\infty}$} 
$p = \frac{1}{2\ell}$\\
\For{every node $u \in N \setminus Z$}{\label{alg:line:CandidateLoop}
    \If{$u$ has a neighbor in $Z$ at distance $ \leq \ell$}{
        set $u$ to be inactive\\
    }\Else{
        $u$ becomes a candidate with probability $p = \frac{1}{2\ell}$ \Comment{so $u$ is active}\\
    }
    \If{$u$ is the only candidate in its $\ell$ neighborhood}{
        $u$ joins $Z$\\
    }
}
\For{every node $z \in Z$}{\label{alg:line:MaximumsTerminateLoop}
    \If{there are two other nodes $z_1,z_2 \in Z$ at distance at most $2\ell+1$}{
        $P_1 \gets $ the path between $z_1$ and $z$  \\
        $P_2 \gets $ the path between $z_2$ and $z$  \\
        apply $f_{\Pi,\infty}$ to $P_1$ and $P_2$ \label{alg:line:applyF}\\
        $z$ now has only incoming edges and can compute its maximal class $B_z$ and assign labels to all incoming edges\\
        mark $z$ as done \label{alg:line:done1}\\
     }
}
\For{every path $P$ between two nodes $z_1,z_2 \in Z$ that are done}{\label{alg:line:PathsTerminateLoop}
    Each node $u \in P$ was already assigned a label set in Line~\ref{alg:line:applyF} \\
    Use the labelsets that were assigned by $z_1,z_2$ to chose labels for all edges in $P$\\
    Propagate the choices of labels to all incoming edges of nodes in $P$\\
    mark all nodes in $P$ as done \label{alg:line:done2}\\
}
\end{algorithm2e}

\begin{lemma}\label{lem:RulingSet}
At the end of the execution of \emph{randomized Compress} \cref{alg:RandCompress} the nodes in $N\setminus Z$ form paths of length in $[\ell, 2\ell]$ with high probability.
\end{lemma}
\begin{proof}
Whenever we add a node $u$ into $Z$, then this node must be active and therefore at distance at least $\ell+1$ from any node already in $Z$. Furthermore nodes may only join $Z$, if no other node at distance $\ell$ is joining the set at the same time, that is are not also a candidate. As a result no two nodes at distance less than $\ell+1$ will join $Z$. So paths between two nodes in $Z$ have length at least $\ell$. Now for a path to be longer than $2\ell$, there must be at least one node $u$ in the middle that has distance $>\ell$ from all nodes in $Z$. But then this means, that $u$ was active in all of the executions of \emph{Elect Maximum}. So we get:
\begin{align*}
Pr[u \notin Z \mid \forall z\in Z : dist(u,z)>\ell] &= \prod_{1\leq i \leq 4\ell c \ln n} Pr[u \text{ doesn't join $Z$ in iteration } i \mid u \text{ is active}]
\end{align*}
Then by \cref{lem:ProbJoinZ} the probability that $u$ did not join $Z$ itself in one of the iterations is $1-\frac{1}{8\ell}$. So we get:
\[
 Pr[u \notin Z \mid \forall z\in Z : dist(u,z)>\ell] \leq (1- \frac{1}{8\ell})^{8\ell c \ln n} \leq e^{-\frac{8\ell c \ln n}{8\ell}} = e^{-c\ln n} = \frac{1}{n^c}
\]
Where we used that $(1+x) \leq e^x$. The statement now follows by a simple union bound over all nodes.
\end{proof}

Now we can use this to prove that we correctly solve the compress problem.

\begin{lemma}\label{lem:RandCompressCorrect}
The algorithm \emph{randomized Compress} \cref{alg:RandCompress} computes a valid solution to the Compress Problem.
\end{lemma}
\begin{proof}
We first take care of the initial case in which $v_l=v_r$, but then $v_l$ is between the two compress paths $P_l$ and $P_r$. So $v_l$ has only incoming edges and is therefore a local maximum. We therefore extend the existing partial $(\gamma, \ell, i)$-decomposition with $P_l \subset V^C_{i-1}$, $P_r \subset V^C_{i-1}$and $v_l \in V^R_{i,1}$. \\
In the second case we compute the subset $Z$, such that paths between nodes of $Z$ have length in $[\ell,2\ell]$ by \cref{lem:RulingSet}. We again extend the partial $(\gamma,\ell, i)$-decomposition, by putting all of the nodes in $Z$ into layer $V^R_{i,1}$ and the nodes in the paths in layer $V^C_{i-1}$. \\
In both cases we produce a valid $(\gamma, \ell, i)$-decomposition, analogous to the argumentation of \cref{lem:compressCorrect}. We then pick labels according to the labelsets given by the feasible function and maximal classes. We output the correct labelsets for both $P_l,P_r$ and $e_l,e_r$. Furthermore the labels that are chosen in $N$ are due to a valid $(\gamma,\ell,i)$-decomposition and using the feasible function $f_{\Pi,\infty}$, so correctness follows from the results of \cite{CP19timeHierarchy}, for details see \cref{ssec:lcls}.
\end{proof}

Finally, we prove that nodes are done and can terminate after a constant number of rounds in expectation, thus proving constant node averaged complexity for all LCL's with worst case complexity $O(\log n)$ by \cref{lem:FastCompressGivesConstant}. We first start with the probability that nodes are \emph{done}, meaning they have picked labels for their incident edges and can terminate.

\begin{lemma}\label{lem:ProbZDone}
Any node $z$ that was initially in $Z$, or that joined $Z$ after the Loop in Line~\ref{alg:line:CandidateLoop}, is marked as \emph{done} with constant probability.
\end{lemma}
\begin{proof}
Consider the $2\ell+1$ neighborhood of $z$ and let $A=(v^{left}_{2\ell+1}, \ldots, v^{left}_{1}, z, v^{right}_{1}, \ldots, v^{right}_{2\ell+1})$ be the path around $z$. The distinction between left and right is done arbitrarily based on the two outgoing edges of $z$ and just for the analysis. 
Because of \cref{lem:RulingSet} the nodes at distance at most $\ell$ of $z$ cannot be in $Z$ themselves, so $v^{left}_\ell, \ldots, v^{left}_1, v^{right}_1, \ldots, v^{right}_\ell$ are not in $Z$ and are also not active during this execution of \emph{Elect Maximums}. We will show that on both sides there is a constant probability of a node being in $Z$ and then argue that both sides are independent of one another. So lets for now just focus on the left side.
If any of the nodes in $v^{left}_{2\ell+1}, \ldots, v^{left}_{\ell+1}$(respectively $v^{right}_{\ell+1}, \ldots, v^{right}_{2\ell+1}$) were in $Z$ initially, then we are done. If this is not the case, than this means that $v^{left}_{\ell+1}$(respectively $v^{right}_{\ell+1}$) was active at the beginning of the execution. So by \cref{lem:ProbJoinZ} $v^{left}_{\ell+1}$ joined $Z$ with probability $\frac{1}{8\ell}$\footnote{In fact even with probability $\frac{1}{4\ell}$, by redoing the analysis of \cref{lem:ProbJoinZ} but conditioned on $v^{left}_1, v^{right}_1, \ldots, v^{right}_\ell$ not joining.} Notice also, that this is independent of the right side, as  $v^{left}_{\ell+1}$ and  $v^{right}_{\ell+1}$ are at least $2\ell$ apart and nodes decide to join $Z$ only considering their $2\ell$ neighborhood. So we get:
\[
Pr[z \text{ is done}\mid z \in Z] \geq Pr[v^{left}_{\ell+1} \text{ joins } Z] \cdot Pr[v^{right}_{\ell+1} \text{ joins } Z] \geq \frac{1}{(8\ell)^2}
\]
We also note that more active nodes in $A$ would only result in a higher probability.
\end{proof}

Now we use this to prove that also the nodes which do not join $Z$ are done quickly and that as a result nodes terminate after a constant number of iterations in expectation. This together with \cref{lem:FastCompressGivesConstant} proves \cref{thm:RandomConst}.
\begin{lemma}\label{lem:constantExpectation}
In every execution of \emph{Elect Maximums} every node has a constant probability of being done. Furthermore every node can terminate after a constant number of rounds in expectation when executing \emph{randomized Compress} \cref{alg:RandCompress}.
\end{lemma}
\begin{proof}
There are three different possibilities for every node $v$:
\begin{itemize}
    \item \textbf{Node $v$ is active:} Then by \cref{lem:ProbJoinZ,lem:ProbZDone} we get:
    \begin{align*}
        Pr[v \text{ is done}] \geq Pr[v \text{ joins } Z] Pr[v \text{ is done} \mid v \in Z] \geq \frac{1}{(8\ell)^3}
    \end{align*}
    \item \textbf{Node $v$ is already in $Z$:} Then by \cref{lem:ProbZDone} we have:
    \begin{align*}
        Pr[v \text{ is done}] = Pr[v \text{ is done} \mid v \in Z] \geq \frac{1}{(8\ell)^2}
    \end{align*}
    \item \textbf{Node $v$ is inactive:} Then $v$ can only be done, if it is between two nodes $z_1,z_2 \in Z$ that are themselves done. Since $v$ is inactive, this means that there is at least one node $z_r$ at distance less than $\ell$ that is in $Z$. Consider the $2\ell$ neighborhood of $v$ and let $A=(v^{left}_{2\ell}, \ldots, v^{left}_{1}, v, v^{right}_{1}, \ldots, v^{right}_{2\ell})$ be the path around $v$. The distinction between left and right is done arbitrarily based on the two outgoing edges of $v$ and just for the analysis. W.lo.g. let $z_r$ be among $v^{right}_{1}, \ldots, v^{right}_{\ell}$. We denote by $1\leq d:=\text{dist}(v,z_r) < \ell$ and refer the reader to \cref{fig:proofConstantExpectation} for an illustration of the scenario. As a result $v^{left}_1, \ldots v^{left}_{\ell-d-1}$ must also be inactive. So we consider two cases for the nodes $v^{left}_{2\ell - d }, \ldots, v^{left}_{\ell-d}$.
    \begin{itemize}
        \item \textbf{Case 1: } $Z \cap \{v^{left}_{2\ell - d }, \ldots, v^{left}_{\ell-d}\} \neq \emptyset$:\\
        Then let $z_l \in Z \cap \{v^{left}_{2\ell - d }, \ldots, v^{left}_{\ell-d}\}$ be the unique node in the set, because $Z$ is the subset of a ruling set according to \cref{lem:RulingSet}. Then the probability that $v$ is done is at least:
        \begin{align*}
            Pr[v \text{ is done}] \geq Pr[z_l \text{ is done} \mid v \in Z]Pr[z_r \text{ is done} \mid v \in Z] \geq \frac{1}{(8\ell)^4}
        \end{align*}
        \item \textbf{Case 2: } $Z \cap \{v^{left}_{2\ell - d }, \ldots, v^{left}_{\ell-d}\} = \emptyset$:\\
        Then at least $v^{left}_{\ell-d}$ must be active and therefore the probability that $v$ is done is at least:
        \begin{align*}
            Pr[v \text{ is done}] \geq Pr[v^{left}_{\ell-d} \text{ joins } Z]Pr[v^{left}_{\ell-d} \text{ is done} \mid v \in Z]Pr[z_r \text{ is done} \mid v \in Z] \geq \frac{1}{(8\ell)^5}
        \end{align*}
    \end{itemize}
\end{itemize}

\begin{figure}[h!]
    \centering
    \includegraphics[width=\textwidth]{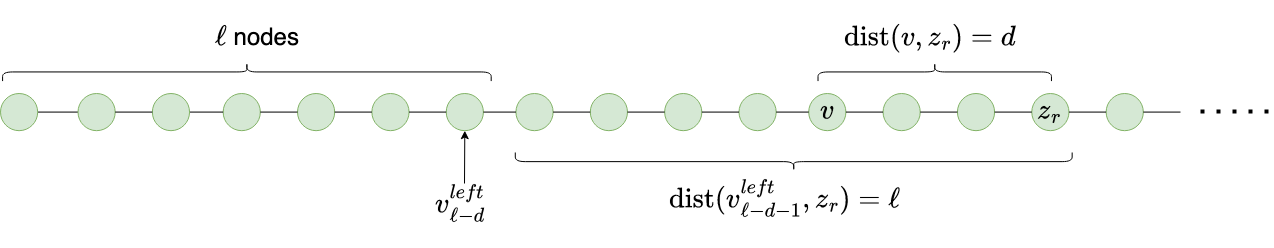}
    \caption{Illustrating the third case in the proof of \cref{lem:constantExpectation}.}
    \label{fig:proofConstantExpectation}
\end{figure}

So even in the worst case of the above three scenarios the probability is at least $\frac{1}{(8\ell)^5}$ as claimed.\\
For the furthermore part, we observe that the number of executions until a node is done is geometrically distributed. Let $p$ be the actual probability that a node $v$ is done after the execution of \emph{Elect Maximums}. Then the random variable $X_v$ of how many iterations of \emph{Elect Maximums} it takes for a node to terminate is geometrically distributed. As a result $\mathbb{E}[X_v] = \frac{1}{p} \leq (8\ell)^5$. This together with \cref{lem:ElectRuntime} proves the claim.
\end{proof}

\section{Improved Algorithms in the Polynomial Regime}\label{sec:polyupper}
In this section we show that, for infinitely many LCL problems with polynomial worst-case complexity, we can improve their node-averaged complexity. 
More precisely, in this section we show that, for a class of problems with worst-case complexity $\Theta(n^{1/k})$, we can provide an algorithm with node-averaged complexity $O(n^{1/(2^k - 1)})$.
As we show in \Cref{sec:lower}, this complexity is almost tight, since for all problems with worst-case complexity $\Theta(n^{1/k})$ we can show a lower bound of $\widetilde{\Omega}(n^{1/(2^k - 1)})$.

\subsection{\boldmath The Hierarchical $2 \frac{1}{2}$-Coloring Problems} \label{ssec:2_and_half_definition}
We now define a class of problems, already presented in \cite{CP19timeHierarchy}, called hierarchical $2 \frac{1}{2}$-coloring, that is parametrized by an integer $k \in \mathbb{Z}^+$. It has been shown in \cite{CP19timeHierarchy} that the problem with parameter $k$ has worst-case complexity $\Theta(n^{1/xk})$. We now give a formal definition of this class of problems and then we provide some intuition.

The set of input labels is $\Sigma_{\mathrm{in}}=\emptyset$. The set of output labels contains four possible labels, that is, $\Sigma_{\mathrm{out}} = \{W,B,E,D\}$, and these labels stand for \emph{white}, \emph{black}, \emph{exempt}, and \emph{decline}. Each node has a level in $\{1,\ldots,k+1\}$, that can be computed in constant time, and the constraints of the nodes depend on the level that they have.  The level of a node is computed as follows.
\begin{enumerate}
    \item Let $i \gets 1$.
    \item Let $V_i$ be the set of nodes of degree at most $2$ in the remaining tree. Nodes in $V_i$ are of level $i$. Nodes in $V_i$ are removed from the tree.
    \item Let $i \gets i+1$. If $i \le k$, continue from step 2.
    \item Remaining nodes are of level $k+1$.
\end{enumerate}
Each node must output a single label in $\Sigma_{\mathrm{out}}$, and based on their level, they must satisfy the following local constraints.
\begin{itemize}
    \item No node of level $1$ can be labeled $E$.
    \item All nodes of level $k+1$ must be labeled $E$.
    \item Any node of level $2\le i \le k$ is labeled $E$ iff it is adjacent to a lower level node labeled $W$, $B$, or $E$.
    \item Any node of level $1 \le i \le k$ that is labeled $W$ (resp.\ $B$) has no neighbors of level $i$ labeled $B$ (resp.\ $W$) or $D$. In other words, $W$ and $B$ are colors, and nodes of the same color cannot be neighbors in the same level.
    \item Nodes of level $k$ cannot be labeled $D$.
\end{itemize}
This problem can be expressed as a standard LCL (a type of LCLs defined in \Cref{sec:different-lcls}) by setting the checkability radius $r$ to be~$O(k)$, since in $O(k)$ rounds a node can determine its level and hence which constraints apply.

\subparagraph{Some intuition on these problems.}
In order to have a bit of intuition on this class of problems, let us consider the case when $k = 2$, for which the worst-case complexity is $\Theta(\sqrt{n})$. Let us focus on nodes of levels $1$ and $2$ (note that these are the only nodes that matter, since nodes of level $3$ will blindly output $E$). Notice that, by the definition of the level of a node, nodes of the same level form paths. Let $Q$ be a path of level-$1$ nodes. By the definition of the LCL problems described above, $Q$ must either be $2$-colored by using the labels $W$ and $B$, or all nodes in $Q$ must be labeled $D$. Now consider a path $P$ of level-$2$ nodes. By the definition of the LCL problems, a node in $P$ must be labeled $E$ if and only if it has a level-$1$ neighbor labeled with $W$ or $B$. On the other hand, the subpaths of $P$ induced by nodes having a level-$1$ neighbor labeled $D$ must be $2$-colored by using the labels $W$ and $B$.

On the lower bound side, a worst-case instance for this LCL with parameter $k = 2$ consists of a path $P$ of length $\Theta(\sqrt{n})$, where to each node $v_j$ of $P$ is attached a path $Q_j$ of length $\Theta(\sqrt{n})$. If an algorithm performs a $2$-coloring of any of the $Q_j$ paths, then it needs to spend $\Omega(\sqrt{n})$ rounds. Otherwise, if none of the $Q_j$ paths gets $2$-colored, then any correct algorithm must $2$-color $P$, spending $\Omega(\sqrt{n})$ rounds.

On the upper bound side, this LCL with parameter $k=2$ can be solved in $O(\sqrt{n})$ rounds in the following way. First, level-$1$ nodes spend $O(\sqrt{n})$ rounds to check if the path that they are in has length $O(\sqrt{n})$: if yes, then the path gets $2$-colored with labels $W$ and $B$; if no, the path gets labeled with $D$. If a level-$2$ node has a level-$1$ neighbor that is colored (i.e., it is labeled $W$ or $B$), then it outputs $E$. Finally, it is possible to prove that the subpaths induced by nodes of level $2$ that have a level-$1$ neighbor labeled $D$ must be of length $O(\sqrt{n})$, hence these subpaths can be $2$-colored in $O(\sqrt{n})$ rounds.

\subsection{Better Node-Averaged Complexity}\label{ssec:2_and_half_better}
We now show that, for the class of LCL problems described in \Cref{ssec:2_and_half_definition}, we can obtain a better node-averaged complexity. The algorithm is similar to the one presented in \cite{chang20} for the worst-case complexity, but it is modified to obtain a better node-averaged complexity (in \Cref{sec:lower} we show that this algorithm is tight up to a $\log n$ factor). 
\begin{theorem}
    The node-averaged complexity of the hierarchical $2 \frac{1}{2}$-coloring problem with parameter $k$ is $O(n^{1/(2^k - 1)})$.
\end{theorem}
\begin{proof}
    At first, all nodes spend $O(1)$ rounds to compute their level. Nodes of level $k+1$ output $E$. Then, the algorithm proceeds in phases, for $i$ in $1,\ldots,k$. In phase $i$, all nodes of level $i$ get a label, and hence let us assume that all nodes of levels $1, \ldots, i-1$ already have a label, and let us focus on level-$i$ nodes.

    Consider a node $v$ of level $i$. Node $v$ proceeds as follows. If $v$ has a neighbor from lower levels that is labeled $W$ or $B$ then $v$ outputs $E$. Otherwise, $v$ spends $t_i = c \cdot \gamma_i$ rounds to check the length of the path containing $v$ induced by nodes of level $i$, for some constant $c$ to be fixed later, and $\gamma_i = n^{2^{i-1}/(2^k - 1)}$. If this length is strictly larger than $t_i$, then $v$ outputs $D$. Otherwise, all nodes of the path are able to see the whole path, and hence they can output a consistent $2$-coloring by using the labels $W$ and $B$.

    The above algorithm correctly solves the problem if we assume that no nodes in level $k$ output $D$. In the following we show that indeed nodes of level $k$ do not output $D$, hence showing the correctness of the algorithm, and then we prove a bound on the node-averaged complexity.
    
    In order to do so, we first prove a useful statement. Let $S$ be the set of nodes of level $i$ that do not directly output $E$ at the beginning of phase $i$. It is possible to assign each node of level $j<i$ to exactly one node in $S$ such that to each node in $S$ are assigned $\Omega(n^{(2^{i-1} - 1) / (2^k - 1)})$ unique nodes of lower layers. Let $v$ be a node in $S$. Since $v\in S$, it means that $v$ is connected to a path $P$ of nodes of level $i-1$ of length strictly larger than $t_{i-1}$, and hence all nodes in $P$ are labeled $D$. Since, by the definition of the levels, $P$ has at most $2$ nodes of higher levels connected to it, then we can charge $t_{i-1}/2$ unique nodes to $v$. By repeating this reasoning inductively, we obtain that for each node of layer $i$ that does not directly output $E$, we can assign at least the following amount of nodes:
    \[
        \prod_{j=1}^{i-1} t_j / 2 = \Omega\left(\prod_{j=1}^{i-1} \gamma_j\right) = \Omega(n^{\sum_{j=1}^{i-1} 2^{j-1} / (2^k -1)}) = \Omega(n^{(2^{i-1} - 1) / (2^k - 1)}).
    \]
    Hence, the number of nodes that participate in phase $i$ is at most $O(n^{1 - (2^{i-1} - 1) / (2^k - 1)}) = O(n^{(2^k - 2^{i-1}) / (2^k - 1)})$. This implies that in phase $k$ the number of participating nodes is at most $O(n^{(2^k - 2^{k-1}) / (2^k - 1)}) = O(n^{2^{k-1} / (2^k - 1)}) = O(\gamma_k)$, where the hidden constant is inversely proportional to $c$. Hence, by picking $c$ large enough, we get that in $t_k$ rounds nodes of level $k$ see the whole path and thus no node of level $k$ outputs $D$, proving the correctness of the algorithm. 
    
    The total time spent during phase $i$ is bounded by 
        \[
        O(\gamma_{i} \cdot n^{1-(2^{i-1} - 1) / (2^k -1)}) = O( n^{2^{i-1}/(2^k - 1)} \cdot n^{1-(2^{i-1} - 1) / (2^k -1)} ) = O( n^{1 + 1/(2^k - 1)}).
    \]
    Therefore, the average time spent in phase $i$ is bounded by $O( n^{1/(2^k - 1)})$. Since this is done for $i \in \{1, \ldots, k\}$, and since $k = O(1)$, then the claim on the node-averaged complexity follows.
\end{proof}

\section{Lower Bounds in the Polynomial Regime}\label{sec:lower}

In this section we show that any LCL problem that requires polynomial time for worst-case complexity requires polynomial time also for node-averaged complexity. More precisely, we prove the following theorem.

\begin{theorem}
    Let $\Pi$ be an LCL problem with worst-case complexity $\Omega(n^{1/k})$ in the \LOCAL model. The randomized node-averaged complexity of $\Pi$ in the \LOCAL model is $\Omega(n^{1/(2^k - 1)} / \log n)$, and the deterministic node-averaged complexity of $\Pi$ in the \LOCAL model is $\Omega(n^{1/(2^k - 1)})$.
\end{theorem}
In order to prove this theorem, we proceed as follows (throughout this section we will use notions presented in \Cref{ssec:lcls}). It is known by \cite{chang20} that if an LCL problem $\Pi$ has worst-case complexity $o(n^{1/k})$, then it can actually be solved in $O(n^{1/(k+1)})$ deterministic rounds. This statement is proved by showing that it is possible to use an algorithm (possibly randomized) running in $o(n^{1/k})$ rounds to construct a good function $f_{\Pi,k+1}$ (that is, a function $f_{\Pi,k+1}$ that, if used, never creates empty classes), implying (as shown in \Cref{ssec:lcls}) the existence of a deterministic algorithm that solves $\Pi$ and has worst-case complexity $O(n^{1/(k+1)})$. In this section we show that it is possible to construct a good function $f_{\Pi,k+1}$ by starting from an algorithm $\mathcal{A}$ with randomized node-averaged complexity $o(n^{1/(2^k - 1)} / \log n)$ or deterministic node-averaged complexity $o(n^{1/(2^k - 1)})$. By \Cref{ssec:lcls}, this implies that if there exists an algorithm with $o(n^{1/(2^k - 1)} / \log n)$ randomized node-averaged complexity or with $o(n^{1/(2^k - 1)})$ deterministic node-averaged complexity, then there exists an algorithm with deterministic worst-case complexity $O(n^{1/(k+1)})$. This shows that any LCL with worst-case complexity $\Omega(n^{1/k})$ has node-averaged complexity at least $\Omega(n^{1/(2^k - 1)} / \log n)$ randomized and $\Omega(n^{1/(2^k - 1)})$ deterministic. While we will use some ideas already presented in \cite{chang20}, handling an algorithm with only guarantees on its node-averaged complexity arises many (new) issues that we need to tackle.

Our statement will be proved even for the case in which the randomized algorithm $\mathcal{A}$ satisfies the weakest possible assumptions, i.e., the assumptions are so relaxed that they are satisfied by any deterministic algorithm, any randomized Las Vegas algorithm, and any randomized Monte Carlo algorithm. The assumptions are the following.
\begin{itemize}
    \item We assume that $\mathcal{A}$ is a randomized algorithm that is only required to work when the unique IDs of nodes are assigned at random, among all possible valid assignments.
    \item We assume that $\mathcal{A}$ is an algorithm that fails with probability at most $1/n^c$ for any chosen constant $c \ge 1$.
    \item We assume that the bound on the node-averaged complexity of $\mathcal{A}$ holds with probability at least $1 - 1/n^c$ for any chosen constant $c \ge 1$.
\end{itemize}
However, in the following, we will assume that the bound on the node-averaged complexity of $\mathcal{A}$ always holds. In fact, observe that we can always convert an algorithm with node-averaged complexity $T$ that holds with probability at least $1 - 1/n^c$ into an algorithm with node-averaged complexity $O(T)$ that holds always, since, even for $c=1$, we can safely assume that when the bound does not hold (that happens with probability at most $1/n$), the runtime is anyways bounded by $n$ (since everything can be solved in $n$ rounds in the LOCAL model).

\subsection{Good Functions}\label{ssec:good-functions}
We now describe how, in \cite{CP19timeHierarchy,chang20}, it is determined whether a good function exists. In the following, we call \emph{solver} the algorithm for solving an LCL given a decomposition, presented in \Cref{ssec:lcls}.

There are two main ingredients that are used in \cite{CP19timeHierarchy,chang20}. One is that there are a finite amount of possible functions, and this is true since the classes of \Cref{def:classes} have finite size, and since there is a finite amount of ways to map maximal classes into independent ones. The other is that, given a function, there is a centralized algorithm that checks whether the function is good. The key question is how to test whether a function is good, and, on a high level, this is done by constructing all possible label-sets that could possibly appear while running the solver, which depend on the tested function. Crucially, this is a finite amount, and there is a recursive procedure that can generate them.

\subparagraph{What actually determines the complexity of a problem.}
We now provide some more intuition about the function $f_{\Pi,k}$, and about how the existence of this function is related with the worst-case complexity of a problem. 
In \cite{CP19timeHierarchy,chang20} it is shown how to determine whether a good function exists (that is, a function that never creates empty classes), and if it exists, how to construct it. Also, it is shown that:
\begin{itemize}
    \item If a problem is solvable in $O(\log n)$ rounds, then a good function $f_{\Pi,\infty}$ exists.
    \item If a problem is solvable in $o(n^{1/k})$, then a good function $f_{\Pi,k+1}$ exists.
\end{itemize}
This implies that, if a problem is solvable in $O(\log n)$ rounds, then we can automatically find a good function that makes the solver work and solve the problem in $O(\log n)$ rounds, and if a problem is solvable in $o(n^{1/k})$, then we can automatically find a good function that makes the solver work and solve the problem in $O(n^{1/(k+1)})$ rounds.
Also, given a problem $\Pi$, it is possible to compute what is the optimal target complexity, that is, we can determine whether the problem can be solved in $O(\log n)$ rounds, and if the answer is negative we can determine the best integer $k$ for which the problem can be solved in $O(n^{1/k})$ rounds.

The intuition about what determines the existence of a good function $f_{\Pi,k}$ is the following. Compress paths are something that is difficult to handle, because they require to restrict the label-sets that we propagate up in order to make them ``independent''. The number of compress layers that we can recursively handle is what determines the complexity of a problem.
\begin{itemize}
    \item If we can handle an arbitrary amount of compress layers, then we can construct a good function $f_{\Pi,\infty}$, and hence the problem can be solved in $O(\log n)$ rounds.
    \item If we can handle only a constant amount of compress layers, say $k-1$, then a good function $f_{\Pi,k}$ exists, but $f_{\Pi,k+1}$ does not, and then the complexity of the problem is $\Theta(n^{1/k})$.
\end{itemize}

\subparagraph{Testing procedure.} We now present an algorithm that tests whether a function $f_{\Pi,k}$ is good. We call this algorithm \emph{testing procedure}.
The procedure depends not only on the function to be tested, but also on a parameter $\ell$ that, in \cite{CP19timeHierarchy,chang20}, is shown that it can be determined solely as a function of $\Pi$. We observe that this algorithm is well-defined also when $k = \infty$, and in fact this algorithm can be used also to determine whether a problem can be solved with $O(\log n)$ worst-case complexity. The idea of the testing procedure is to keep track of all label-sets that one could possibly obtain while running the solver.
For each of these label-sets, we also keep track of a subtree (where nodes are also marked with the layers of a decomposition) where, if we run the solver by using the function that we are testing, we would obtain an edge with such a label-set. While this is not necessary for \emph{testing} a function, we will use these trees later for \emph{constructing} a function given an algorithm.
We now formally describe the testing procedure and later we will give more intuition on it.
\begin{enumerate}
    \item Initialize $S$ with all the possible values of the label-set $g(v)$ of $v$ (as defined in \Cref{def:computing-label-set}) that could be obtained when $v$ is a leaf. Note that the possible values are a finite amount that only depends on the amount of input labels of $\Pi$. Initialize $\mathcal{R}_1$ by inserting one pair $((\tilde{T},u),L)$ for each element $L$ in  $S$, where $\tilde{T}$ is a tree composed of $2$ nodes $\{u,v\}$ and $1$ edge $\{u,v\}$, and $L=g(v)$. Node $v$ is marked as a rake node of layer $1$, while $u$ is marked as a temporary node.\label{item:leaves}
    
    \item For $i = 1, \ldots, k$ do the following. If, at any step, an empty label-set is obtained, then the tested function is not good.\label{item:tryall}
    
    \begin{enumerate}
        \item Do the following in all possible ways. Consider $x$ arbitrary elements $((\tilde{T}_j,v_j),L_j)$ of $\mathcal{R}_i$, where $1\le j\le x$ and $1 \le x \le \Delta$. Construct the tree $T$ as the union of all trees $\tilde{T}_j$, where all the nodes $v_j$ (note that each node $v_j$ has degree $1$) are identified as a single node, call it $v$, which, after this process, has degree $x$ in $T$. Let $F_{\mathrm{incoming}}$ be the set of edges connected to $v$, and let $\mathcal{L}_{\mathrm{incoming}}$ be the label-set assignment given by the sets $L_j$. The node $v$ is marked as a rake node of layer $i$.
        If $v$ has an empty maximal class w.r.t. $F_{\mathrm{incoming}}$, $F_{\mathrm{outgoing}} = \{\}$, and $\mathcal{L}_{\mathrm{incoming}}$, then the tested function is not good.\label{item:top-level-can-complete}
        
        \item Do the following in all possible ways. Consider $x$ arbitrary elements $((\tilde{T}_j,v_j),L_j)$ of $\mathcal{R}_i$, where $1\le j\le x$ and $1 \le x \le \Delta -1$. Construct the tree $T$ as the union of all trees $\tilde{T}_j$, where all the nodes $v_j$ are identified as a single node, call it $v$. Attach an additional neighbor $u$ to $v$. Let $F_{\mathrm{outgoing}} = \{\{u,v\}\}$. Let $F_{\mathrm{incoming}}$ be the set of edges connected to $v$, excluding $\{u,v\}$, and let $\mathcal{L}_{\mathrm{incoming}}$ be the label-set assignment given by the sets $L_j$. The node $v$ is marked as a rake node of layer $i$, while $u$ is marked as a temporary node.
        Let $L = g(v)$ (as defined in \Cref{def:computing-label-set}). If $L$ is empty, then the function is not good. Add $((T, u), L)$ to $\mathcal{R}_i$ if no pair with second element $L$ is already present.\label{item:new-rake-labelsets}
        
        \item Repeat the previous two step until nothing new is added to $\mathcal{R}_i$. This must happen, since there are a finite amount of possible label-sets. \label{constructfunction-b}
        
        \item If $i = k$, stop.
        
        \item Initialize $\mathcal{C}_i = \emptyset$.
        
        \item Do the following in all possible ways. Construct a graph starting from a path $H$ of length between $\ell$ and $2 \ell$ where we connect nodes of degree $1$ to the nodes of $H$ satisfying: (i) all nodes in $H$ have degree at most $\Delta$; (ii) the two endpoints of $H$ have an outgoing edge that connects respectively to nodes $u_1$ and $u_2$ that are nodes of degree $1$; (iii) all the other edges connecting degree-$1$ nodes to the nodes of $H$ are incoming for $H$.
        Next, replace each incoming edge $e$ and the node of degree $1$ connected to it with a tree $\tilde{T}$ of a pair $((\tilde{T},u),L)$ in $\mathcal{R}_i$, by identifying $u$ with the node of the path connected to $e$. Different trees can be used for different edges. The nodes $u_1$ and $u_2$ are marked as temporary nodes, while the nodes of the path are marked as compress nodes of layer $i$.
        Use the function as described in \Cref{def:computing-label-set} to compute the label-sets $L_1$ and $L_2$ of the two endpoints. If $L_1$ or $L_2$ is empty, then the function is not good.
        Otherwise, add the pair $((H,u_1),L_1)$ (resp.\ $((H,u_2),L_2)$) to $\mathcal{C}_i$ if no pair with second element $L_1$ (resp.\ $L_2$) is already present. The \emph{representative tree} of $P = (H,F_{\mathrm{incoming}},F_{\mathrm{outgoing}},\mathcal{L}_{\mathrm{incoming}})$ is defined as $r(P) = T$.\label{item:compress-paths}
        
        \item Set $\mathcal{R}_{i+1} = \mathcal{R}_i \cup \mathcal{C}_i$. If $\mathcal{R}_{i+1} = \mathcal{R}_{i}$, stop.
    \end{enumerate}
\end{enumerate}
On a high level, the testing procedure does the following. 
\Cref{item:leaves} computes all possible label-sets that we can get from leaf nodes; leaf nodes are then marked as rake nodes. 
\Cref{item:top-level-can-complete} considers all possible ways to assign a label-set to all the incident edges of a node $v$, and checks, for each case, if there is a way for $v$ to label its incident edges (with a label from the provided label-sets) such that the constraints of node $v$ are satisfied. Node $v$ is marked as a rake node. This case corresponds to a rake node with no neighbors in higher layers.
\Cref{item:new-rake-labelsets} considers all possible ways to assign a label-set to all but one incident edge of a node $v$; let $e$ be the edge that does not have a label-set assigned. Then, the label-set of $e$ is computed as a function of all the label-sets of the other edges incident to $v$, and then $v$ is marked as a rake node. This case corresponds to a rake node with one neighbor in higher layers.
In \Cref{constructfunction-b} the last two steps are repeated until no new label-sets are obtained.
In other words, items \ref{item:top-level-can-complete}, \ref{item:new-rake-labelsets}, and \ref{constructfunction-b} generate all possible label-sets that we can get by performing only rakes, that is, before performing the first compress.
Finally, \Cref{item:compress-paths} considers all possible label-sets that we can obtain for the endpoints of a compress path. By repeating recursively the procedure interleaving between rake and compress nodes, and by checking at any point that we do not get empty classes or empty label-sets, we can test whether a function is good or not. 
An example of a tree generated by this procedure is shown in \Cref{fig:testing}.

It is possible to prove that the testing procedure generates exactly those label-sets that could possibly be obtained by running the solver \cite{CP19timeHierarchy, chang20}. Hence, if empty classes or empty label-sets are never obtained, then the function can indeed be used to solve a problem.
Observe that all the edges of a representative tree can be oriented such that: all nodes have at most one outgoing edge; any directed path contains layers in non-decreasing order (see \Cref{fig:testing}). While in a generic tree decomposition this may be false, this is not an issue: the obtained label-sets, in this restricted set of cases, are still all the possible ones that could be obtained on an arbitrary tree decomposition.

\begin{figure}
	\centering
	\includegraphics[width=0.5\textwidth]{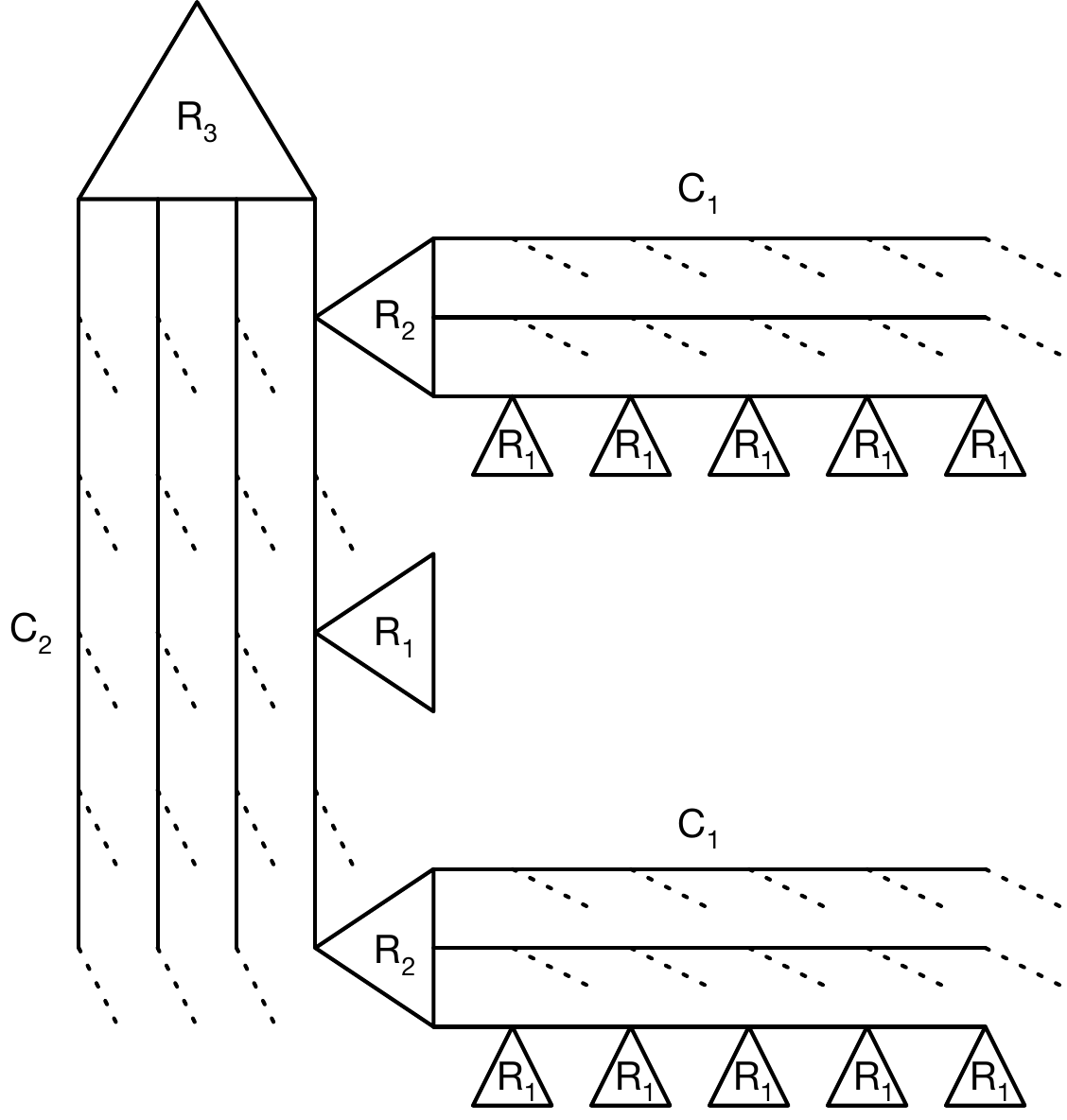}
	\caption{An example of a tree generated by the testing procedure. Triangles marked $R_i$ correspond to trees of constant size containing only nodes marked rake of layer $i$. Paths marked $C_i$ correspond to nodes marked compress of layer $i$. At the end of each path, on the side not connected to a triangle, is connected a temporary node, not shown in the picture. The top node of the triangle marked $R_3$ may be a temporary node.}
	\label{fig:testing}
\end{figure}

\subsection{Some Useful Ingredients}
As already discussed, the main goal of this section is to show that it is possible to construct a good function $f_{\Pi,k+1}$ by starting from an algorithm $\mathcal{A}$ with randomized node-averaged complexity $o(n^{1/(2^k - 1)} / \log n)$ or deterministic node-averaged complexity $o(n^{1/(2^k - 1)})$.
Before diving into that, we provide some ingredients that will later be useful. 

\subparagraph{Pumping lemma for trees.}
We will use a fundamental ingredient provided in \cite{CP19timeHierarchy}, that, informally, allows us to take a compress path of some class, and make the path longer while preserving its class.
\begin{lemma}[Pumping lemma for trees \cite{CP19timeHierarchy}]\label{lem:pump}
Assume we are given an LCL $\Pi = (\Sigma_{\mathrm{in}},\Sigma_{\mathrm{out}},C_W,C_B)$ in the black-white formalism.
Let a \emph{compress path} be a tuple $(H,F_{\mathrm{incoming}},F_{\mathrm{outgoing}},\mathcal{L}_{\mathrm{incoming}})$ satisfying the following.
\begin{itemize}
    \item To each node $v_i$ of $H$ are connected $d_i$ neighbors, $u_{i,1},\ldots,u_{i,d_i}$. $F_{\mathrm{incoming}}$ is the set of edges connecting the nodes of $H$ with such neighbors. 
    \item To each edge $e$ of $F_{\mathrm{incoming}}$ is assigned a label-set $L_e \subseteq \Sigma_{\mathrm{out}}$. $\mathcal{L}_{\mathrm{incoming}} = (L_e)_{e \in F_{\mathrm{incoming}}}$ is this assignment.
    \item To each endpoint of $H$ is connected one additional neighbor. $F_{\mathrm{outgoing}}$ is the set containing the two edges connecting the endpoints of $H$ to such neighbors.
\end{itemize}
Let $(H,F_{\mathrm{incoming}},F_{\mathrm{outgoing}},\mathcal{L}_{\mathrm{incoming}})$ be a compress path where $H$ has length $x$. Let $w$ be an arbitrary integer that is at least $x$. 
There exists a constant $\ell$ such that, if $x \ge \ell$, then it is possible to construct a compress path $(H',F'_{\mathrm{incoming}},F'_{\mathrm{outgoing}},\mathcal{L}'_{\mathrm{incoming}})$  satisfying the following.
\begin{itemize}
    \item The length of $H'$ is at least $w$ and at most $w + \ell$.
    \item The maximal class of $H'$ w.r.t.\ $\Pi$, $F'_{\mathrm{incoming}}$, $F'_{\mathrm{outgoing}}$, and $\mathcal{L}'_{\mathrm{incoming}}$ is equal to the maximal class of $H$ w.r.t.\ $\Pi$, $F_{\mathrm{incoming}}$, $F_{\mathrm{outgoing}}$, and $\mathcal{L}_{\mathrm{incoming}}$.
    \item The set of label-sets contained in $\mathcal{L}'_{\mathrm{incoming}}$ is equal to the set of label-sets contained in $\mathcal{L}_{\mathrm{incoming}}$.
\end{itemize}
\end{lemma}

\subparagraph{Pumping trees.}
Let $T$ be a tree constructed in the testing procedure. Observe that, in this tree, the nodes are either marked rake, compress, or temporary, and with a layer number. We provide a procedure $\mathrm{pump}_k(w,T)$ that modifies $T$ to make all compress paths longer. In particular, all compress paths of layer $i$ will be of length between $w_i$ and $w_i + \ell$, where $w_i$ is defined to be $w^{2^{i-1} / (2^k - 1)}$. The procedure $\mathrm{pump}_k(w,T)$ performs the following operation recursively, by starting from compress layer $i=1$ and going up, and by considering all compress paths present in $T$.

\begin{itemize}
    \item Let $(H,F_{\mathrm{incoming}},F_{\mathrm{outgoing}},\mathcal{L}_{\mathrm{incoming}})$ be a compress path induced by a connected component of nodes marked as compress nodes of layer $i$.

    \item Apply \Cref{lem:pump} with parameter $w_i$ and get $(H',F'_{\mathrm{incoming}},F'_{\mathrm{outgoing}},\mathcal{L}'_{\mathrm{incoming}})$.

    \item Remove all the subtrees that are connected to $H$ via edges in $F_{\mathrm{incoming}}$.

    \item Replace $H$ with $H'$.

    \item Let $\{u,v\}$ be an edge in $F'_{\mathrm{incoming}}$, where $u$ is the node of $H'$. Note that $\{u,v\}$ has a label-set assigned according to $\mathcal{L}'_{\mathrm{incoming}}$. For each edge $\{u,v\}$ in $F'_{\mathrm{incoming}}$ (where $u$ is the node of $H'$) we attach a copy of one subtree that was connected to $H$ satisfying that the label-set assigned to the edge connecting the subtree to $H$ is the same as the label-set of the edge $\{u,v\}$.
\end{itemize}
Note that the last operation is possible since the set of the label-set in $\mathcal{L}'_{\mathrm{incoming}}$ is the same as the set of the label-sets in $\mathcal{L}_{\mathrm{incoming}}$, and hence the label-set of $\{u,v\}$ must be a label-set that occurs in some edge in $F_{\mathrm{incoming}}$ that connects some subtree to a node of $H$.

\begin{lemma}\label{lem:sizeRepr}
    Let $T$ be a representative tree. Let $T'=\mathrm{pump}_k(w,T)$. $T'$ contains $O(w)$ nodes.
\end{lemma}
\begin{proof}
    By construction of representative trees (see \Cref{fig:testing}), the number of nodes is bounded by
    $O(\prod_{i=1}^k w_i) = O(w^{\sum_{i=1}^{k} 2^{i-1} / (2^k - 1)}) = O(w)$.
\end{proof}

\subparagraph{Adding nodes.}
We now provide a procedure that takes as input a value $n$ and a pumped tree of at most $n$ nodes containing at least one temporary node, and adds nodes to the tree in order to achieve that the amount of nodes are exactly $n$. The procedure $\mathrm{add}(n,T)$ is defined as follows.
Let $j$ be the maximum layer index of temporary nodes in $T$, and let $v$ be an arbitrary temporary node $v$ of layer $j$.
We attach a path to $v$ in such a way that the number of nodes becomes exactly $n$.

\subparagraph{Optimal form of an algorithm.}
It has been shown in \cite{Feuilloley17} that any algorithm can be normalized in a way that nodes do not spend useless rounds of computation. While such a result has been proved on graphs with linearly bounded growth, it can be easily adapted, in the case of LCLs in the black-white formalism, to be more general.
\begin{lemma}[\cite{Feuilloley17}]\label{lem:neighborsrun}
    Let $A$ be an algorithm that solves an LCL problem $\Pi$ with node-averaged complexity $T$. Then, if a node $v$ runs for $t$ rounds, either there exist at least $t/2 - 1$ nodes in the $t/2$-radius neighborhood of $v$ that run for at least $t/2$ rounds, or $v$ could terminate earlier.
\end{lemma}
\begin{proof}[Proof Sketch]
    Suppose that in the $t$-radius neighborhood of $v$ there are at most $t/2 - 2$ nodes that run for at least $t/2$ rounds. This means that the connected component $C$ induced by $v$ and nodes that still have to terminate after $t/2$ rounds is of size strictly less than $t/2$. Also, in strictly less than $t$ rounds,  $v$ sees the state that all these nodes had at time $t/2$, and of their neighbors. We show that this implies that $v$ can terminate in strictly less than $t$ rounds, proving the claim. Since $C$ is surrounded by nodes that already terminated, then $v$, solely as a function of its state, the state that the nodes of $C$ had at time $t/2$, and the output of the nodes connected to $C$, can simulate the execution of $A$ for all nodes of $C$ without additional communication, and obtain an output for all the nodes in $C$, and hence $v$ can terminate in strictly less than $t$ rounds. 
\end{proof}

\subparagraph{From node-averaged complexity to expected runtime.}
In the following, we assume $\mathcal{A}$ to be an algorithm for $\Pi$ that has been normalized according to \Cref{lem:neighborsrun}. We now show that a stronger notion of node-averaged complexity must hold for $\mathcal{A}$, namely that the expected running time of each node is bounded. In the case of randomized algorithms, the expectation is taken over all possible assignments of random bits and IDs, while for deterministic algorithms the expectation is taken over all possible assignments of IDs (if we imagine to run the deterministic algorithm with an ID assignment chosen uniformly at random among the valid ones).
\begin{lemma}\label{lem:expectedruntime}
    Let $\mathcal{A}$ be an algorithm with node-averaged complexity at most $w^{1/(2^k - 1)} / f(n)$, where:
    \begin{itemize}
        \item $w = w(n)$ is in $O(n)$;
        \item for all $n$ that are at least a large enough constant $n'$, $f(n)$ is at least a large enough constant.
    \end{itemize}
    Let $T = \mathrm{add}(n,\mathrm{pump}_k(w,T'))$, where $T'$ is a representative tree. There exists a constant $c > 0$ satisfying the following.
    For an arbitrary $i$, let $v$ be an arbitrary node in a compress layer $i$ satisfying that $v$ is at distance at least $c \cdot w^{2^{i-1}/(2^k - 1)} / f(n)$ from the endpoints of the compress path.
    
    Assume that $\mathcal{A}$ is run on $T$. Then, the expected running time of node $v$ is at most $c \cdot w^{2^{i-1}/(2^k - 1)} / f(n)$. 
\end{lemma}
\begin{proof}
    Assume, for a contradiction, that there exists a node $v$ in a compress layer $i$ with expected running time strictly larger than $t = c \cdot w^{2^{i-1}/(2^k - 1)} / f(n)$, and that is at distance at least $t$ from the endpoints. Let $B$ be the $t$-radius neighborhood of $v$. Observe that, by construction (see \Cref{fig:testing}), this neighborhood contains $y = O(t \cdot \prod_{j=1}^{i-1} w^{2^{j-1}/(2^k - 1)}) = O(c \cdot w^{(2^i - 1)/(2^k -1)} / f(n))$ nodes. Moreover, for $n$ large enough, this neighborhood does not contain all nodes of $T$ ($n'$ is set to be the smallest value of $n$ for which this is true). This implies that we can make $x = \Omega(n / y)$ copies of $B$ and connect them in such a way that each copy of $v$ has the same view in the constructed tree $T$ and in $B$. Let $U$ be the resulting tree. 
    By using \Cref{lem:neighborsrun}, we obtain that, by running $\mathcal{A}$ on $U$, the average of the expected running times is strictly larger than $x \cdot (t/2)(t/2 - 1)/n=\Omega(t^2/y)$, that, for some constant $c''$, is at least,
    \[
        c'' \cdot \frac{c^2 \cdot w^{2^{i}/(2^k - 1)} / f(n)^2}{c \cdot w^{(2^i - 1)/(2^k -1)} / f(n)} = \frac{c'' \cdot c}{f(n)} \cdot w^{(2^i - 2^i + 1) / (2^k - 1)} = c'' \cdot c \cdot w^{1 / (2^k - 1)} / f(n).
    \]
    This implies that the average running time is strictly larger than $c'' \cdot c \cdot w^{1 / (2^k - 1)} / f(n)$. By setting $c = c' / c''$, we get a contradiction on the assumption on the runtime of $\mathcal{A}$.
\end{proof}

For randomized algorithms, we will use \Cref{lem:expectedruntime} for the case $f(n) = \Theta(\log n)$, for which we get the following corollary.
 
\begin{corollary}\label{cor:expectedruntime}
    Let $\mathcal{A}$ be a randomized algorithm with node-averaged complexity at most $c' \cdot w^{1/(2^k - 1)} / \log n$, where $w = w(n)$ is in $O(n)$, and $c'$ is a small enough constant. Let $T = \mathrm{add}(n,\mathrm{pump}_k(w,T'))$, where $T'$ is a representative tree. 

    There exists a constant $c > 0$ directly proportional to $c'$ satisfying the following.
    For an arbitrary $i$, let $v$ be an arbitrary node in a compress layer $i$ satisfying that $v$ is at distance at least $c \cdot w^{2^{i-1}/(2^k - 1)} / \log n$ from the endpoints of the compress path.
    
    Assume that $\mathcal{A}$ is run on $T$. Then, the expected running time of $v$ is at most $c \cdot w^{2^{i-1}/(2^k - 1)} / \log n$, for all $n \ge n'$ for some constant $n'$.
\end{corollary}
\begin{proof}
    This corollary follows by setting $f(n) := \log n /c'$ in \Cref{lem:expectedruntime}.
\end{proof}

As already mentioned, \Cref{lem:expectedruntime} holds also if $\mathcal{A}$ is a deterministic algorithm that is run on a random ID assignment. However, by the definition of deterministic node-averaged complexity, the bound on the node-averaged complexity must hold for any possible valid ID assignment. Moreover, for deterministic algorithms, we will use $f(n) = \Theta(1)$. We show the following corollary.

\begin{corollary}\label{cor:expectedruntimedet}
Let $\mathcal{A}$ be a deterministic algorithm with node-averaged complexity at most $c' \cdot w^{1/(2^k - 1)}$, where $w = w(n)$ is in $O(n)$, and $c'$ is a small enough constant.
    Let $T = \mathrm{add}(n,\mathrm{pump}_k(w,T'))$, where $T'$ is a representative tree. Let $\mathcal{X}$ be an arbitrary subset of the possible ID assignments in $T$.

There exists a constant $c > 0$ directly proportional to $c'$ satisfying the following.
    For an arbitrary $i$, let $v$ be an arbitrary node in a compress layer $i$ satisfying that $v$ is at distance at least $c \cdot w^{2^{i-1}/(2^k - 1)}$ from the endpoints of the compress path.
    Assume that there exist at least $n$ ID assignments $X_1,\ldots,X_n$ in $\mathcal{X}$ satisfying the following. Let $I_i$ be the set of IDs used by nodes at distance at most $c \cdot w^{2^{i-1}/(2^k - 1)}$ from $v$ when using the ID assignment $X_i$. For all $i \neq j$, the sets $I_i$ and $I_j$ are disjoint.
    
    Then, there exists an ID assignment from $\mathcal{X}$ satisfying that, by running $\mathcal{A}$ with such an assignment, the runtime of $v$ is bounded by $c \cdot w^{2^{i-1}/(2^k - 1)}$.
\end{corollary}
\begin{proof}
Assume for a contradiction that there are $n$ ID assignments satisfying the requirements, and for all such assignments, $v$ has runtime strictly larger than $c \cdot w^{2^{i-1}/(2^k - 1)}$.
Similarly as in the proof of \Cref{lem:expectedruntime}, we can create an instance where the average runtime of the nodes is too high, reaching a contradiction.
In particular, as in the proof of \Cref{lem:expectedruntime}, we can construct an instance by copying the neighborhood of $v$ for enough times so that the average runtime is high. However, differently from the randomized case, we need to ensure that the constructed instance has unique IDs. By the assumption on the IDs, we can assign different IDs in each copy.
\end{proof}

\subparagraph{Designated nodes, edges, and neighborhoods.}
Let $T'$ be a representative tree, and let $T = \mathrm{add}(n,\mathrm{pump}_k(w,T'))$.
For each compress path $H$ in $T$, we define the \emph{designated node} of $H$ to be the white node $v$ in the middle of the path $H$ (breaking ties arbitrarily). Moreover, we define the \emph{designated edge} of the path $H$ to be an arbitrary edge in $H$ that is incident to $v$.

Let $v$ be a designated node of level $i$. We define the designated neighborhood of $v$ w.r.t.\ a constant $\beta$ and a function $g$ as the set $S$ containing the following nodes. Initialize $S$ with $v$ and all the nodes at distance at most $w_i / (\beta g(n))$ from $v$ in the compress path containing $v$. Then, start from a node $u$ in $S$ and add to $S$ all its neighbors $z$ of strictly lower layers. If $z$ is in a compress path of layer $j$, then also add all the nodes at distance at most $w_j / \beta$ from $z$ in the compress path containing $z$. Repeat this operation until nothing new is added.

\subparagraph{Independent executions.}
We now prove a statement about randomized algorithms running in $T(n) = o(w^{1/(2^k -1)}/\log n )$, where $w = w(n)$ is in $O(n)$. For this purpose, we consider the set $\mathcal{S}$ of designated neighborhoods w.r.t.\ $g(n) = \log n$ and some constant $\beta$ to be fixed later.
More in detail, we use \Cref{cor:expectedruntime} to show that,  with large enough probability, it holds that any node of the graph, within its running time, is not able to communicate with two nodes that are part of different designated neighborhoods in $\mathcal{S}$.
\begin{lemma}\label{lem:onemodified}
    Let $\mathcal{A}$ be a randomized algorithm with node-averaged complexity $o(w^{1/(2^k -1)}/\log n )$, where $w = w(n)$ is in $O(n)$.
    Then, for any arbitrary constant $c$, with probability at least $1 - 1/n^c$, a node is able to communicate with the nodes of at most one designated neighborhood in $\mathcal{S}$.
\end{lemma}
\begin{proof}
    Let $P$ be a compress path of layer $i$ and let $v$ be the designated node of $P$. By construction, $P$ is of length at least $w_i$, and $P$ is connected to a node of a higher layer on at most one side (see \Cref{fig:testing}). Assuming the higher layer exists, let $u$ be the last node of $P$ on that side. We prove that there exists a node $z$ in $P$ with the following properties: $z$ lies between $u$ and $v$; with high probability, $z$ does not see neither $u$ nor $v$, nor nodes at distance at most $w_i / \beta$ from them. Showing the existence of such a node $z$ would imply the claim, since, by definition, a designated neighborhood, defined as a function of a designated node of layer $i$, contains:
    \begin{itemize}
        \item some nodes $V$ that are close to the center of a compress path of layer $i$;
        \item all nodes that are within some bounded distance of nodes in $V$; by construction, if these nodes lie on some lower-layer compress path, then they are close to the beginning of such a path.
    \end{itemize}
    Hence, for a node to see two different designated neighborhoods, it must be able to see, at the same time, both the beginning and the middle of the same compress path, which is not possible if $z$ exists.

    Now, what is left, is to show the existence of node $z$. Let $z$ be the node at distance $w_i / 4$ from $v$ that lies in the path connecting $u$ and $v$. Let $P'$ be the path containing $z$ induced by nodes at distance strictly larger than $w_i / \beta$ from both $v$ and $u$. Assume $\beta \ge 8$, we get that $P'$ has length at least $w_i / 4$. Consider the subpath $P''$ of $P'$ that lies between $z$ and $u$, which has length at least $w_i / 8$. We prove that in $P''$, with high probability, there is a node that terminates in a runtime that is strictly less than half the length of $P''$, implying that $z$ cannot see $u$. The same proof holds for the path between $z$ and $v$. 

    We can split $P''$ into $\delta \log n$ disjoint subpaths of length at least $ w_i / (16 \delta \log n) $, for any chosen constant $\delta$. For each subpath $P_i$, let $v_i$ be the central node of the path, breaking ties arbitrarily.

    By applying \Cref{cor:expectedruntime}, we get that the expected runtime of each node $v_i$ is at most $t = w^{2^{i-1} / (2^k - 1)} / (\alpha \log n)$, for any chosen constant $\alpha$.
    The subpaths are of length at least \[
    w_i / (16 \delta \log n) = w^{2^{i-1}/(2^k - 1)}  / (16 \delta \log n).
    \]
    Hence, we can choose $\alpha$ large enough, as a function of $\delta$, so that the expected runtime of the nodes $v_i$ is less than $1/4$ of the length of $P_i$. Observe that, by the Markov inequality, the probability that a node $v_i$ runs for more than $1/2$ of the length of $P_i$ is at most $1/2$, and also observe that these events are independent for each $v_i$. Finally, note that, in order to be possible for $z$ to communicate with $u$, it must hold that all nodes $v_i$ run for strictly more than $1/2$ of the length of $P_i$, and this happens with probability at most $1/2^{\delta \log n} = 1 / n^\delta$.
\end{proof}

\begin{lemma}\label{lem:goodbitsrand}
    Let $\mathcal{A}$ be a randomized algorithm with node-averaged complexity $o(w^{1/(2^k -1)}/\log n )$ and failure probability at most $1 / n^c$, where $w = w(n)$ is in $O(n)$.
    For each designated neighborhood $S \in \mathcal{S}$, let $E^1_{S}$ be the event satisfying that, by running $\mathcal{A}$ on the designated node $v$ in $S$, $v$ does not communicate with any node outside $S$. Note that  $E^1_{S}$ solely depends on the random bits (including random IDs) assigned to the nodes of $S$.
    Moreover, let $E^2_S$ be an event 
    that solely depends on the random bits assigned to the nodes of $S$,
    and satisfying $P[E^2_{S}\mid E^1_{S}]=p_2\ge 2 / n^c $.
    Then, there exists a non-empty subset of possible random bits assignments satisfying the following.
    \begin{itemize}
        \item For each $S \in \mathcal{S}$, the events $E^1_S$ and $E^2_S$ hold.
        \item The failure probability of $\mathcal{A}$, when run using such random bits assignments, is at most $16 / (p_2 \cdot n^c)$.
    \end{itemize}
\end{lemma}
\begin{proof}
    Let $p_1$ be the probability that $E^1_{S}$ happens.
    By \Cref{lem:onemodified}, with probability at least $1 - 1/n^c$, each node is able to communicate with the nodes of at most one designated neighborhood in $\mathcal{S}$. Observe that, if we restrict to random bits assignments satisfying this condition, we get that:
    \begin{itemize}
        \item the failure probability increases at most to $2 / n^c$;
        \item the probability that $E^1_S$ holds is still at least $p_1/2$;
        \item the probability that $E^2_S$ holds, conditioned on $E^1_S$, is still a least $p_2/2$.
    \end{itemize}
    First, we restrict to these random bits. Then,
    for each $S \in \mathcal{S}$, we restrict the random bit assignment even further, by taking an assignment at random among the ones satisfies both $E^1_S$ and $E^2_S$. We get that the failure probability of each node is at most $\frac{2}{n^c} \cdot \frac{2}{p_1} \cdot \frac{2}{p_2}$, since a node is able to see at most one neighborhood $S$ where random bits have been restricted to satisfy $E^1_S$ and $E^2_S$, and such an assignment could have happened with probability at least $p_1/2 \cdot p_2/2$.

    By the Markov inequality, and by applying \Cref{lem:expectedruntime} with suitable constants, by running $\mathcal{A}$ on the designated node $v$ in $S$, $v$ does not communicate with any node outside $S$ with probability at least $1/2$. Hence, $p_1 = 1/2$. We thus get that the failure probability is bounded by $16 / (p_2 \cdot n^c)$.
\end{proof}

By setting $E^2_S$ as an event that always holds, we get the following corollary.
\begin{corollary}\label{cor:goodbitsrand}
    Let $\mathcal{A}$ be a randomized algorithm with node-averaged complexity $o(w^{1/(2^k -1)}/\log n )$ and failure probability at most $1 / n^c$, where $w = w(n)$ is in $O(n)$.
    Then, there exists a non-empty subset of possible random bits assignments satisfying the following.
    \begin{itemize}
        \item For each $S \in \mathcal{S}$, by running $\mathcal{A}$ on the designated node $v$ in $S$, $v$ does not communicate with any node outside $S$.
        \item The failure probability of $\mathcal{A}$, when run using such random bits assignments, is at most $16 / n^c$.
    \end{itemize}
\end{corollary}

We now prove a statement for deterministic algorithms that run in $T(n) = o(w^{1/(2^k -1)})$, where $w = w(n)$ is in $O(n)$.  For this purpose, we consider the set $\mathcal{S}$ of designated neighborhoods w.r.t.\ $g(n) = 1$ and some constant $\beta$.
\begin{lemma}\label{lem:s-independent}
    Let $\mathcal{A}$ be a deterministic algorithm with node-averaged complexity $T(n) = o(w^{1/(2^k -1)})$, where $w = w(n)$ is in $O(n)$. Let $S \in \mathcal{S}$ be a designated neighborhood. Let the ID space be $\{1,\ldots,n^q\}$, for some constant $q$ large enough. There are at least $n^{q-2}$ disjoint sets of IDs such that, for each of them, it is possible to assign the IDs to the nodes of $S$ such that, by running $\mathcal{A}$ on the designated node $v$ in $S$, $v$ does not communicate with any node outside $S$.
\end{lemma}
\begin{proof}
  From the ID space $\{1,\ldots,n^q\}$, we can construct $n^{q-1}$ disjoint sets of $n$ IDs, and for each set we can construct an arbitrary specific ID assignment that uses only IDs from that set. Let $X$ be the set of such ID assignments (observe that $|X| = n^{q-1}$). 
  
  Let $i$ be the layer of the designated node in $S$.
  Let $\mathcal{P}$ be the set containing, for each $j \in \{1,\ldots,i\}$, all the compress subpaths of layer $j$ induced by nodes in $S\setminus \{v\}$. We prove that we can process each subpath $P \in \mathcal{P}$, and discard at most $n$ elements among the remaining elements in $X$, to make sure that, for all the remaining ID assignments, there exists at least one node in $P$ running in at most half of the length of the subpath.
  Since there are at most $n$ elements in $\mathcal{P}$, we obtain that there exists a subset of $X$ of size $n^{q-1} - n^2 \ge n^{q-2}$ where $v$ cannot see outside $S$, proving the statement.

We process all subpaths $P \in \mathcal{P}$ in an arbitrary order, and each time we restrict the set $X$ to make sure that it satisfies the required condition. Let $j$ be the layer of $P$. Consider the node $v$ at the center of $P$, breaking ties arbitrarily. Since $P$ is of length at least $w^{2^{j-1}/(2^k -1)} / \beta$ for a fixed constant $\beta$, for $c$ small enough $v$ is at distance at least $c \cdot w^{2^{j-1}/(2^k -1)}$ from the endpoints of $P$.
Let $n'$ be the amount of ID assignments that are still in $X$, and assume that $n' \ge n$. We take $n$ arbitrary ID assignments from $X$, let $Y$ be this set. We apply \Cref{cor:expectedruntimedet} on $v$ and the ID assignments in $Y$, and we obtain that there exists an element $y \in Y$ where the runtime of $v$ is bounded by $c \cdot w^{2^{j-1}/(2^k -1)}$ for an arbitrarily small constant $c$. Now, we can replace $y$ with another element from $X$, apply the lemma again, and obtain the same result. This operation can be repeated for $n' - n +1$ times, in order to obtain $n' - n +1$ ID assignments in $X$ where the runtime of $v$ is bounded by $c \cdot w^{2^{j-1}/(2^k -1)}$. We discard all the other ID assignments. We obtain that we removed $n-1$ ID assignments from $X$, and that, in all the remaining ones, the runtime of $v$ is bounded by $c \cdot w^{2^{j-1}/(2^k -1)}$, that for $c$ small enough is less than half of the length of $P$.
\end{proof}

\begin{corollary}\label{cor:goodbitsdet}
    Let $\mathcal{A}$ be a deterministic algorithm with node-averaged complexity $o(w^{1/(2^k -1)})$, where $w = w(n)$ is in $O(n)$. Let the ID space be $\{1,\ldots,n^c\}$. There are at least $n^{c-3}$ disjoint sets of IDs such that, for each of them, 
    it is possible to assign the IDs to the nodes of each $S \in \mathcal{S}$, such that,
    by running $\mathcal{A}$ on the designated node $v$ in $S$, $v$ does not communicate with any node outside $S$.
\end{corollary}
\begin{proof}
    We split the ID space into $n$ disjoint subsets of size $n^{c-1}$. Note that $n$ is an upper bound on the elements of $\mathcal{S}$. To each element $S \in \mathcal{S}$ we assign a different subset.
    We apply \Cref{lem:s-independent} on each $S \in \mathcal{S}$, and we obtain that, for each $S \in \mathcal{S}$, there exist $n^{c-3}$ ID assignments for the whole graph, where the designated node $v$ of $S$ does not see outside $S$. Since $v$ does not see outside $S$, changing the IDs outside $S$ does not affect the runtime of $v$. Hence, we can combine the ID assignments of all the designated neighborhoods to produce $n^{c-3}$ ID assignments satisfying the requirements.
\end{proof}

\subsection{\boldmath The $f_{\Pi, k+1}$ Function Given $\mathcal{A}$}

\subparagraph{Function definition.}
We now show how to define the function $f_{\Pi,k+1}$ by either using a given algorithm $\mathcal{A}$ that solves $\Pi$ with randomized node-averaged complexity $o(w(n)^{1/(2^k-1)}/\log n)$, or by using a given algorithm $\mathcal{A}$ that solves $\Pi$ with deterministic node-averaged complexity $o(w(n)^{1/(2^k-1)})$, where $w(n) \le c \cdot n$ for a small enough constant $c > 0$ and $n$ large enough. Observe that, by setting $w(n) := c\cdot n$, we obtain exactly the main goal of this section. However, we prove a generic statement for all functions $w$, with the hope that such a result could be useful in the future when studying other complexity regimes.

Recall that the input of the function $f_{\Pi,k+1}$ is a compress path $P = (H,F_{\mathrm{incoming}},F_{\mathrm{outgoing}},\mathcal{L}_{\mathrm{incoming}})$, and it is required to produce an independent class for it. 
As a technicality, let us mention that, for each compress path $P$, the representative tree $r(P)$ is well defined, since we can define $f_{\Pi,k+1}$ while being tested by the testing procedure.

We will define the function $f_{\Pi,k+1}$, as a function of a parameter $n$ that is at least some large enough constant $n_0$.
For deterministic algorithms, we first set $n_0$ to be the smallest value satisfying that the running time of $\mathcal{A}$, on all instances of size $n\ge n_0$, is at most $t = w(n)^{1/(2^k - 1)} / \alpha$, for some constant $\alpha$ to be fixed later. Such a value exists by the assumption on the node-averaged complexity of $\mathcal{A}$. Then, we update $n_0 = \max\{n_0,n'\}$, where $n'$ is the value in \Cref{cor:expectedruntimedet}.
For randomized algorithms, we first set $n_0$ to be the smallest value satisfying that the running time of $\mathcal{A}$, on all instances of size $n\ge n_0$, is at most $t = w(n)^{1/(2^k - 1)} / (\alpha \log n)$, for some constant $\alpha$ to be fixed later. Such a value exists by the assumption on the node-averaged complexity of $\mathcal{A}$. Then, we update $n_0 = \max\{n_0,n'\}$, where $n'$ is the value in \Cref{cor:expectedruntime}.

Let $N_f(w)$ be the maximum size of any graph obtained from the function $\mathrm{pump}_k$, as a function of the parameter $w$, when applied to the graphs constructed in the testing procedure with the function $f$. Then, let $N(w)$ be the maximum value of $N_f(w)$, taken over all possible functions $f$, again as a function of the parameter $w$. In other words, $N(w)$ is the maximum size of any graph obtained from the function $\mathrm{pump}_k$, as a function of the parameter $w$, over all possible functions $f$.

In the following, we consider an arbitrary $n$ that is at least $n_0$, and let $w = w(n)$. By \Cref{lem:sizeRepr}, $N(w)=O(w)$. Thus, $N \le c \cdot n$ for a small enough constant $c > 0$. In the following, we assume $c$ to be small enough so that $n \ge 16|\Sigma_{\mathrm{out}}|\Delta N + 1$. 

We define the function $f_{\Pi,k+1}$ on the input $P = (H,F_{\mathrm{incoming}},F_{\mathrm{outgoing}},\mathcal{L}_{\mathrm{incoming}})$ as follows. 
First, construct the tree $T = \mathrm{add}(n,\mathrm{pump}(w,r(P)))$, and let $H'$ be the pumped path corresponding to the nodes of $H$ in $T$. Let $v$ be the designated node of $H'$, and let $e$ be the designated edge of $H'$. Let $i$ be the layer number of $v$. The goal is to run $\mathcal{A}$ on $T$ and use its output on $e$ to define the function $f_{\Pi,k+1}$. Observe that the output of $\mathcal{A}$, and its runtime, depend on the random bits (if $\mathcal{A}$ is randomized) and on the ID assignment assigned to the nodes of $T$. 
Let $g(n) = \log n$ if $\mathcal{A}$ is randomized, and $g(n) = 1$ otherwise. 
Let $B$ be the subset of random bits and ID assignments to the nodes of $T$ given by \Cref{cor:goodbitsrand} and \Cref{cor:goodbitsdet}.
We consider all assignments in $B$, and we take the output $o$ on $e$ that appears more often (breaking ties arbitrarily).
Observe that this output appears with probability at least $1/|\Sigma_{\mathrm{out}}|$.

Let $(H',F'_{\mathrm{incoming}},F'_{\mathrm{outgoing}},\mathcal{L}'_{\mathrm{incoming}})$ be the result of applying \Cref{lem:pump} on $P$ with parameter $w_i$. For each endpoint $u$ of $H'$, we can compute a label-set as follows. Let $u_1$ be the node, among the endpoints of the designated edge $e$, that is closest to $u$, and let $u_z = u$. Consider the subpath $u_1,\ldots,u_z$ of $H'$ induced by $u_1$, $u_z$, and the nodes in between.
We define the label-set of $e$ as $\{o\}$. Then, we compute the label-set of each edge $e_i = \{u_i,u_{i+1}\}$, for $i = 1, \ldots, z$, by applying  \Cref{def:computing-label-set} as if the current node $u_i$ were to be a single rake node with outgoing edge $e_i$, where $e_z$ is defined to be the edge outgoing from $u$. 
Observe that the label-sets $L_1$ and $L_2$ of the endpoints induce an independent class for $H'$, because, for any choice $(l_1,l_2) \in L_1 \times L_2$, by construction, we can pick a feasible assignment for $P$. 
By \Cref{lem:pump}, this independent class is valid also for $H$.

\subparagraph{Function correctness.}
We start by proving the correctness of the function created by using a deterministic algorithm.
\begin{lemma}\label{lem:summarydet}
    Assume that the following holds.
    \begin{itemize}
        \item $f_{\Pi,k+1}$ is a function that has been constructed by using trees of size $n$ and a deterministic algorithm with node-averaged complexity $o(w^{1/(2^k -1)})$, where $w = w(n)$ is in $O(n)$.
        \item $n \ge \Delta N + 1$.
        \item The ID space is $\{1,\ldots,n^c\}$, for some large enough constant $c$.
        \item $((T,v), L)$ is a pair obtained by the testing procedure by using the function $f_{\Pi,k+1}$.
    \end{itemize}
    Let $T'=\mathrm{pump}_k(w, T)$. Then
    \begin{itemize}
        \item $T'$ has at most $N$ nodes, where $N = \Theta(w(n))$, and
        \item there exists a set $\mathcal{X}$ containing $O(n^{c-4})$ disjoint sets of IDs
    \end{itemize}
    satisfying that, for each $X \in \mathcal{X}$, it is possible to assign IDs to the nodes of $T'$, such that:
    \begin{itemize}
        \item All the IDs are from $X$.
        \item There exists a set of nodes $D$ in $T'$ that, by running $\mathcal{A}$, they do not see outside $T'$.
        \item Any labeling of $T'$ that agrees with the outputs of the nodes of $D$ restricts the labels on the edge incident to the node of $T'$ that corresponds to $v$ to a subset of $L$.
    \end{itemize}
\end{lemma}
\begin{proof}
Recall that $v$ is a temporary node, and hence $T'$ contains at least one temporary node. Thus, we can construct $T'' = \mathrm{add}(n,T')$.

Let $\mathcal{L}$ be the partial labeling obtained by putting the most probable output (as previously defined) on each designated edge of $T''$. 

We process the designated neighborhoods in $T''$ one by one, and each time we assign the IDs in the neighborhood such that: the designated node does not see outside the designated neighborhood; the designated edge has the most probable output; the assigned IDs are not used in the rest of the graph.

Let $S$ be the current designated neighborhood to which we want to assign IDs.
Observe that, by \Cref{cor:goodbitsdet}, there are at least $n^{c-3} / |\Sigma_{\mathrm{out}}|$ disjoint ID assignments such that the designated edge gets the required output. Since there are at most $n$ IDs that have been already used to assign IDs to previously-handled designated neighborhoods, there are still at least $n^{c-3} / |\Sigma_{\mathrm{out}}| - n > 0$ possible choices that satisfy the requirements. 

We showed that there exists one ID assignment that satisfies the above requirements. Note that, by repeating this procedure with unused IDs, it is possible to construct $n^{c-4}/ |\Sigma_{\mathrm{out}}| = O(n^{c-4})$ such ID assignments.

Since the designated nodes do not see outside the designated neighborhoods, we get that there exist $O(n^{c-4})$ ID assignments (using disjoint sets of IDs) such that the algorithm outputs a labeling that agrees with $\mathcal{L}$.
By \Cref{lem:pump}, and by the definition of $f_{\Pi,k+1}$, we get that the labels allowed for the edge incident to the temporary node are a subset of $L$.
\end{proof}
\begin{corollary}
    Let $\mathcal{A}$ be a deterministic algorithm with node-averaged complexity $o(w^{1/(2^k -1)})$, where $w = w(n)$ is in $O(n)$. The function $f_{\Pi,k+1}$ constructed using $\mathcal{A}$ passes the testing procedure.
\end{corollary}
\begin{proof}
    Assume that the testing procedure fails. This means that an empty class on some node $v$ in some representative tree $T$ is obtained. A necessary condition for obtaining an empty class is that $T$ contains at least a compress path (otherwise, the function would not even be used), implying that $T$ contains at least two temporary nodes. Let $L_1,\ldots,L_d$ be the label-sets assigned to the incoming edges of $v$.
    We use \Cref{lem:summarydet} to construct $d$ pairs $((T'_1,v_1),L_1),\ldots,((T'_d,v_d),L_d)$ with disjoint ID assignments, and we merge them into a single tree $T''$ of at most $\Delta N\le n$ nodes by identifying nodes $v_1,\ldots,v_d$ into a single node $v'$. If $v$ has degree $d+1$, we add an additional neighbor to $v'$, and the number of nodes is now at most $\Delta N + 1 \le n$.
    Observe that, in $T''$, the following holds. Node $v'$ has the same degree as $v$, and $d$ incident edges are restricted to have the label-sets $L_1,\ldots,L_d$, if $T''$ were to be of size exactly $n$. 
    However, note that we can increase the number of nodes to be exactly $n$ by using $\mathrm{add}(n, T'')$. This is doable since $T''$ contains at least one temporary node different from $v'$ (since, as argued, $T$ contains at least two temporary nodes).
    By the definition of class, and by the assumption that $v$ has an empty class, there is no valid labeling for the edges of $v$, which is a contradiction on the correctness of $\mathcal{A}$.
\end{proof}

We now prove the correctness of the function created by using a randomized algorithm.
\begin{lemma}\label{lem:summaryrand}
    Assume that the following holds.
    \begin{itemize}
        \item $f_{\Pi,k+1}$ is a function that has been constructed by using trees of size $n$ and a randomized algorithm with node-averaged complexity $o(w^{1/(2^k -1)} / \log n)$, where $w = w(n)$ is in $O(n)$.
        \item $n \ge 16|\Sigma_{\mathrm{out}}|N + 1$.
        \item $C$ is a class obtained by the testing procedure on some representative tree $T$ that contains at least one temporary node $v$, by using the function $f_{\Pi,k+1}$.
    \end{itemize}
    Let $T'=\mathrm{pump}_k(w, T)$ and $T'' = \mathrm{add}(n,T')$. Then, 
    \begin{itemize}
        \item $T'$ has at most $N$ nodes, where $N = \Theta(w(n))$.
    \end{itemize}
    Moreover, by running $\mathcal{A}$ on $T''$, with non-zero probability,
    \begin{itemize}
        \item all nodes in $T'$ do not fail, and
        \item there exists a node $u$ in $T'$ that has an output compatible with the class $C$ (this, in particular, implies that $C$ is not empty).
    \end{itemize}
\end{lemma}
\begin{proof}
Let $\mathcal{L}$ be the partial labeling obtained by putting the most probable output (as previously defined) on each designated edge of $T''$. 

We run $\mathcal{A}$ on $T''$ by using the random bits given by \Cref{cor:goodbitsrand}, restricted to the case in which, for each $S \in \mathcal{S}$, the designated edge gets the most probable output. That is, we apply \Cref{lem:goodbitsrand} by setting the event $E^2_S$ to the event in which the designated edge gets the most probable output when restricted to the random bit assignments given by \Cref{cor:goodbitsrand}. Note that $p_2 = 1 / |\Sigma_{\mathrm{out}}| $. We get that the failure probability of $\mathcal{A}$ is at most $16 / (1 / |\Sigma_{\mathrm{out}}| \cdot n^c) = 16 |\Sigma_{\mathrm{out}}| / n^c$. By a union bound, the probability that all nodes in $T'$ succeed is at least $1 - N \cdot 16 |\Sigma_{\mathrm{out}}| / n^c \ge 1 - N \cdot 16 |\Sigma_{\mathrm{out}}| / n \ge 1 - N \cdot 16 |\Sigma_{\mathrm{out}}| / (N \cdot 16 |\Sigma_{\mathrm{out}}| + 1) > 0$.
Hence, by the probabilistic method, there exists an assignment of output labels that is valid for all the nodes of $T'$ and that agrees with $\mathcal{L}$. By \Cref{lem:pump}, and by the definition of $f_{\Pi,k+1}$, we get that the output on the edges incident to $u$ is compatible with $C$.
\end{proof}
\begin{corollary}
    Let $\mathcal{A}$ be a randomized algorithm with node-averaged complexity $o(w^{1/(2^k -1)} / \log n)$, where $w = w(n)$ is in $O(n)$. The function $f_{\Pi,k+1}$ constructed using $\mathcal{A}$ passes the testing procedure.
\end{corollary}
\begin{proof}
    Assume that the testing procedure fails. This means that an empty class on some representative tree $T$ is obtained. Note that there must be at least one compress path, otherwise the function is not even used. This implies that $T$ contains at least one temporary node. By \Cref{lem:summaryrand}, any class obtained by the testing procedure is non-empty, which is a contradiction.
\end{proof}

\section{Open Questions}\label{sec:open}
We conclude with some open questions. We showed that all problems that have $O(\log n)$ worst-case complexity can be solved with $O(\log^* n)$ node-averaged complexity. We leave open to determine for which problems this can be improved. 
\begin{oq}
    For which LCLs with $O(\log n)$ worst-case complexity can we obtain $o(\log^* n)$ node-averaged complexity?
\end{oq}

We showed that all problems that have worst-case complexity $\Theta(n^{1/k})$ must have node-averaged complexity $\Omega(n^{1 / (2^k - 1)} / \log n)$. We conjecture that the $\log n$ factor is an artefact of our proof technique and that it should not be there.
\begin{oq}
    Can we prove a lower bound of $\Omega(n^{1 / (2^k - 1)} )$ rounds for the node-averaged complexity of all problems with worst-case complexity $\Theta(n^{1/k})$?
\end{oq}

We showed that for some LCL problems that have worst-case complexity $\Theta(n^{1/k})$, we can provide an algorithm with node-averaged complexity $O(n^{1 / (2^k - 1)})$. It is not clear if this can be done for all problems with  worst-case complexity $\Theta(n^{1/k})$.
\begin{oq}
    Can we prove an upper bound of $O(n^{1 / (2^k - 1)} )$ rounds for the node-averaged complexity of all problems that have worst-case complexity $\Theta(n^{1/k})$?
\end{oq}
\begin{oq}
    Can we prove a lower bound of $\Omega(n^{1 / k} )$ rounds for the node-averaged complexity of some problems that have worst-case complexity $\Theta(n^{1/k})$?
\end{oq}

\bibliographystyle{plainurl}
\bibliography{biblio}

\appendix

\section{Additional Related Work About LCLs}\label{apx:related-lcls}
LCLs were introduced by Naor and Stockmeyer \cite{NaorStockmeyer95}, but locally checkable problems were studied in the distributed setting even before \cite{AfekKY97}.
Since then, LCL problems have been studied a lot, and we now known, for different possible topologies of graphs, what kind of worst-case complexities are possible.

\subparagraph{Paths and cycles.}
We know that in paths and cycles there are only three possible complexities: $O(1)$, $\Theta(\log^* n)$, $\Theta(n)$ \cite{NaorStockmeyer95,CKP19exponential}. Moreover, we know that randomness does not help to solve problems faster.

\subparagraph{Trees.}
The possible deterministic complexities in trees are $O(1)$, $\Theta(\log^* n)$, $\Theta(\log n)$, and $\Theta(n^{1/k})$ for all integer $k\ge 1$; randomness either helps exponentially or not at all, and only for problems with deterministic complexity $\Theta(\log n)$, that hence have randomized complexity either $\Theta(\log n)$ or $\Theta(\log \log n)$ \cite{brandt21trees,BHOS19HomogeneousLCL,CKP19exponential,CP19timeHierarchy,BBOS18almostGlobal,chang20}.

\subparagraph{General graphs.}
In general graphs, some \emph{complexity gaps} that hold in the case of trees still hold, but now there are also many dense areas; while in trees randomness either helps exponentially or not at all, there are cases on general graphs where randomness helps only polynomially \cite{CP19timeHierarchy,CKP19exponential,FischerGhaffari17LLL,RG20NetDecomposition,BHKLOS18lclComplexity,BBOS20paddedLCL}.

\subparagraph{LCLs in other settings.}
On $d$-dimensional balanced toroidal grids, it is known that  the only possible complexities are $O(1)$, $O(\log^* n)$, and $\Theta(n^{1/d})$ (even by allowing randomness) \cite{lcls_on_grids}.

For problems that, in the black-white formalism (as defined in \Cref{sec:definitions}), can be expressed by using at most two labels, on regular trees the only possible deterministic complexities are $O(1)$, $\Theta(\log n)$ and $\Theta(n)$, implying that any LCL with complexity $\Theta(\log^* n)$ needs to be expressed with at least $3$ labels \cite{binary_lcls}.

It is known that any LCL on trees that can be solved in $T$ rounds in the LOCAL model of distributed computing can be solved in $O(T)$ rounds in  the CONGEST model; it is also known that this does not hold if we consider general graphs \cite{bcmos21}.

On \emph{rooted} trees, it is known that all possible complexities are $O(1)$, $\Theta(\log^* n)$, $\Theta(\log n)$, $\Theta(n^{1/k})$ for all integer $k\ge 1$, but perhaps surprisingly, randomness never helps in solving problems faster  \cite{LCLs_in_rooted_trees}.

LCLs have been studied also in other models of interest such as MPC \cite{B0FLMOU22}.

\subparagraph{Decidability.}
LCLs have been studied not only from a point of view of complexity theory. Researchers tried also to address the following questions.
Given a specific LCL, can we decide its time complexity with a centralized algorithm? Is it possible to automate the design of algorithms for solving LCLs?

In general, the complexity of an LCL is not decidable: even on unlabeled non-toroidal grid graphs it is undecidable whether the complexity of an LCL is $O(1)$ \cite{NaorStockmeyer95}. However, there are some positive results in more restricted but still interesting settings. 

On paths and cycles, it is possible to determine what is the time complexity of a given problem, but it becomes EXPTIME-hard if some input is provided to the nodes \cite{NaorStockmeyer95, lcls_on_grids, lcls_on_paths_and_cycles,balliu19lcl-decidability,chang20}.

In unlabeled toroidal grids, it is decidable whether the complexity of an LCL is $O(1)$, but it is undecidable whether its complexity is $\Theta(\log^* n)$ or $\Theta(n)$ \cite{lcls_on_grids}. 

On trees, given an LCL, it is possible to decide on which side of the gap $\omega(\log n)$ \--- $n^{o(1)}$ its complexity lies. Moreover, it is decidable if an LCL has complexity $\Theta(n^{1/k})$ for some $k$, and it is also possible to determine the exact value of $k$ \cite{CP19timeHierarchy,chang20}, but the algorithm is very far from being practical. However, if we restrict to regular trees with no inputs, then there is a practical polynomial-time algorithm (in the size of the description of the LCL problem) for deciding if an LCL has complexity $\Theta(n^{1/k})$ and determining the exact value of $k$ \cite{B0COSS22_LCLregularTrees}. Similarly, for rooted trees, there are efficient algorithms for determining the optimal asymptotic complexity of a given LCL \cite{LCLs_in_rooted_trees, B0COSS22_LCLregularTrees}. On unrooted regular trees with no input, if we restrict to problems that can be expressed by using at most two labels in the black-white formalism, the deterministic complexity of a problem is decidable \cite{binary_lcls}.

It is still an open question whether, on trees, we can obtain decidability for the lower complexities, e.g., it is unknown whether we can decide if an LCL on trees can be solved in $O(1)$ rounds or if it requires $\Omega(\log^* n)$ rounds, and whether it can be solved in $O(\log^* n)$ rounds or if it requires $\Omega(\log n)$ with deterministic algorithms.

\section{\boldmath An Algorithm For Solving All LCLs in $O(D)$ Rounds}\label{sec:diamsolver}
In order to give more intuition about the generic method that can be used to solve LCL problems, we repropose a simplified setting already presented in \cite{bcmos21}, where we restrict a bit the class of problems that we consider, as follows. We are given a tree of constant maximum degree, where to each edge is assigned an input label that comes from a finite set, and the goal is to assign an output label to each edge, also from a finite set, in such a way that, for each node, the multiset of incident input-output pairs of labels is contained in a list of given allowed configurations (this list is the same for all nodes). In other words, in this simplified setting, an LCL problem is described by providing a list of allowed configurations $C$, that are multisets of input-output pairs of labels. We now describe a procedure, already presented in \cite{bcmos21}, that solves any problem of this form in $O(D)$ rounds, where $D$ is the diameter of the tree. 
The procedure works as follows:
\begin{enumerate}
    \item Each leaf node $v$, as a function of the input label $\ell_i$ of its incident edge $e$ and the list of allowed configurations $C$, computes the set $S_v$ of output labels $\ell_o$ satisfying that the multiset $\{(\ell_i,\ell_o)\}$ is in $C$, that is, $v$ computes the set $S_v$ of output labels $\ell_o$ that, if assigned to $e$, would make $v$ happy. 
    \item Each leaf $v$ sends the set $S_v$ to its neighbor.
    \item Leaves are removed from the tree. Let $L$ be the set of nodes that became leaves after the removal operation. Each node $v \in L$ proceeds as follows. Let $u_1, \ldots, u_d$ be the neighbors of $v$ that got removed in previous steps, let $\ell_{i_1}, \ldots, \ell_{i_d}$ be the input labels on the edges connecting node $v$ to the nodes $u_1, \ldots, u_d$, and let $\ell_i$ be the input label connecting $v$ to its neighbor that is still present. Node $v$ computes the set $S_v$ of output labels $\ell_o$ satisfying that there exists a choice $(\ell_{o_1},\ldots,\ell_{o_d}) \in S_{u_1} \times \ldots \times S_{u_d}$ such that the multiset $\{(\ell_i,\ell_o),(\ell_{i_1},\ell_{o_1}),\ldots,(\ell_{i_d},\ell_{o_d})\}$ is in $C$. In other words, node $v$ computes the set $S_v$ of output labels for its remaining edge satisfying that, even if the output from $S_v$ is chosen adversarially, there is still a choice that $v$ can make, over the sets received from its removed neighbors, that would give an assignment of labels that makes $v$ happy. Then, node $v$ sends $S_v$ to its remaining neighbor.
    \item Repeat step 3 until the graph is empty.
    \item The last removed node chooses a label from each received set, in such a way that the resulting multiset of input-output pairs of labels is a configuration in $C$. (It may happen that two neighboring nodes get removed last, at the same time, but this case can be handled in a similar way.)
    \item Removed nodes are put back in reverse order. Observe that, when a node $v$ is put back, the output label $\ell_o$ of the edge connecting $v$ to its only neighbor that is currently present has already been assigned, and that this output label $\ell_o$ is in $S_v$. Node $v$ picks a label from the sets assigned to the edges connecting $v$ to its other neighbors (that is, neighbors that are going to be put back in the next step) in a way that is compatible with $\ell_o$.
\end{enumerate}
It is not difficult to see that this algorithm computes a correct solution for the problem, assuming that the computed sets never become empty, that is, as long as, for each $v \in V$, $S_v \neq \emptyset$. Also, it is clear that the time complexity of this algorithm is $O(D)$. It has been shown in \cite{bcmos21} that, if it happens that some node $v$ gets $S_v = \emptyset$, then it means that the problem is unsolvable. Moreover, \cite{bcmos21} showed that by following this algorithm, we get a generic way to solve all problems in $O(D)$ rounds that is actually bandwidth efficient, since the sets that are sent at each step have constant size. 

\subparagraph{Improving the round complexity.}
The previously described procedure shows how to solve all (solvable) problems in $O(D)$ rounds, which is nice for the CONGEST model, but is a trivial result in the LOCAL model. Moreover, the diameter $D$ could be as high as $\Omega(n)$, and for some problems this complexity may be suboptimal. Hence, in some cases we would like to obtain a faster algorithm, and this is what has been shown in \cite{CP19timeHierarchy, chang20}.
The idea of \cite{CP19timeHierarchy, chang20} is that, if, instead of removing only nodes of degree $1$ at each step, we also sometimes remove nodes of degree $2$ as well, then it takes less steps to obtain an empty graph. One of the issues to handle when running this modified procedure is that we now have to find a way to assign sets to nodes of degree $2$, which may form long paths, and this is the part that turns out to be quite challenging. 

\subparagraph{Restricting to a specific kind of LCLs.}
In order to prove our results, we use and extend ideas presented in \cite{CP19timeHierarchy, chang20, bcmos21}. In particular, in \cite{bcmos21} it has been shown that, in bounded-degree trees, if an LCL problem has time complexity $T$ in the LOCAL model, then it has time complexity $O(T)$ in the CONGEST model, both for deterministic and randomized worst-case complexities. In order to show this result, the authors of \cite{bcmos21} not only extended the results of \cite{CP19timeHierarchy, chang20}, but also provided a more accessible version of some of the proofs of \cite{CP19timeHierarchy, chang20}. The reason why these proofs are more accessible is that \cite{bcmos21} does not show results for standard LCLs, but only for a restriction of those, that are LCLs that can be expressed in a formalism called black-white. Importantly, the authors also showed that, if we restrict to trees, for any standard LCL (LCLs as they are usually defined in the literature), we can define an LCL in the black-white formalism that has the same complexity of the original one, up to an \emph{additive} constant. In other words, considering LCLs in the black-white formalism is not really a restriction if the graph class we are working on is trees (in general graphs this turns out to not be the case). In this paper, we follow a similar route of \cite{bcmos21}: in order to keep our proofs more accessible we prove our statements for LCLs expressed in the black-white formalism, but we also prove that, for any standard LCL, we can define an LCL in the black-white formalism that has the same node-averaged complexity of the original one, up to a \emph{multiplicative} constant factor, implying that our results hold for all LCLs as well. This equivalence is shown in \Cref{lem:node-edge-enough}.

\section{Different Ways to Define LCLs}\label{sec:different-lcls}
In this section we prove an equivalence, for node-averaged complexity, between different definitions of LCLs. We start by providing the standard definition of LCLs, that in the following we will call standard LCLs.

\subparagraph{Locally Checkable Labeling problems.}\label{def:generic-lcls}
An LCL problem $\Pi$ is defined as a tuple $(\Sigma_{\mathrm{in}},\Sigma_{\mathrm{out}},C, r)$, where:
\begin{itemize}
    \item $\Sigma_{\mathrm{in}}$ and $\Sigma_{\mathrm{out}}$ are finite sets of labels that represent the possible input and output labels.
    
    \item The parameter $r\ge 1$ is an integer, and it represents the so-called \emph{checkability radius} of the LCL problem $\Pi$.

    \item $C$ is a finite set of labeled graphs that represent allowed neighborhoods, and more precisely $C$ is a finite set of pairs $(H, v)$, where:
    \begin{itemize}
        \item $H=(V_H, E_H)$ is a graph, and $v\in V_H$.
        \item The eccentricity of $v$ in $H$ is at most $r$.
        \item To each pair $(v,e)\in V_H\times E_H$ is assigned a label $\ell_{\mathrm{in}}\in \Sigma_{\mathrm{in}}$ and a label $\ell_{\mathrm{out}}\in \Sigma_{\mathrm{out}}$.
    \end{itemize}
\end{itemize}
Solving an LCL problem $\Pi$ on a graph $G=(V,E)$ means that:
\begin{itemize}
   \item To each node-edge pair $(v,e) \in V \times E$ is assigned a label $i(e) \in \Sigma_{\mathrm{in}}$.
    \item The task is to assign a label $o(e) \in \Sigma_{\mathrm{out}}$ to each node-edge pair $(v,e) \in V \times E$ such that, for each node $v \in V$, it holds that the (labeled) graph $N^r_v$ induced by $v$'s radius-$r$ neighborhood is in $C$, that is, $N^r_v$ is isomorphic to a (labeled) graph contained in $C$.
\end{itemize}

\subparagraph{Node-edge formalism.}
Using the black-white formalism (as defined in \Cref{sec:definitions}) requires the graph to be properly $2$-colored. We now describe a formalism, called node-edge, that does not have this requirement. We will soon show an equivalence between the node-edge formalism and the black-white formalism.
A problem $\Pi$ described in the node-edge formalism is a tuple $(\Sigma_{\mathrm{in}},\Sigma_{\mathrm{out}},C_N,C_E)$, where:
\begin{itemize}
    \item $\Sigma_{\mathrm{in}}$ and $\Sigma_{\mathrm{out}}$ are finite sets of labels.
    \item $C_N$ and $C_E$ are both multisets of pairs, where each pair $(\ell_{\mathrm{in}},\ell_{\mathrm{out}})$ is in $\Sigma_{\mathrm{in}} \times \Sigma_{\mathrm{out}}$. The multisets in $C_E$ have size exactly $2$.
\end{itemize}
Solving a problem $\Pi$ on a graph $G$ means that:
\begin{itemize}
    \item To each node-edge pair $(v,e) \in V \times E$ is assigned a label $i(e) \in \Sigma_{\mathrm{in}}$.
    \item The task is to assign a label $o(e) \in \Sigma_{\mathrm{out}}$ to each node-edge pair $(v,e) \in V \times E$ such that, for each node $v \in V$ (resp.\ each edge $e \in E$) it holds that the multiset of incident input-output pairs is in $C_V$ (resp.\ in $C_E$).
\end{itemize}

\subparagraph{Equivalence between black-white and node-edge formalisms.}
Observe that, if we focus only on what is a problem, and we forget about what it means to solve a problem, it is clear that the node-edge formalism is a subcase of the more general black-white formalism. In fact, it is easy to see that a problem $\Pi$ in the node-edge formalism implicitly defines also a problem $\Pi'$ in the black-white formalism, while a problem $\Pi'$ in the black-white formalism implicitly defines also a problem $\Pi$ in the node-edge formalism if it satisfies that $C_B$ contains only multisets of size exactly~$2$.

Moreover, it has been show in \cite{bcmos21} that also the time complexities of $\Pi$ and $\Pi'$ are related. In fact, one can prove that, given an algorithm $A$ for $\Pi$, we can define a new algorithm $A'$ for $\Pi'$ with the same asymptotic worst-case round complexity, and vice versa. 

We now describe the ideas that show that, given an algorithm $A'$ for $\Pi'$, we can use it to solve $\Pi$ (the other direction can be shown in a similar way). 
The idea is to simulate the algorithm for $\Pi'$ on a virtual graph $G'$, defined as a function of the real graph $G = (V,E)$ as follows.
Let $G' = (W \cup B, E')$, where $W = V$, $B = E$, and we connect $v \in V$ with $e \in B$ if $v \in e$. In other words, original nodes are white, and we add a black node in the middle of each original edge. Observe that, by solving $\Pi'$ on $G'$, and then mapping the output obtained for the edges of $G'$ to the node-edge pairs of $G$, we also solve $\Pi$ on $G$. Also, observe that simulating the execution of a $T$ round algorithm for $G'$ only costs $T/2 + O(1)$ rounds on $G$.

\subparagraph{Equivalence between standard LCLs and LCLs in the node-edge formalism.}
In \cite{bcmos21} it has been shown that, given a standard LCL $\Pi$ on trees, it is possible to define a node-edge checkable problem $\Pi'$ satisfying that, given a solution for $\Pi$, it is possible to spend $O(1)$ rounds to solve $\Pi'$, and vice versa. 
\begin{lemma}[\cite{bcmos21}]\label{lem:lclequivalence}
    For any LCL problem $\Pi$ on trees with checkability radius $r = O(1)$ we can define a node-edge checkable problem $\Pi'$ that satisfies the following:
    \begin{itemize}
        \item There exists an $O(1)$-rounds algorithm that, given in input a solution for $\Pi'$, outputs a solution for $\Pi$.
        \item There exists an $O(1)$-rounds algorithm that, given in input a solution for $\Pi$, outputs a solution for  $\Pi'$.
    \end{itemize}
\end{lemma}
By combining this lemma with the equivalence that we discussed before, we obtain that by studying the worst-case complexity of LCL problems on trees in the black-white formalism is not a restriction, and hence we get results that hold for LCLs in the more general form. \Cref{fig:lcls-MM} shows an example of an LCL problem defined in different formalisms. In the following we show a similar statement for the node-averaged complexity of LCLs.

\begin{figure}
	\centering
	\includegraphics[width=0.75\textwidth]{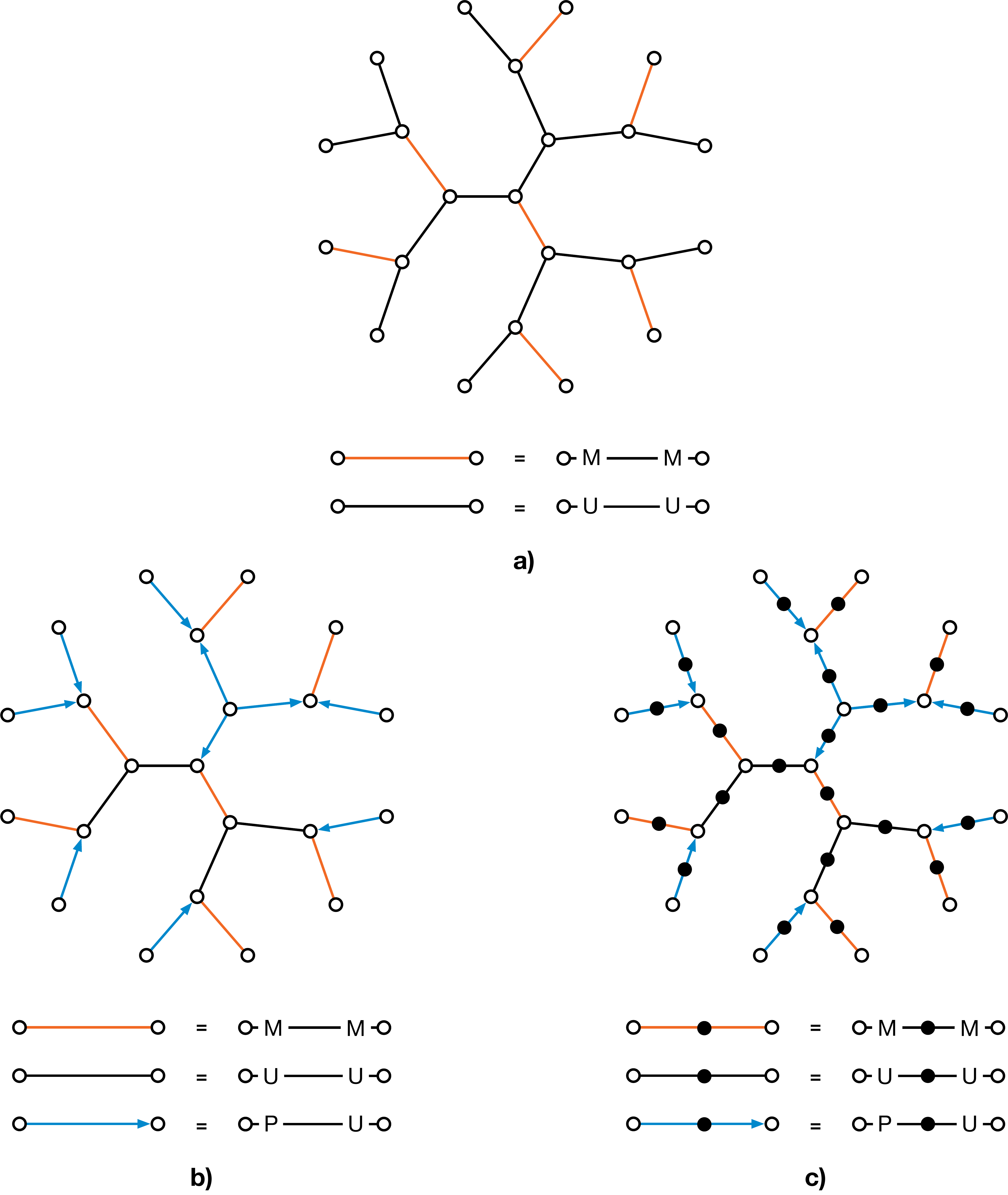}
	\caption{a) In the maximal matching problem, a node is either matched or all its neighbors are matched. This can be described as an LCL by letting each edge be either labeled $\{M,M\}$ or $\{U,U\}$, and then listing all possible valid radius-$2$ neighborhoods. b) The maximal matching problem encoded in the node-edge formalism. One way is to require unmatched nodes to orient all their edges outgoing, and then require all nodes with at least one incoming edge to be matched. These constraints can be given in the node-edge formalism. c) The $2$-colored graph in which the maximal matching problem has the same complexity as in the node-edge setting.}
	\label{fig:lcls-MM}
\end{figure}

\subparagraph{Equivalence between standard LCLs and black-white formalism for node-averaged complexity.}
In order to prove that we can restrict to the black-white formalism also for the case of node-averaged complexity, we first prove an equivalence between standard LCLs and the node-edge formalism, and then we prove an equivalence between the node-edge formalism and the black-white formalism. 

In the case of worst case complexity, if we have an algorithm $A_1$ that produce a result in $T_1$ rounds, and we have a different algorithm $A_2$ that takes the output of $A_1$ in input and produces a result in $T_2$ rounds, then we can execute $A_2$ after $A_1$ and obtain a running time of $T_1 + T_2$. For node-averaged complexity, unfortunately, the same does not hold. Nevertheless we can achieve something similar.
\begin{lemma}\label{lem:concat}
    Suppose that we have an algorithm $A_1$ with node-averaged complexity $T_1$, and an algorithm $A_2$ with worst-case complexity $T_2$. Consider the algorithm $A_3$ obtained by combining $A_1$ and $A_2$, where nodes execute $A_2$ only after all nodes in their $T_2$ radius neighborhood terminated the execution of $A_1$. The node-averaged complexity of $A_3$ is $O(T_1 \cdot \Delta^{T_2})$.
\end{lemma}
\begin{proof}
We show a charging scheme that satisfies that the sum of the charges is an upper bound on the sum of the running times of the nodes.
For each node $v$ of the graph, let $t(v)$ be the last node in the $T_2$-radius neighborhood of $v$ that terminates the execution of $A_1$, break ties arbitrarily. We charge the running time of $v$ in $A_3$ to $t(v)$. Observe that each node is charged $O(\Delta^{T_2})$ times, and that the sum of the running times in $A_1$ of the nodes charged at least once is $O(n \cdot T_1)$. Hence, the sum of the charges, and the running times in $A_3$, is $O(n \cdot T_1 \cdot \Delta^{T_2})$, implying that the node-averaged complexity of $A_3$ is $O(T_1 \cdot \Delta^{T_2})$, as required. 
\end{proof}
We are now ready to prove an equivalence, for node-averaged complexity, between standard LCLs and LCLs in the node-edge formalism.
\begin{lemma}\label{lem:avg-standard-ne-eq}
    For any LCL problem $\Pi$ on trees with checkability radius $r$ and node-averaged complexity $T$ we can define a node-edge
checkable problem $\Pi'$ with node-averaged complexity $\Theta(T)$.
\end{lemma}
\begin{proof}
    We apply \Cref{lem:concat} where $T_1$ is the node-averaged complexity of the given LCL problem $\Pi$, and $T_2$ is the worst-case time complexity required to convert a solution for $\Pi$ into a solution of its equivalent node-edge checkable variant. By \Cref{lem:lclequivalence} we have that $T_2 = O(1)$. Since $\Delta = O(1)$, then the claim follows.
\end{proof}

We now prove an equivalence, for node-averaged complexity, between LCLs in the node-edge formalism and LCLs in the black-white formalism.
\begin{lemma}\label{lem:bw-ne-eq}
    For any node-edge checkable LCL problem $\Pi$ on trees with node-averaged complexity $T$ we can define an LCL $\Pi'$ in the black-white formalism with node-averaged complexity $\Theta(T)$.
\end{lemma}
\begin{proof}
    Let $\Pi = (\Sigma_{\mathrm{in}},\Sigma_{\mathrm{out}},C_N,C_E)$.
    We define $\Pi' = \Pi$, that is, the input and output labels are exactly the same, the white constraint of $\Pi'$ is $C_N$, and the black constraint of $\Pi'$ is $C_E$.
    Observe that, while $\Pi$ and $\Pi'$ are syntactically the same, an algorithm for $\Pi$ is designed for working on a tree with no $2$-coloring given, and it needs to assign an output label to each node-edge pair, while an algorithm for $\Pi'$ is designed for working on a tree that is properly $2$-colored and where all black nodes have degree $2$, and it needs to assign an output label to each edge.
    
    We first show that, given an algorithm $A$ for solving $\Pi$ with node-averaged complexity $T$, we can design an algorithm $A'$ for solving $\Pi'$ with node-averaged complexity $O(T)$. The algorithm $A'$ works as follows. White nodes simulate the execution of $A$, while black nodes relay messages that are exchanged between white nodes. This algorithm clearly solves $\Pi'$. Each white node spends exactly twice the running time of $A$, since each round of $A$ is simulated with $2$ rounds of $A'$ (it takes $2$ rounds for white nodes to exchange messages). Observe that black nodes, in order to know their output, need to wait that the two incident white neighbors terminate. Hence, we can charge their running time to their slowest neighbor. Since each node is charged at most $\Delta = O(1)$ times, then we obtain that $A'$ has node-averaged complexity $O(T)$.

    We now show that, given an algorithm $A'$ for solving $\Pi'$ with node-averaged complexity $T$, we can design an algorithm $A$ for solving $\Pi$ with node-averaged complexity $O(T)$. The algorithm $A$ is defined as follows. Each edge of the tree is assigned to one incident node, arbitrarily. In the algorithm $A$, the nodes pretend to be white nodes, and simulate the execution of $A'$. Moreover, the nodes simulate, for each incident edge assigned to them, the execution of a black node, that is connected to the two nodes incident to the edge. This algorithm clearly solves $\Pi$, and its running time can be bounded as follows. The running time of each node is charged to its longest running simulated (black or white) node. Observe that each simulated node is charged at most once, and that the sum of the running times of the simulated nodes that have been charged at least once is an upper bound on the sum of the running times in $A$. This amount is at most $n \cdot T$, where $n$ is the size of the simulated graph. Since the number of nodes in the real graph and in the simulated one is asymptotically the same, then we obtain that the node-averaged complexity of $A$ is at most $O(T)$.
\end{proof}

By combining \Cref{lem:avg-standard-ne-eq} and \Cref{lem:bw-ne-eq} we obtain that focusing on LCLs in the black-white formalism is not a restriction, and that our results apply to standard LCLs as well. In particular, we obtain the following.
\begin{lemma}\label{lem:node-edge-enough}
    For any LCL problem $\Pi$ on trees with checkability radius $r$ and node-averaged complexity $T$ we can define an LCL $\Pi'$ in the black-white formalism with node-averaged complexity $\Theta(T)$.
\end{lemma}

\end{document}